\documentclass[a4paper]{article}
\usepackage{amsmath,amsthm,amssymb,graphicx}
\usepackage[margin=3cm]{geometry}
\usepackage{bbm}
\usepackage{xparse}
\usepackage{color}
\usepackage{enumitem}
\usepackage[square]{natbib}
\setcitestyle{numbers}

\usepackage{epstopdf}
\usepackage{chngcntr}
\counterwithin{figure}{section}

\theoremstyle{plain}
\theoremstyle{definition}
\numberwithin{equation}{section}

\newcommand{\compresslist}{ 
\setlength{\itemsep}{0.5pt}
\setlength{\parskip}{0pt}
\setlength{\parsep}{0pt}
}

\theoremstyle{plain}
\theoremstyle{definition}
\numberwithin{equation}{section}

\newtheorem{definition}{Definition}[section]
\newtheorem{theorem}[definition]{Theorem}
\newtheorem{assumption}[definition]{Assumption}
\newtheorem{example}[definition]{Example}
\newtheorem{remark}[definition]{Remark}

\newtheorem*{claim*}{Claim}
\newtheorem*{notation*}{Notation}
\newtheorem{lemma}[definition]{Lemma}
\newtheorem{corollary}[definition]{Corollary}
\newtheorem{proposition}[definition]{Proposition}

\newcommand{\proj}{\Finv}

\newcommand{\primalfs}{\underline{\mathbf{P}}_{\info}}
\newcommand{\primalls}{\widetilde{P}_{\vec{\mu},\set}}
\newcommand{\primalx}{P_{\XX,\PP,\set}}
\newcommand{\primalxa}{\widetilde{P}_{\XX,\PP,\set}}
\newcommand{\primalxs}{\underline{P}_{\XX,\PP,\set}}
\newcommand{\primalxsa}{\widetilde{\underline{P}}_{\XX,\PP,\set}}
\newcommand{\MSUPXs}{\underline{\MM}_{\XX,\PP,\set}}
\newcommand{\MSUPXsa}{\underline{\MM}^\eta_{\XX,\PP,\set}}

\newcommand{\MSUPs}{\underline{\MM}}

\newcommand{\argmin}{\arg\!\min}

\newcommand{\p}{\pi}

\newcommand{\E}{\mathbb{E}}
\newcommand{\D}{\mathbb{D}}
\newcommand{\F}{\mathbb{F}}

\newcommand{\M}{\mathbb{M}}

\newcommand{\R}{\mathbb{R}}

\renewcommand{\S}{\mathbb{S}}
\newcommand{\Q}{\mathbb{Q}}
\newcommand{\Z}{\mathbb{Z}}
\newcommand{\N}{\mathbb{N}}
\renewcommand{\P}{\mathbb{P}}
\newcommand{\T}{\mathbb{T}}

\newcommand{\GG}{\mathcal{G}}
\newcommand{\DD}{\mathcal{D}}
\newcommand{\CC}{\mathcal{C}}

\newcommand{\FF}{\mathcal{F}}
\newcommand{\PP}{\mathcal{P}}
\newcommand{\XX}{\mathcal{X}}
\newcommand{\YY}{\mathcal{Y}}
\newcommand{\II}{\mathcal{I}}

\newcommand{\MM}{\mathcal{M}}
\newcommand{\LL}{\mathcal{L}}
\newcommand{\BB}{\mathcal{B}}

\renewcommand{\AA}{\mathcal{A}}
\newcommand{\ZZ}{\mathcal{Z}}

\newcommand{\Fc}{\check{F}}
\newcommand{\Fcs}{\Fc^{(N)}}

\newcommand{\Fh}{\hat{F}}
\newcommand{\Fhs}{\Fh^{(N)}}
\newcommand{\Tt}{\tilde{\tau}}
\newcommand{\Th}{\hat{\tau}}

\newcommand{\dts}{\hat{\gamma}}
\newcommand{\ts}{\gamma}

\newcommand{\set}{\mathfrak{P}}
\newcommand{\info}{\II}

\newcommand{\vmu}{\vec{\mu}}

\newcommand{\MMUS}{\MM_{\vmu,\set}}
\NewDocumentCommand{\MMUE}{O{\vmu}O{\set}m}{\MM_{{#1},{#2},{#3}}}
\NewDocumentCommand{\MMUEs}{O{\vmu}O{\set}m}{\underline{\MM}_{{#1},{#2},{#3}}}

\newcommand{\MSUP}{\MM}
\newcommand{\MSUPI}{\MM_{\info}}
\newcommand{\MSUPX}{\MM_{\XX,\PP,\set}}
\newcommand{\MSUPXa}{{\MM}^\eta_{\XX,\PP,\set}}

\newcommand{\MNS}{\hat{\M}^{(N)}_{\info}}

\newcommand{\PT}{\hat{\Pi}^{(N)}}

\newcommand{\GH}{\hat{G}}

\newcommand{\PN}{\hat{\P}^{(N)}}

\newcommand{\DDO}{\hat{\D}^{(N)}}
\newcommand{\ST}{\tilde{\S}}
\newcommand{\SH}{\hat{\S}}
\NewDocumentCommand{\Dual}{O{\vec{\mu}}O{\set}m}{V_{#1,#2}^{#3} }
\newcommand{\dual}{V_{\XX,\PP,\set}}
\newcommand{\dualO}{V_{\XX,\PP,\info}}

\newcommand{\dualf}{\mathbf{V}_{\info}}
\newcommand{\dualfsp}{\mathbf{V}^{sp}_{\info}}
\newcommand{\dualt}{\hat{\mathbb{V}}^{(N)}}

\newcommand{\duall}{\widetilde{V}_{\XX,\PP,\set}}

\newcommand{\primall}{\widetilde{P}_{\XX,\PP,\set}}
\newcommand{\primal}{P_{\vec{\mu},\set}}
\newcommand{\primalf}{\mathbf{P}_{\info}}

\newcommand{\Lin}{\text{Lin}}
\NewDocumentCommand{\Lip}{O{}}{\mathfrak{G}_{#1}(\R_+)}
\NewDocumentCommand{\Lips}{O{}}{\mathfrak{G}_{#1}(\R^d_+)}
\newcommand{\Levy}{L\'{e}vy}
\newcommand{\Itos}{It\^{o}'s }

\newcommand{\td}{\mathrm{d}}
\newcommand{\AzAs}{Arzel\'{a}-Ascoli }
\NewDocumentCommand\OneOptWithDefault{O{default}m}{Text using #1 (could be the default) and #2}
\newcommand{\overbar}[1]{\mkern 1.5mu\overline{\mkern-1.5mu#1\mkern-1.5mu}\mkern 1.5mu}

\newcommand{\indicator}[1]{\mathbbm{1}_{\left\{ {#1} \right\} }}
\newcommand{\indicators}[1]{\mathbbm{1}_{{#1}}}

\newcommand\restr[2]{{
  \left.\kern-\nulldelimiterspace 
  #1 
  \vphantom{\big|} 
  \right|_{#2} 
  }}

\usepackage[
bookmarks,
bookmarksopen=true,
pdftitle={},
pdfauthor={Jan Obloj, Peter Spoida},
pdfcreator={Jan Obloj, Peter Spoida},
pdfsubject={},
pdfkeywords={},
colorlinks=false,
anchorcolor=black,
filecolor=magenta, 
menucolor=red, 
linkcolor=black, 
]{hyperref}

\title{On robust pricing--hedging duality in continuous time
\thanks{We are grateful for helpful discussions we have had with Mathias Beiglb\"ock, Bruno Bouchard, Yan Dolinsky, Kostas Kardaras, Marcel Nutz, Mete Soner, Peter Spoida, Nizar Touzi as well as participants in the \emph{Workshop on Robust optimization in Finance} in December 2012 and \emph{Workshop on Robust Techniques in Financial Economics} in March 2014, both at ETH Zurich, \emph{Labex Luis Bachelier -- SIAM -- SMAI Conference on Financial Mathematics} in June 2014 in Paris, and \emph{SIAM 2014 Conference on Financial Mathematics} in Chicago.}
}
\author{Zhaoxu Hou\thanks{Zhaoxu Hou gratefully acknowledges PhD studentship from the Oxford-Man Institute of Quantitative Finance and support from Balliol College in Oxford. 
E-mail: \texttt{zhaoxu.hou@maths.ox.ac.uk}}\hspace{3mm} and  \hspace{3mm} Jan Ob\l{}\'{o}j\thanks{
Jan Ob\l\'oj gratefully acknowledges funding received from the European Research Council under the European Union's Seventh Framework Programme (FP7/2007-2013) / ERC grant agreement no. 335421 and is also thankful to the Oxford-Man Institute of Quantitative Finance and St John's College in Oxford for their financial support. E-mail: \texttt{jan.obloj@maths.ox.ac.uk}; web: \url{http://www.maths.ox.ac.uk/people/jan.obloj}}\\
Mathematical Institute, 
University of Oxford\\ AWB, ROQ, Oxford OX2 6GG, UK}
\date{\today}

\begin{document}

\setcounter{secnumdepth}{3}
\setcounter{tocdepth}{3}

\maketitle

\begin{abstract}
We pursue robust approach to pricing and hedging in mathematical finance. We consider a continuous time setting in which some underlying assets and options, with continuous paths, are available for dynamic trading and a further set of European options, possibly with varying maturities, is available for static trading. Motivated by the notion of prediction set in \citet{Mykland_prediction_set}, we include in our setup modelling beliefs by allowing to specify a set of paths to be considered, e.g.\ super-replication of a contingent claim is required only for paths falling in the given set. Our framework thus interpolates between model--independent and model--specific settings and allows to quantify the impact of making assumptions or gaining information. We obtain a general pricing-hedging duality result: the infimum over superhedging prices is equal to supremum over calibrated martingale measures. In presence of non-trivial beliefs, the equality is between limiting values of perturbed problems. In particular, our results include the martingale optimal transport duality of \citet{Yan} and extend it to multiple dimensions and multiple maturities.
\end{abstract}

\section{Introduction}
\textbf{Two approaches to pricing and hedging.} 
The question of pricing and hedging of a contingent claim lies at the heart of mathematical finance. 
Following Merton's seminal contribution \cite{merton_no_dominance1973}, we may distinguish two ways of approaching it. First, one may want to make statements ``based on assumption sufficiently weak to gain universal support\footnote{\citet{merton_no_dominance1973}\label{foot:Merton}}," e.g.\ market efficiency combined with some broad mathematical idealisation of the market setting. We will refer to this perspective as \emph{the model-independent approach}. While very appealing at first, it has been traditionally criticised for producing outputs which are too imprecise to be of practical relevance. This is contrasted with the second, \emph{model-specific approach} which focuses on obtaining explicit statements leading to unique prices and hedging strategies. ``To do so, more structure must be added to the problem through additional assumptions at the expense of loosing some agreement$^{\ref{foot:Merton}}$." Typically this is done by fixing a filtered probability space $(\Omega, \FF,(\FF_t)_{t\ge 0}, \P)$ with risky assets represented by some adapted process $(S_t)$. 

The model-specific approach, originating from the seminal works of \citet{Samuelson:65} and \citet{BlackScholes:73}, has revolutionised the financial industry and became the dominating paradigm for researchers in quantitative finance. 
Accordingly, we refer to it also as the \emph{classical approach}. The original model of Black and Scholes has been extended and generalised, e.g.\ adding stochastic volatility and/or stochastic interest rates, trying to account for market complexity observed in practice. Such generalisations often lead to market incompleteness and lack of unique rational warrant prices. Nevertheless, no-arbitrage pricing and hedging was fully characterised in body of works on the Fundamental Theorem of Asset Pricing (FTAP) culminating in \citet{FTAP}. The feasible prices for a contingent claim correspond to expectations of the (discounted) payoff under equivalent martingale measures (EMM) and form an interval. The bounds of the interval are also given by the super- and sub- hedging prices. Put differently, the supremum of expectations of the payoff under EMMs is equal to the infimum of prices of super-hedging strategies. We refer to this fundamental result as the \emph{pricing--hedging duality}. 
\smallskip\\
\textbf{Short literature review.} 
The ability to obtain unique prices and hedging strategies, which is the strength of the model-specific approach, relies on its primary weakness -- the necessity to postulate a fixed probability measure $\P$ giving a full probabilistic description of future market dynamics. Put differently, this approach captures risks within a given model but fails to tell us anything about the model uncertainty, also called the \emph{Knightian uncertainty}, see \citet{Knight:21}.
Accordingly, researchers extended the classical setup to one where many measures $\{\P_\alpha:\alpha\in \Lambda\}$ are simultaneously deemed feasible. This can be seen as weakening assumptions and going back from the model-specific towards model-independent. The pioneering works considered \emph{uncertain volatility}, see \citet{Lyons:95} and \citet{AvellanedaLevyParas:95}. More recently, a systematic approach based on quasi-sure analysis was developed with  stochastic integration based on capacity theory in \citet{DenisMartini:06} and on the aggregation method in \citet{SonerTouziZhang:11}, see also \citet{NeufeldNutz:13}. In discrete time a corresponding generalisation of the FTAP and the pricing-hedging duality was obtained by \citet{BouchardNutz:14} and in continuous time by \citet{biagini2014robust}, see also references therein. We also mentions that setups with frictions, e.g.\ trading constraints, were considered, see \citet{BayraktarZhou:14}.

In parallel, the model-independent approach has also seen a revived interest. This was mainly driven by the observation that with the increasingly rich market reality this ``universally acceptable" setting may actually provide outputs precise enough to be practically relevant. Indeed, in contrast to when \citet{merton_no_dominance1973} was examining this approach, at present typically not only underlying is liquidly traded but so are many European options written on it. Accordingly, these should be treated as \emph{inputs} and hedging instruments, thus reducing the possible universe of no-arbitrage scenarios. \citet{breeden1978prices} were first to observe that if many (all) European options for a given maturity trade then this is equivalent to fixing the marginal distribution of the stock under any EMM in classical setting. \citet{Hobson} in his pioneering work then showed how this can be used to compute model-independent prices and hedges of lookback options. Other exotic options were analysed in subsequent works, see \citet{BHR:01}, \citet{CoxWang:11}, \citet{cox2011robust}. The resulting no-arbitrage price bounds could still be too wide even for market making but the associated hedging strategies were shown to perform remarkably well when compared to traditional delta-vega hedging, see \citet{OblojUlmer:10}. Note that the superhedging property here is understood in a pathwise sense and typically the strategies involve buy-and-hold positions in options and simple dynamic trading in the underlying. The universality of the setting and relative insensitivity of the outputs to (few) assumptions earned the setup the name of \emph{robust approach}. 

In the wake of financial crisis, significant research focus shifted back to the model-independent approach and many natural questions, such as establishing  the pricing-hedging duality and a (robust) version of the FTAP, were pursued. 
In a one-period setting, the pricing-hedging duality was linked to the Karliin-Ishi duality in linear programming by \citet{DavisOblojRaval:14}. \citet{mass_transport} re-interpreted the problem as a \emph{martingale optimal transport problem} and established general discrete time pricing-hedging duality as an analogue of the Kantorovich duality in the optimal transport: here the primal elements are martingale measures, starting in a given point and having fixed marginal distribution(s) via the \citet{breeden1978prices} formula. The dual elements are sub- or super- hedging strategies and the payoff of the contingent claim is the ``cost functional." An analogue result in continuous time, under suitable continuity assumptions, was obtained by \citet{Yan} who also, more recently, considered the discontinuous setting \cite{dolinsky2014martingale}. These topics remain an active field of research. \citet{acciaio2013model} considered pricing-hedging duality and FTAP with an arbitrary market input in discrete time and under significant technical assumptions. These were relaxed offering great insights in a recent work of \citet{BurzoniFrittelliMaggis:15}.  
\citet{GalichonHenryLabordereTouzi:11} applied the methods of stochastic control to deduce the model-independent prices and hedges, see also \citet{HOST:14}. 
Several authors considered setups with frictions, e.g.\ transactions costs in \citet{Dolinsky:2014aa} or trading constraints in \citet{emergenceofbubble} and \citet{FahimHuang:14}.
\smallskip\\
\textbf{Main contribution.}
The present work contributes to the literature on robust pricing and hedging of contingent claims in two ways.
First, inspired by \citet{Yan}, we study the pricing-hedging duality in continuous time and extend their results to multiple dimensions, different market setups and options with uniformly continuous payoffs. Our results are general and obtained in a parsimonious setting. We specify explicitly several important special cases including: the setting when finitely many options are traded, some dynamically and some statically, and the setting when all European call options for $n$ maturities are traded. The latter gives the martingale optimal transport (MOT) duality with $n$ marginal constraints which was also recently studied in a discontinuous setup by \citet{dolinsky2014martingale} and, in parallel to our work, by \citet{GuoTanTouzi:15}.

Our second main contribution is to propose a \emph{robust approach} which subsumes the model-independent setting but allows to include assumptions and move gradually towards the model-specific setting. In this sense, we strive to provide a setup which connects and interpolates between the two ends of spectrum considered by \citet{merton_no_dominance1973}. In contrast, all of the above works on model-independent approach stay within Merton \cite{merton_no_dominance1973}'s ``universally accepted" setting and analyse the implications of incorporating the ability to trade some options at given market prices for the outputs: prices and hedging strategies of other contingent claims. We amend this setup and allow to express modelling beliefs. These are articulated in a pathwise manner. More precisely, we allow the modeller to deem certain paths \emph{impossible} and exclude them from then analysis: the superhedging property is only required to hold on the remaining set of paths $\set$. This is reflected in the form of the pricing-hedging duality we obtain. 

Our framework was inspired by \citet{Mykland_prediction_set}'s idea of incorporating a prediction set of paths into pricing and hedging problem. On a philosophical level we start with the ``universally acceptable" setting and proceed by \emph{ruling out} more and more scenarios as \emph{impossible}, see also \citet{Cassese:08}. We may proceed in this way until we end up with paths supporting a unique martingale measure, e.g.\ a geometric Brownian motion, giving us essentially a model-specific setting. However, the hedging arguments are always required to work for all the paths which remain under consideration and a (strong) arbitrage would be given be a strategy which makes positive profit for all remaining paths, see also the recent work of \citet{BurzoniFrittelliMaggis:15}. This should be contrasted with another way of interpolating between model-independent and model-specific: one which starts from a given model $\P$ and proceeds by \emph{adding} more and more \emph{possible} scenarios $\{\P_\alpha: \alpha\in \Lambda\}$. This naturally leads to probabilistic (quasi-sure) hedging and different notions of no-arbitrage, see \citet{BouchardNutz:14}.

Our approach to establishing the pricing-hedging duality involves both discretisation, as in \citet{Yan}, as well as a variational approach as in \citet{GalichonHenryLabordereTouzi:11}. We first prove an ``unconstrained'' duality result: 
Theorem \ref{theorem:Main} states that for any derivative with bounded and uniformly continuous payoff function $G$, the minimal initial set-up cost of a portfolio consisting of cash and dynamic trading in the risky assets (some of which could be options themselves) which superhedges the payoff $G$ for every non-negative continuous path, is equal to the supremum of the expected value of $G$ over all non-negative continuous martingale measures\footnote{Note that here and throughout, we assume that all assets are discounted or, more generally, are expressed in terms of some numeraire.}. This result is shown through an elaborate discretisation procedure building on ideas in \cite{Yan, dolinsky2014martingale}.
Subsequently, we develop a variational formulation which allows us to add statically traded options, or specification of prediction set $\set$, via Lagrange multipliers. In some cases this leads to ``constrained" duality result, similar to ones obtained in works cited above, with superhedging portfolios allowed to trade statically the market options and martingale measures required to reprice these options. In particular Theorems \ref{thm: countable puts}  and  \ref{thm:joint_distribution_weak_main_2} extend the duality obtained respectively in \citet{DavisOblojRaval:14} and \cite{Yan}. However in general we obtain an \emph{asymptotic} duality result with the dual and primal problems defined through a limiting procedure. The primal value is the limit of superhedging prices  on $\epsilon$-neighbourhood of $\set$ and the dual value is the limit of supremum of expectation of the payoff over $\epsilon$-(miss)calibrated models, see Definitions \ref{Defpath} and \ref{def: approximation_of_primal}. 

The paper is organised as follows. Section \ref{section:setup} introduces our robust framework for pricing and hedging and defines the primal (pricing) and dual (hedging) problems.  
 Section \ref{section:main results} contains all the main results. First, in Section \ref{section: General duality}, we present the unconstrained pricing-hedging duality in Theorem \ref{theorem:Main} and derive constrained (asymptotic) duality results under suitable compactness assumptions. This allows us in particular to treat the case of finitely many traded options.  Then in Sections \ref{section: Martingale optimal transport duality for bounded claims}--\ref{section: Martingale optimal transport duality for unbounded claims--exact} we apply the previous results to the martingale optimal transport case.  All the result except Theorem \ref{theorem:Main} are proved in Section \ref{sec:first_proofs}. Theorem \ref{theorem:Main} is proved in Sections \ref{section:discretisation of the dual} and \ref{section:discretisation of the primal}. The proof proceeds via discretisation: of the primal problem in Section \ref{section:discretisation of the dual} and of the dual problem in Section \ref{section:discretisation of the primal}. Proofs of two auxiliary results are relegated to the Appendix.

\begin{notation*}
We gather here, principal notation used in this paper.
\begin{itemize}\compresslist
\item $\Omega$ is the set of all $\R_+^{d+K}$ valued continuous functions $f:[0,T]\to \R_+^{d+K}$ s.t.\ $f_0=(1,\ldots,1)$.
\item $\info\subset \Omega$ encodes further market information, e.g.\ the payoff constraints at maturity. 
\item $\set\subset \info$ is the prediction set, i.e.\ the set of paths the agent wants to consider.
\item For Banach spaces $E$ and $F$, $\CC(E,F)$ denotes the set of continuous $F$-valued function $f$ on $E$, endowed with the usual sup norm $\|\cdot\|_{\infty}$. 
\item $\DD([0,T],\R^d)$ is the set of all $\R^d$-valued measurable functions $f:[0,T]\to \R_+^d$.
\item $\D([0,T],\R^{d})$ is the space of all $\R^d$-valued right continuous functions $f:[0,T]\to\R^d$ with left limits.
\item $\S=(\S^{(1)},\ldots,\S^{(d)})$ is the canonical process on $\Omega$ and $\F=(\FF_{i})_{i=1}^n$ its natural filtration.
\item For any $m\ge 1$, $|\cdot|:\R^m\to \R$ is the norm $|x|=\sup_{1\le i\le m}|x^{(i)}|$, where $x=(x^{(i)},\ldots,x^{(m)})$. 
\item $\|\S\|=\sup\{|\S_t|\,:\, t\in [0,T]\}.$
\item $\XX$ is the set of market options available for static trading at time $t=0$.
\item $\PP$ is a linear pricing operator on $\XX$ specifying the initial prices for $X\in \XX$.
\item $\AA$ is the set of $\ts$ such that $\ts:\Omega\to \DD([0,T_n],\R^d)$ is progressively measurable and of bounded variation, satisfying 
 \begin{equation*}
  \int_0^t\ts_u(S)\cdot dS_u\ge-M, \quad \forall\, S\in \info,\,t\in [0,T_n], \text{ for some $M>0$.}
  \end{equation*}
See Section \ref{subsection:Admissible portfolio}.
\item For any $p\ge 0$, $\AA^{(p)}$ is the set of $\ts$ such that $\ts:\Omega\to \DD([0,T_n],\R^d)$ is progressively measurable and of bounded variation, satisfying 
 \begin{equation*}
  \int_0^t\ts_u(S)\cdot dS_u\ge-M(1+\sup_{0\le u\le t}|S_u|^p), \quad \forall\, S\in \Omega,\,t\in [0,T_n], \text{ for some $M>0$.}
  \end{equation*}
See Section \ref{section: Martingale optimal transport duality for unbounded claims}.
\item $\MSUP$ is defined to be the set of probability measure $\P$ on the space $(\Omega,\FF_{T_n},\F)$ such that $\S$ is a local martingale under $\P$, see Subsection \ref{section: Market models}. 
\item Let $\MSUPs$ be the collection of $\P\in \MSUP$ such that $\P=\P^W \circ M^{-1}$ for some continuous martingale $M$ defined on $(\Omega^W,\FF^W_{T_n},\F^W,P^W)$, where $(\Omega^{W},\FF^W_{T_n},\F^W, P^W)$ is a complete probability space together with a finite dimensional Brownian motion $\{W_t\}_{t=0}^{\infty}$ and the natural filtration $\FF^{W}_t=\sigma\{W_s|s\le t\}$, see Subsection \ref{section: Market models}. 
\item For any notation $\mathfrak{N}$ which is defined by using $\MSUP$, we write $\underline{\mathfrak{N}}$ to denote the one that is defined in the same way as $\mathfrak{N}$ but by using $\MSUPs$ instead of $\MSUP$, which is considered as an analogue of $\mathfrak{N}$, for example, as defined in Subsection \ref{section: General duality}, $\primalf(G)$ is denoted $\displaystyle \sup_{\P\in \MSUP}\E_{\P}[G(\S)]$, and hence $\primalfs(G)=\displaystyle \sup_{\P\in \MSUPs}\E_{\P}[G(\S)]$.
\item $d_p$ is the \Levy--Prokhorov's metric on probability measures on $\R_+^d$ given by
\begin{equation}
d_p(\mu,\nu):=\sup_{f\in \mathfrak{G}^b_1(\R_+^d)}\Big|\int f d\nu-\int f d\mu\Big|, 
\end{equation} 
where $\displaystyle  \mathfrak{G}^b_1(\R_+^d):=\big\{f\in C(\R^d_+,\R): \|f\|\le 1 \text{ and } |f(x)-f(y)|\le |x-y| \;\forall x, y\big\}$ (for more details, see \citet{measure_theory}, Chapter 8, Theorem 8.3.2.). 
\end{itemize}
\end{notation*}

\section{Robust Modelling Framework}\label{section:setup}
\subsection{Traded assets}

We consider a financial market with $d+1$ assets: a numeraire (e.g.\ the money market account) and $d$ underlying assets $S^{(i)},\ldots,S^{(d)}$, which may be traded at any time $t\le T_n$. All prices are denominated in the units of the numeraire. In particular, the numeraire's price is thus normalised and equal to one. We assume that the price path $S_t^{(i)}$ of each risky asset is continuous. The assets start at $S_0=(1,\ldots,1)$ and are assumed to be non-negative. We work on the canonical space $\CC([0,T_n],\R_+^{d})$, the set of all $\R_+^d$-valued continuous functions on $[0,T_n]$.

We pursue here a robust approach and do not postulate any probability measure which would specify the dynamics for $S$. Instead we assume that there is a set $\XX$ of market traded options with prices known at time zero, $\PP(X)$, $X\in \XX$. In all generality, an option $X\in \XX$ is just a mapping $X:\CC([0,T_n],\R_+^{d}) \to \R$, measurable with respect to the $\sigma$-field generated by coordinate process. However most often we will consider European options, i.e.\ $X(S)=f(S_{T_i})$ for some $f$ and for maturities $0<T_1<\ldots<T_n=T$.
The trading is frictionless so prices are linear and options in $\XX$ may be bought or sold at time zero at their known prices. 

Further, we allow some of the options to be traded continuously. We do this by augmenting the set of risky assets so that there are $d+K$ assets which may be traded at any time $t\le T_n$: $d$ underlying assets $S$ and $K$ options $X^{(c)}_1(S),\ldots, X^{(c)}_K(S)$. We assume that $X^{(c)}_1, \ldots, X^{(c)}_K$ are European options with maturity $T_n$ and have continuous price paths. In addition, they have non-negative payoffs and their prices today $\PP(X^{(c)}_i)$'s are strictly positive. Hence, by normalisation, we can assume without loss of generality the price of each option starts at $1$ and never goes below $0$. We now consider a natural extension of the path space 
$$\Omega =\{f\in\CC([0,T_n],\R_+^{d+K})\,:\, f_0=(1,\ldots,1) \}.$$ The coordinate process on $\Omega$ is denoted $\S=(\S_t)_{0\le t \le T_n}$ i.e.\ 
\begin{align*}
\S=(\S^{(1)},\ldots,\S^{(d+K)})\,:\, [0,T]\to \R^d_+,
\end{align*}
and $\F=(\FF_{t})_{0\le t\le T_n}$ is its natural filtration. However, not every $\omega$ in $\Omega$ is a good candidate for price path of these assets. It can be seen from the fact that the prices of option $X^{(c)}_i$ and $S$ at time $T_n$ should always respect the payoff function $X^{(c)}_i$. Therefore, for the purpose of pricing and hedging duality, we only need to consider the set of possible price paths of these $d+K$ assets, denoted $\info$, i.e.\
$$\info = \{\omega\in  \Omega\,:\, \omega^{(d+i)}_{T_n} = X^{(c)}_{i}(\omega^{(1)}_{T_n},\ldots, \omega^{(d)}_{T_n})/\PP(X^{(c)}_{i})\; \forall i\le K \}.$$
$\info$, called the information space, encodes not only the prices of these underlying assets and options at time zero, but also future payoff constraints.

\subsection{Trading strategies}\label{subsection:Admissible portfolio}
A trading strategy consists of two parts. The first part is static hedging $X$, which is a linear combination of market traded options. In contrast, the other part, known as the dynamic trading, features a potentially continuous trading in the underlying asset and a few selected European options. Heuristically, the capital gain from this trading activity takes the integral form of $\int\ts_u(S)\cdot dS_u$. To define this integral properly, we need to impose some regularity condition on $\ts$. Here, we follow \citet{Yan} and consider $\ts:[0,T_n]\to \R^{d+K}$ of finite variation for which, using integration by parts formula, for any continuous $S$ we set
\begin{equation*}
\int_0^t\ts_u(S)\cdot dS_u=\ts_tS_t-\ts_0S_0-\int_0^tS_u\cdot d\ts_u,
\end{equation*} 
where the last term on the right hand side is a Stieltjes integral.

Further, $\gamma$ is required to be progressively measurable with respect to a filtration which, in our context, is the natural filtration generated by the canonical process. More precisely, we have:
\begin{definition}
We say that a map $\phi:\Omega\to \DD([0,T_n],\R^{d+K})$
is \textsl{progressively measurable}, if $\forall \upsilon,\hat{\upsilon}\in A$, 
\begin{equation}
  \upsilon_u=\hat{\upsilon}_u, \quad \forall u\in[0,t] \quad \Rightarrow \quad \phi(\upsilon)_t=\phi(\hat{\upsilon})_t.\label{eq:pm}
  \end{equation}
\end{definition}
We say $\ts$ is admissible if $\gamma:\Omega\to \DD([0,T_n],\R_+^{d+K})$ is progressively measurable and of finite variation, satisfying 
 \begin{equation}
  \int_0^t\ts_u(S)\cdot dS_u\ge-M, \quad \forall\, S\in \info,\,t\in [0,T_n], \text{ for some $M>0$.}\label{eq:admissible}
  \end{equation}
Let $\AA$ be the set of such integrands. The set of simple integrands, i.e.\ $\ts\in \AA$ such that $\ts(\omega)$ is a simple function $\forall \omega \in \Omega$, is denoted $\AA^{sp}$.

An admissible (semi-static) trading strategy is a pair $(X,\gamma)$ where $X=a_0+\sum_{i=1}^m a_iX_i$, for some $m$, $X_i\in \XX$ and $a_0,a_i\in \R$, $i=1,\ldots, m$ and $\gamma\in \AA$. The cost of following such a trading strategy is equal to the cost of setting up its static part, i.e.\ of buying the options at time zero, and is equal to 
\begin{align*}
\PP(X):=a_0+\sum_{i=1}^m a_i\PP(X_i).
\end{align*}
We denote the class of admissible (semi-static) trading strategies by $\AA_\XX$ and $\AA_\XX^{sp}$ for $\ts\in \AA$ or $\AA^{sp}$ respectively. 

\subsection{Beliefs}

As argued in the Introduction, we allow our agents to express modelling beliefs. These are encoded as restrictions of the pathspace and may come from time series analysis of the past data, or idiosyncratic views about market in the future. 
Put differently, we are allowed to rule out paths which we deem \emph{impossible}. The paths which remain are referred to as \emph{prediction set} or \emph{beliefs}. Note that such beliefs may also encode one agent's superior information about the market. 

We will consider pathwise arguments and require that they work provided the price path $S$ falls into the predictions set $\set\subseteq\info$.  Any path falling out of $\set$ will be ignored in our considerations. This binary way of specifying beliefs is motivated by the fact that in the end we only see one paths and hence we are interested in arguments which work pathwise. Nevertheless, the approach is very parsimonious and as $\set$ changes from all paths in $\info$ to a support of a given model we essentially interpolate between model-independent and model-specific setups. It also allows to incorporate the information from time-series of data coherently into the option pricing setup, as no probability measure is fixed and hence no distinction between real world and risk neutral measures is made. The idea of such a prediction set first appeared in \citet{Mykland_prediction_set}; also see \citet{NadtochiyObloj:15} and \cite{emergenceofbubble} for an extended discussion. 

As the agent rejects more and more paths, i.e.\ takes $\set$ smaller and smaller, the framework's outputs -- the robust price bounds, should get  tighter and tighter. This can be seen as a way to quantify the impact of making assumptions or acquiring additional insights or information. 

\subsection{Superreplication}

Our prime interest is in understanding robust pricing and hedging of a derivative with payoff $G:\Omega\to \R$ whose price is not quoted in the market. Our main results will consider bounded payoffs $G$ and, since the setup is frictionless and there are no trading restrictions, without any loss of generality we may consider only the superhedging price. The subhedging follows by considering $-G$.

\begin{definition}\label{Defpath}
\
  \begin{enumerate}  
 \item A portfolio $(X,\ts)\in\AA_\XX $ is said to \textsl{super-replicate} $G$ on $\set$ if 
  \begin{equation}
 X(S)+\int_0^{T_n}\ts_u(S)\cdot dS_u\ge G(S),\;\; \forall S\in \set.\label{eq:super replicating}
  \end{equation}
  \item The (minimal) super-replicating cost of $G$ on $\set$ is defined as
\begin{equation}
\dual(G):=\inf\Big\{\PP(X)\,:\, \exists (X,\ts)\in \AA_{\XX} \text{ s.t.\ $(X,\ts)$ super-replicates $G$ on $\set$}\Big\}.\label{eq:minimal_super_replication}
\end{equation}
\item The approximate super-replicating cost of $G$ on $\set$ is defined as
\begin{align}
\duall(G):=\inf\Big\{\PP(X)\,:\,& \exists (X,\ts)\in \AA_{\XX} \text{ s.t.\ }  \nonumber\\
&\text{ $(X,\ts)$ super-replicates $G$ on $\set^{\epsilon}$ for some $\epsilon> 0$}\Big\},\label{eq:L_super_replication}
\end{align}
where $\set^{\epsilon}=\{\omega\in \info\,:\,\inf_{\upsilon\in \set}\|\omega-\upsilon\|\le \epsilon\}$.
\item Finally, we let $\dual^{sp}(G)$, respectively $\duall^{sp}(G)$, denote the super-replicating cost of $G$ in \eqref{eq:minimal_super_replication}, respectively in \eqref{eq:L_super_replication}, but with $(X,\ts)\in \AA^{sp}_\XX$.
  \end{enumerate}
\end{definition}

\subsection{Market models}\label{section: Market models}

Our aim is to relate the robust (super)hedging cost, as introduced above, to the classical pricing-by-expectation arguments. To this end we look at all classical models which reprice market traded options.

\begin{definition}\label{def:primal}
We denote by $\MSUP$ the set of probability measures $\P$ on $(\Omega,\FF_{T_n},\F)$ such that $\S$ is a $\P$--martingale and let $\MSUPI$ be the set of probability measures $\P\in \MSUP$ such that $\P(\info) =1$.\\
A probability measure $\P\in\MSUPI$ is called a $(\XX,\PP,\set)$--market model, or simply a calibrated model, if $\P(\set)=1$ and $\E_{\P}[X]=\PP(X)$ for all $X\in \XX$. The set of such measures is denoted $\MSUPX$.\\
More generally, a probability measure $\P\in\MSUPI$ is called an $\eta-(\XX,\PP,\set)$--market model if $\P(\set^\eta)>1-\eta$ and 
$|\E_{\P}[X]-\PP(X)|<\eta$ for all $X\in \XX$. The set of such measures is denoted $\MSUPXa$.
\end{definition}

Whenever we have $\P\in \MSUPX$ it provides us with a feasible no-arbitrage price $\E_{\P}[G(\S)]$ for a derivative with payoff $G$. The robust price for $G$ is given as 
$$\primalx(G):=\sup_{\P\in\MSUPX}\E_{\P}[G(\S)],$$
where throughout the expectation is defined with the convention that $\infty-\infty=-\infty$. In the cases of particular interest, $(\XX,\PP)$ will determine uniquely the marginal distributions of $\S$ at given maturities and $\primalx(G)$ is then the value of the corresponding martingale optimal transport problem. We will often use this terminology, even in the case of arbitrary $\XX$.

In practice, the market prices $\PP$ are an idealised concept and may be obtained from averaging of bid-ask spread or otherwise. It might not be natural to require a perfect calibration and the concept of $\eta$--market model allows for a controlled degree of mis-calibration. This leads to the approximate value given as
$$\primalxa(G):=\lim_{\eta\searrow 0}\sup_{\P\in\MSUPXa}\E_{\P}[G(\S)].$$
As we will show below, both the approximate superhedging cost and the approximate robust pricing cost, while being motivated by practical considerations, appear very naturally when considering abstract pricing-hedging duality.

In some instances, for technical reasons, it will be convenient to consider only $\P$ arising within a Brownian setup. We denote the collection of $\P\in \MSUP$ such that $\P=\P^W \circ M^{-1}$ for some continuous martingale $M$ defined on some probability space satisfying the usual assumptions $(\Omega^W,\FF^W_{T_n},\F^W,P^W)$ with a finite dimensional Brownian motion $\{W_t\}_{t\geq 0}$ which generates the filtration $\F^W$. We write $\MSUPXs$ to denote $\MSUPX\cap \MSUPs$, $\MSUPXsa$ for $\MSUPXa\cap \MSUPs$ and $\primalxs(G):=\sup_{\P\in\MSUPXs}\E_{\P}[G(\S)]$, $\primalxsa(G):=\lim_{\eta\searrow 0}\sup_{\P\in\MSUPXsa}\E_{\P}[G(\S)]$.

\section{Main results}\label{section:main results}

Our prime interest, as discussed in the Introduction, is in establishing a general robust pricing--hedging duality. Given a non-traded derivative with payoff $G$ we have two candidate robust prices for it. The first one, $\dual(G)$, is obtained through pricing-by-hedging arguments. The second one, $\primalx(G)$, is obtained by pricing-via-expectation arguments. In a classical setting, the analoguous two prices are equal. This is trivially true in a complete market and is a fundamental result for incomplete markets, see Theorem 5.7 in \citet{FTAP}.

Within the present pathwise robust approach, the pricing--hedging duality was obtained for specific payoffs $G$ in literature linking robust approach with the Skorokhod embedding problem, see \citet{Hobson} or \citet{OblojEQF:10} for discussion. Subsequently, an abstract result
was established in \citet{Yan}, when $\Omega=\info=\set$, $n=d=1$, $K=0$ and $\XX$ is the set of all call and put options with $\PP(X)=\int_0^{\infty} X(x)\mu(dx)$ for all $X\in\XX$, where $\mu$ is a probability measure on $\R_+$ with mean equal to $1$:
$$ \dualO(G)=P_{\XX,\PP,\info}(G)\;\; \text{ for a `strongly continuous' class of bounded $G$ }.$$
The result was extended to unbounded claims by broadening the class of admissible strategies and imposing a technical assumption on $\mu$. Below we extend this duality to a much more general setting of abstract $\XX$, possibly involving options with multiple maturities, a multidimensional setting and with an arbitrary prediction set $\set$. 

Note that, for any Borel $G:\Omega \to \R$, the inequality 
\begin{align}
\dual(G)\ge \primalx(G) \label{eq:primallessthandual}
\end{align}
is true as long as there is at least one $\P\in \MSUPX$ and at least one $(X,\gamma)\in \AA_\XX$ which superreplicates $G$ on $\set$. Indeed, since $\ts$ is progressively measurable in the sense of \eqref{eq:pm}, the integral $\int_0^\cdot \ts_u(\S)\cdot d\S_u$, defined pathwise via integration by parts, agrees a.s. with the stochastic integral under $\P$. Then, by \eqref{eq:admissible}, the stochastic integral is a $\P$ super-martingale and hence $\E_{\P}\Big[\int_0^{T_n} \ts_u(\S)\cdot d\S_u\Big]\le 0$. This in turn implies that
\begin{align*}
\E_{\P}\Big[G(\S)\Big]\leq \PP(X).
\end{align*}
The result follows since $(X,\ts)$ and $\P$ were arbitrary.

\subsection{General duality}\label{section: General duality}

We first consider the case without constraints: $\XX=\emptyset$ and $\set=\info$. As this context will be our reference point we introduce notation to denote the super-hedging cost and the robust price. We let
\begin{equation}
\dualf(G):=\inf\Big\{x\,:\, \exists \ts\in \AA \text{ s.t.\ } 
\text{$\ts$ super-replicates $G-x$ on $\info$}\Big\},\quad \primalf(G):=\sup_{\P\in \MSUP_{\info}}\E_{\P}[G(\S)].
\label{eq:full_super_replication}
\end{equation}
We also write $\primalfs(G)$ for $\displaystyle \sup_{\P\in \MSUPs_{\info}}\E_{\P}[G(\S)]$ and $\dualfsp(G)$ for super-replicating cost of $G$ using $\ts\in \AA^{sp}$.

\begin{assumption}\label{assumption:key assumption}
Either $K=0$ or $X^{(c)}_1,\ldots, X^{(c)}_K$ are bounded and uniformly continuous with market prices $\PP(X^{(c)}_1),\ldots, \PP(X^{(c)}_K)$ satisfying that there exists an $\epsilon> 0$ such that for any $(p_k)_{1\le k\le K}$ with $|\PP(X^{(c)}_k) - p_{k}|\le \epsilon$ for all $k\le K$, $\MSUP_{\tilde{\info}}\neq \emptyset$, where
\begin{align*}
\tilde{\info}:= \{\omega\in \Omega\,:\, \omega^{(d+i)}_{T_n} = X^{(c)}_{i}(\omega^{(1)}_{T_n},\ldots, \omega^{(d)}_{T_n})/p_{i}\; \forall i\le K \}.
\end{align*}
\end{assumption}

\begin{theorem}\label{theorem:Main}
Under Assumption \ref{assumption:key assumption}, for any bounded and uniformly continuous $G:\Omega \to \R$ we have
\begin{align*}
\dualfsp(G) = \dualf(G)=\primalf(G)=\primalfs(G).
\end{align*}
\end{theorem}

An analogous duality in a quasi-sure setting was obtained in \citet{Possamai} and earlier papers, as discussed therein. However, while similar in spirit, there is no immediate link between our results or proofs and these in \cite{Possamai}. Here, we consider a comparatively smaller set of admissible trading strategies and we require a pathwise superhedging property. Consequently, we also need to impose stronger regularity constraints on $G$. The inequality 
\begin{equation*}
\dualf(G)\ge \sup_{\P\in \MSUP_{\info}}\E_{\P}[G(\S)]
\end{equation*}
is a special case of \eqref{eq:primallessthandual}. 
Sections \ref{section:discretisation of the dual} and  \ref{section:discretisation of the primal} are mainly devoted to the proof of the much harder reverse inequality
\begin{equation} 
\dualf(G)\le \sup_{\P\in \MSUPs_{\info}}\E_{\P}[G(\S)],\label{eq:thm 2}
\end{equation}
which then implies Theorem \ref{theorem:Main}. The proof proceeds through discretisation of both the primal and the dual problem.

We let $\Lin(\XX)$ denote the set of finite linear combinations of elements of $\XX$ and
\begin{align*}
\Lin_N(\XX)=\Big\{a_0+\sum_{i=1}^{m}a_iX_i\,:\,m\in \N,\,X_i\in \XX,\,\sum_{i=0}^m|a_i|\le N\Big\}.
\end{align*}  
Then, similarly to e.g.\ Proposition 5.2 in \citet{HOST:14}, a calculus of variations characterisation of $\duall$ is a corollary of Theorem \ref{theorem:Main}. From that we are able to deduce pricing-hedging duality between the approximate values.

\begin{corollary}\label{cor: calculus of variations}
Under Assumption \ref{assumption:key assumption}, let $\set$ be a measurable subset of $\info$ and $\XX$ such that all $X\in \XX$ are uniformly continuous and bounded. Then for any uniformly continuous and bounded $G:\Omega\to \R$ we have: 
\begin{equation}
\duall^{sp}(G) = \duall(G)=\inf_{X\in \Lin_N(\XX),\,N\ge 0}\Big\{\primalf(G-X-N\lambda_{\set})+\PP(X)\Big\},\label{eq: second_main_2_1}
\end{equation}
where $\lambda_{\set}(\omega):=\inf_{\upsilon\in \set}\|\omega-\upsilon\|\wedge 1$. 
\end{corollary}
\begin{remark}\label{remark: remark_v_penalty}
As a by-product of the proof of Corollary \ref{cor: calculus of variations}, we show that for any bounded $G$, 
\begin{equation}\label{eq: remark_v_penalty}
\duall(G)= \inf_{N\ge 0}\widetilde{V}_{\XX,\PP,\info}(G-N\lambda_{\set}) \;\;\text{ and }\;\; \duall^{sp}(G)= \inf_{N\ge 0}\widetilde{V}^{sp}_{\XX,\PP,\info}(G-N\lambda_{\set}).
\end{equation}
\end{remark}

\begin{assumption}\label{assumption: XX}
$\Lin_1(\XX)$ is a compact subset of $\CC(\Omega, \R)$ and every $X\in \XX$ is bounded and uniformly continuous.
\end{assumption}

\begin{theorem}\label{thm: second main_2}
Given $\info$, $\set$ and $\XX$ satisfy conditions in Corollary \ref{cor: calculus of variations}, if $\MSUPX^{\eta}\neq \emptyset$ for any $\eta >0$, then for any uniformly continuous and bounded $G:\Omega\to \R$ we have 
\begin{align}
\duall(G) \ge \primalxa(G), \label{eq: approximate weak duality}
\end{align}
and if $\XX$ satisfies Assumption \ref{assumption: XX}, then $\MSUPXs^{\eta}\neq \emptyset$ for any $\eta >0$ and equality holds: 
\begin{equation}
\duall(G) = \primalxa(G) = \primalxsa(G). \label{eq: second_main_2_2}
\end{equation}
\end{theorem}

\begin{example}[Finite $\XX$]\label{example: finite options}
\
Consider
$\XX = \{ X_1, \ldots, X_m\}$,
where $X_i$'s are bounded and uniformly continuous. In this case, $\Lin_1(\XX)$ is a convex and compact subset of $\CC(\Omega, \R)$. Therefore, if $\MSUPX^{\eta} \neq \emptyset$ for any $\eta>0$, we can apply Theorem \ref{thm: second main_2} to conclude $\duall(G)=\primall(G)$. 
\end{example}

We end this section with consideration if the approximate superhedging and robust prices, $\widetilde{V}, \widetilde{P}$, are close to the precise values $V,P$. First, we focus on the case of finitely many traded put options and no beliefs. We consider 
\begin{align}
\XX=\{(K^{(i)}_{k,j}-\S^{(i)}_{T_j})^+,\; 1\le i\le d,\, 1\le j\le n,\, 1\le k\le m(i,j)\}, \label{eq: countable puts}
\end{align}
where $0< K^{(i)}_{k,j} < K^{(i)}_{k^{\prime},j}$ for any $k < k^{\prime}$ and $m(i,j)\in \N$.
To simplify the notation, we write
\begin{align*}
\PP((K^{(i)}_{k,j} - \S^{(i)}_{T_j})^+) = p_{k,i,j} \quad \forall i,j,k.
\end{align*} 

\begin{assumption}\label{assump: countable puts}
Market put prices are such that there exists an $\epsilon> 0$ such that for any $(\tilde{p}_{k,i,j})_{i,j,k}$ with $|\tilde{p}_{k,i,j} - p_{k,i,j}|\le \epsilon$ for all $i,j,k$, there exists a $\tilde{\P}\in \MSUP_{\info}$ such that 
\begin{align*}
\tilde{p}_{k,i,j} = \E_{\tilde{\P}}[(K^{(i)}_{k,j} - \S^{(i)}_{T_j})^+] \quad \forall i,j,k.
\end{align*}  
\end{assumption}
\begin{remark}
Assumption \ref{assump: countable puts} can be rephrased as saying that the market prices $(\XX, \PP)$ are in the interior of the no-arbitrage region. 
\end{remark}
\begin{theorem}\label{thm: countable puts}
Let $\XX$ be given in $\eqref{eq: countable puts}$, prices $\PP$ be such that Assumption \ref{assump: countable puts} holds and $\info$ satisfy Assumption \ref{assumption:key assumption}. Then for any uniformly continuous and bounded $G:\Omega\to \R$, we have
\begin{align*}
V_{\XX,\PP,\info}(G) = P_{\XX,\PP,\info}(G).
\end{align*}
\end{theorem}
The above result establishes a general robust pricing-hedging duality when finitely many put options are traded. 
It extends in many ways the duality obtained in \citet{DavisOblojRaval:14} for the case of $d=n=1$ and $K=0$. Note that in general $\widetilde{V}_{\XX,\PP,\info}(G)=V_{\XX,\PP,\info}(G)$ so it follows from Example \ref{example: finite options} that in Theorem \ref{thm: countable puts} we also have $\widetilde{P}_{\XX,\PP,\info}(G)=P_{\XX,\PP,\info}(G)$. These equalities may still hold, but may also fail dramatically, when non-trivial beliefs are specified. We present two examples to highlight this.

\begin{example} In this example we consider $\set$ corresponding to Black-Scholes model. For simplicity, consider the case without any traded options $K=0, \XX =\emptyset, d=1$ and let\footnote{See also Step 4 in the proof of Theorem \ref{theorem: new_section_main} in Section \ref{sec:first_proofs}.}
$$\set = \{\omega \in \Omega\,:\, \omega \text{ admits quadratic variation and } d\langle \omega \rangle_t = \sigma^2 \omega^2_t dt, 0\leq t\leq T\}.$$
Then $\MSUP_{\set} = \{\P_\sigma\}$, where $\S$ is a geometric Brownian motion with constant volatility $\sigma$ under $\P_\sigma$.
The duality in Theorem \ref{thm: second main_2} then gives that for any bounded and uniformly continuous $G$ 
\begin{align*}
\widetilde{V}_{\set}(G) =&\, \inf\{x\,:\, \exists \gamma\in \AA \text{ s.t.\ } 
\text{$\ts$ super-replicates $G-x$ on $\set^{\epsilon}$ for some } \epsilon>0\}\\
=&\, \lim_{\eta \searrow 0} \sup_{\P\in \MSUP_{\set}^{\eta}}\E_{\P}[G].
\end{align*}
However in this case, $\P$ has full support on $\Omega$ so that $\set^{\epsilon} = \Omega$ and $\MSUP_{\set}^{\epsilon}=\MSUP$ for any $\epsilon>0$. The above then boils down to the duality in Theorem \ref{theorem:Main} and we have
\begin{align}
\widetilde{V}_{\set}(G)=V_{\info}(G)= \sup_{\P\in \MSUP}\E_{\P}[G] \ge \E_{\P_\sigma}[G] = P_{\set}(G),\label{eq: example duality gap}
\end{align} 
where for most $G$ the inequality is strict.
\end{example}

\begin{example}
Consider again the case with no traded options, $K=0, \XX =\emptyset, d=1$ and let
$$\set = \{\omega \in \Omega\,:\, \|\omega\|\le b\} \quad \text{ for some $b\ge 1$}.$$
Let $G$ be bounded and uniformly continuous and consider the duality in Theorem \ref{thm: second main_2}. For each $N\in \N$ pick $\P^{(N)}\in \MSUP_{\set}^{1/N}$ such that
\begin{equation*}
\E_{\P^{(N)}}[G]\ge\sup_{\P\in \MSUP_{\set}^{1/N}}\E_{\P}[G]-1/N.
\end{equation*}
By Doob's martingale inequality, 
\begin{equation*}
\P^{(N)}(\|\S\| > M) \le \sum_{i=1}^d \frac{\E_{\P^{(N)}}[\S_{T}]}{M}\le \frac{d}{M}.
\end{equation*}
Hence by considering $\tau_M(S) = \inf\{t \ge 0\,:\, \|S\| > M\}\wedge T$, we know
\begin{align*}
|\E_{\P^{(N)}}[G(\S^{\tau_M})] - \E_{\P^{(N)}}[G(\S)]|\le \frac{2d\|G\|_{\infty}}{M}
\end{align*}
and for $M> b+1$,
\begin{equation*}
\P^{(N)}(|\S^{\tau_M}_{T}| > b+1/N) \le  \P^{(N)}(\|\S^{\tau_M}\| > b+1/N)
= \P^{(N)}(\|\S\|> b+1/N) \le 1/N,
\end{equation*}
where the last inequality follows from the fact that $\P^{(N)}\in \MSUP_{\set}^{1/N}$.

Write $\pi^{(N)} := \LL_{\P^{(N)}}(\S^{\tau_M}_{T})$. $\pi^{(N)}$'s are  probability measures on a compact subset of $\R^d_+$, with mean 1. It follows that there exists $\{\pi^{(N_k)}\}_{k\ge 1}$, a subsequence of $\{\pi^{(N)}\}_{N\ge 1}$, converging to some $\pi$ with mean 1, and by Portemanteau Theorem, for $\epsilon>0$
\begin{align*}
\pi\big(\{\vec{x}\in \R_+^d: |x_i|\le b+\epsilon\, \forall\, i\le d\}\big) \ge \limsup_{k\to\infty}\pi^{(N_k)}\big(\{\vec{x}\in \R_+^d: |x_i|\le b+\epsilon\, \forall\, i\le d\}\big )=1.
\end{align*}
Since $\epsilon>0$ is arbitrary, $\pi\big(\{\vec{x}\in \R_+^d: |x_i|\le b\, \forall\, i\le d\}\big) = 1$ by Dominated Convergence Theorem.
It follows from Theorem \ref{thm: second main_2} that
\begin{align*}
\widetilde{V}_{\set}(G) =\lim_{N\to \infty}\sup_{\P\in\MSUPs^{1/N}_{\set}}\E_{\P}[G]
\le \limsup_{k\to\infty}\sup_{\P\in\MSUPs_{\pi^{(N_k)}}}\E_{\P}[G] +\frac{1}{N_k}+\frac{2d \|G\|_{\infty}}{M} \le \underline{P}_{\pi}(G)+ \frac{2d\|G\|_{\infty}}{M}.
\end{align*}
where the last inequality follows from Lemma \ref{thm:continuity in probaility}. It is straightforward to see that any $\P\in \MSUPs_{\pi}$ is supported on $\set$, and hence from above we have, for all large $M$, 
\begin{align*}
\widetilde{V}_{\set}(G) \le P_{\set}(G) + \frac{\|G\|_{\infty}}{M}\quad \textrm{ and hence }\quad \widetilde{V}_{\set}(G) \le P_{\set}(G).
\end{align*}
We conclude that in this example
\begin{equation*}
\widetilde{V}_{\set}(G) = \widetilde{P}_{\set}(G) = P_{\set}(G) = V_{\set}(G).
\end{equation*}
\end{example}

\subsection{Martingale optimal transport duality for bounded claims}\label{section: Martingale optimal transport duality for bounded claims}

We focus now on the cases when $(\XX,\PP)$ determine uniquely certain distributional properties of $\S$ under any $\P\in \MSUPX$.
We start with the case when $\XX$ is large enough so that the market prices $\PP$ pin down the (joint) distribution of $(\S^{(1)}_{T}, \ldots, \S^{(d)}_{T})$ under any calibrated model. Later we consider the case when only marginal distributions of $\S^i_T$ for $i\le d$ are fixed. In the former case we limit ourselves to one maturity and $\set=\info$ which simplifies the exposition. It is possible to extend these results along the lines of the latter case, when we consider prices at multiple maturities and a non-trivial prediction set, however this would increase the complexity of the proof significantly. 

Let $n=1$ and $T=T_n$. We assume market prices for a rich family of basket options are available. We consider 
\begin{equation}\label{eq:Xfulldistribution}
\XX \textrm{ s.t. } \Lin(\XX) \textrm{ is a dense subset of }\{f(\S^{(1)}_{T},\ldots, \S^{(d)}_{T})|\,f:\R_+^d\to \R \text{ bounded, Lip.\  cont.}\}.
\end{equation}
In particular, $\XX$ is large enough to determine uniquely the distribution of $\S_T$ under any calibrated model, i.e.\ there exists a unique probability distribution $\pi$ on $\R_+^d$ such that 
\begin{equation}\label{eq:pidef}
\E_\P[X]=\PP(X)=\int_{\R_+^d}X(s_1,\ldots,s_d)\pi(\td s_1,\ldots,\td s_d),\quad \forall X\in \XX, \P\in \MM_{\XX,\PP,\info}.
\end{equation}
As an example, we could take $\XX$ equal to the RHS in \eqref{eq:Xfulldistribution}.
A martingale measure $\P\in \MM_{\info}$ is a calibrated model if and only if the distribution of $\S^{(1)}_T, \ldots, \S^{(d)}_T$ under $\P$ is $\pi$. Accordingly we write $\MM_{\XX,\PP,\info}=\MM_{\pi,\info}$ with $\underline{\MM}_{\pi,\info}$, $P_{\pi,\info}$ etc.\ defined analogously. Note that in a Brownian setting, we can always define a continuous martingale $M$ valued in $\R^{d+K}_+$ with $M_0=1$, $(M^{(1)}_T, \ldots, M^{(d)}_T)\sim \pi$ and $M^{(d+i)}_T = X^{(c)}_i(M^{(1)}_T, \ldots, M^{(d)}_T)$ for every $i\le K$ simply by taking conditional expectations of a suitably chosen random variable distributed according to $\pi$ and satisfying payoff constraints. It follows that the following equivalence holds. 
\begin{lemma}\label{lemma: assumption_on_pi}
 For a probability measure $\pi$ on $\R_+^d$, $\underline{\MM}_{\pi,\info}\neq \emptyset$ if and only if ${\MM}_{\pi,\info}\neq \emptyset$ if and only if 
 \begin{equation}
\int_{\R^{d}_+}s_i \pi(\td s_1,\ldots,\td s_d)=1,\  i=1,\ldots,d.\label{eq:assumption_on_pi}  
\end{equation}
\end{lemma}
Note that if \eqref{eq:assumption_on_pi} fails, then one of the forwards is mispriced leading to arbitrage opportunities\footnote{This may be, depending on the sign of mispricing and the admissibility criterion, a \emph{strong arbitrage} in \citet{cox2011robust} or \emph{model independent arbitrage} in \citet{davis_hobson2007} and \citet{acciaio2013model} or else a weaker type of approximate arbitrage, e.g.\ a \emph{weak free lunch of vanishing risk}; see \citet{cox2011robust} and \citet{emergenceofbubble}.}. We exclude this situation from our setup. The following is then a multi-dimensional extension of the pricing-hedging duality in \citet{Yan}.

\begin{theorem}\label{thm:joint_distribution_weak_main_2}
Consider traded options $\XX$ and information space $\info$ satisfying \eqref{eq:Xfulldistribution} and Assumption \ref{assumption:key assumption}, with market prices $\PP$ such that $\MM_{\XX,\PP,\info}\neq \emptyset$.
Then for any uniformly continuous and bounded $G$, we have  
\begin{align*}
\dualO(G)=P_{\pi,\info}(G).
\end{align*}
\end{theorem}
It is clear that the above result holds if instead of assuming every $X\in \XX$ is bounded and Lipschitz continuous, we allow bounded and uniformly continuous European payoffs, as long as $\XX$ contains a subset made of bounded and Lipschitz continuous payoffs, which is rich enough to guarantee uniqueness of $\pi$ which satisfies \eqref{eq:pidef}.

We now turn to the case when $\XX$ is much smaller and the market prices determine marginal distributions of $\S^{(i)}_T$ for $i\le d$. For concreteness, let us consider the case when put options are traded
\begin{align}\label{eq:Xputs}
\XX=\{(K-\S^{(i)}_{T_j})^+: i=1,\ldots,d,\,\, j=1,\ldots,n,\,\,K\in \R_+\}.
\end{align}
Arbitrage considerations, see e.g.\ \citet{cox2011robust} and \citet{emergenceofbubble}, show that absence of (weak type of) arbitrage is equivalent to $\MSUPX\neq \emptyset$. Note that the latter is equivalent to market prices $\PP$ being encoded by probability measures  $(\mu^{(i)}_j)$ with 
\begin{align}\label{eq:def_mus}
p_{i,j}(K)=\PP((K-\S^{(i)}_{T_j})^+)=\int (K-s)^+\mu^{(i)}_j(\td s), 
\end{align}
where, for each $i=1,\ldots, d$, $\mu^{(i)}_1,\ldots,\mu^{(i)}_n$ have finite first moments, mean $1$ and increase in \textsl{convex order} ($\mu^{(i)}_1\preceq\mu^{(i)}_2\preceq\cdots\preceq\mu^{(i)}_n$), i.e.\ $\int \phi(x)\mu^{(i)}_1(\td x)\le\ldots\le\int \phi(x)\mu^{(i)}_n(\td x)$ for any convex function $\phi:\R_+\to \R$. In fact, as noted already by \citet{breeden1978prices}, probability measures $\mu^{(i)}_j$ are defined by
\begin{align*}
\mu^{(i)}_j([0,K])= p^{\prime}_{i,j}(K+) \quad \text{ for }K\in \R_+.
\end{align*}
We may think of $(\mu^{(i)}_j)$ and $\set$ as the modelling inputs. The set of calibrated market models $\MSUPX$ is simply the set of probability measures $\P\in\MSUP$ such that $\S^{(i)}_{T_j}$ is distributed according to $\mu^{(i)}_j$, and $\P(\set)=1$. Accordingly, we write $\MSUPX=\MMUS$ and $\primal(G)=\primalx(G)$. Furthermore, since $\mu^{(i)}_j$'s all have means equal to $1$, under any $\P\in \MMUS$, $\S$ is a (true) martingale.

\begin{remark}
It follows, see \citet{Strassen}, that $\MM_{\vec{\mu},\info}\neq \emptyset$ if and only if $\mu^{(i)}_1,\ldots,\mu^{(i)}_n$ have finite first moments, mean $1$ and increase in \textsl{convex order}, for any $i=1,\ldots,d$. However, in general, the additional constraints associated with a non-trivial 
$\set\subsetneq \info$ are much harder to understand.
\end{remark}

In this context we can improve Theorem \ref{thm: second main_2} and narrow down the class of approximate market models requiring that they match exactly the marginal distributions at the last maturity. 
\begin{definition}\label{def: approximation_of_primal}\
Let $\MMUE{\eta}$ be the set of all measure $\P\in \MSUP$ such that $\LL_{\P}(\S^{(i)}_{T_j})$, the law of $\S^{(i)}_{T_j}$ under $\P$ satisfies 
\begin{eqnarray*}
\LL_{\P}(\S^{(i)}_{T_n})=\mu^{(i)}_n\textrm{ and }d_p(\LL_{\P}(\S^{(i)}_{T_j}),\mu^{(i)}_j)\le \eta, \text{ for } j=1,\ldots,n-1,\; i=1,\ldots,d,
\end{eqnarray*}
and furthermore $\P(\set^{\eta})\ge 1-\eta$. Finally, let 
\begin{equation*}
\primalls(G):=\lim_{\eta\searrow 0}\sup_{\P\in \MMUE{\eta}}\E_{\P}[G(\S)].
\end{equation*} 
\end{definition}
Note that $\MMUE{\eta}\subset \MSUPX^{\epsilon(\eta)}$ for a suitable choice\footnote{One can take $\epsilon(\eta)=\sqrt{\eta} + 2f(1/\sqrt{\eta})$ with $ f(K) = \max_{1\le i\le d} \big\{p_{i,n}(K)-K+1 \big\}$.} of $\epsilon(\eta)$ which converges to zero as $\eta\to 0$. 
It follows that $\primalls(G)\leq \primalxa(G)$. The following result extends and sharpens the duality obtained in Theorem \ref{thm: second main_2} to the current setting.

\begin{theorem}\label{thm: weak_main}
Let $\set$ be a measurable subset of $\info$, $\XX$ be given by \eqref{eq:Xputs} and $\PP$ be such that, for any $\eta>0$, $\MMUE{\eta}\neq\emptyset$, where $\vec{\mu}$ is defined via \eqref{eq:def_mus}. Then for any uniformly continuous and bounded $G$ the robust pricing-hedging duality holds between the approximate values:
\begin{align*}
\duall(G)=\primalxa(G)=\primalls(G).
\end{align*}
\end{theorem}

\subsection{Martingale optimal transport duality for unbounded claims}\label{section: Martingale optimal transport duality for unbounded claims}

We want to extend Theorem \ref{thm: weak_main} to unbounded exotic options, including a lookback option. However, the admissibility condition considered so far, and given by \eqref{eq:admissible}, is too restrictive and has to be relaxed. To see this consider $d=1$, $K = 0$, $\XX$ is given by \eqref{eq:Xputs} and $G(\S)=\sup_{0\le t\le T_n}\S_{t}$. If $G$ could be super-replicated by an admissible trading strategy $(X,\gamma)\in \AA_\XX$ then, following similar arguments as for \eqref{eq:primallessthandual}, we see that 
\begin{align*}
\primalf(G-X)\le 0.\label{eq:extension_contradiction}
\end{align*}
This is clearly impossible since $X$ is bounded and there exists $\P\in \MM$ such that $\E_{\P}[G(\S)]=\infty$. The argument is similar if instead of puts we took all call options. We conclude that we need to enlarge the set of dynamic trading strategies $\AA$. 

We fix $p>1$ and, following \citet{Yan}, define the following admissibility condition: $\ts$ is admissible if $\gamma:\Omega\to \DD[0,T_n]$ is progressively measurable and of bounded variation, satisfying 
 \begin{equation}
  \int_0^t\ts_u(S)\cdot dS_u\ge-M\big(1+\sup_{0\le s\le t}|S_s|^p\big), \quad \forall\, S\in \info,\,t\in [0,T_n], \text{ for some $M>0$.}\label{eq:admissible_2}
  \end{equation}
To avoid confusion, we denote by $\AA^{(p)}$ the set of all such $\ts$. We also say $(X,\ts)\in \AA^{(p)}_\XX$ if $\ts\in \AA^{(p)}$ and $X=a_0+\sum_{i=1}^m a_iX_i$, for some $m$ and $X_i\in \XX^{(p)}$ given by
$$\XX^{(p)}:=\{f(\S^{(i)}_{T_j}): |f(x)|\leq K(1+|x|^p) \textrm{ for some }K>0,\textrm{ for } j=1,\ldots, n, i=1,\ldots, d\}.$$
As previously with $\XX$ in \eqref{eq:Xputs}, the above set $\XX^{(p)}$ is large enough to determine uniquely the marginal distributions of $\S^{(i)}_{T_j}$. That is $\MSUPX\neq \emptyset$ implies that there exist unique probability measures $\mu^{(i)}_j$ such that $\LL_{\P}(\S^{(i)}_{T_j})=\mu^{(i)}_j$, $i=1,\ldots, d$, $j=1,\ldots,n$ for any $\P\in \MSUPX=\MMUS$. We write $V^{(p)}_{\XX,\PP,\set}$ for the superreplication cost $V_{\XX,\PP,\set}(G)$ but with $(X,\gamma)\in \AA^{(p)}_{\XX}$ and $\widetilde{V}^{(p)}_{\XX,\PP,\set}$ for the approximative value. We need to assume that $\mu$'s admit $p^{\textrm{th}}$ moment.
\begin{assumption}\label{ass:assumption_on_measure}
Assume $\vec{\mu}=(\mu^{(i)}_j: i=1,\ldots, d,\ j=1,\ldots,n)$ are probability measures on $\R_+$, with mean $1$, admitting finite $p$-th moment for some $p>1$ and $\mu^{(i)}_1\preceq\mu^{(i)}_2\preceq\cdots\preceq\mu^{(i)}_n$, $i=1,\ldots, d$.
\end{assumption}

\begin{theorem}\label{thm: extended_weak_thm}
Let $\vec{\mu}$ satisfy Assumption \ref{ass:assumption_on_measure}, $\set$ be a measurable subset of $\info$ such that for any $\eta>0$, $\MMUE{\eta}\neq\emptyset$. Then, under Assumption \ref{assumption:key assumption}, for any uniformly continuous $G$ that satisfies
\begin{align*}
|G(\S)|\le L(1+\sup_{0\le t\le T_n}|\S_t|^p),
\end{align*}
the following robust pricing-hedging duality holds 
\begin{align*}
\widetilde{V}^{(p)}_{\XX^{(p)},\PP,\set}(G)=\primalls(G),
\end{align*}
where $p$ is the same as in Assumption \ref{ass:assumption_on_measure}.
\end{theorem}

\subsection{Martingale optimal transport duality with exact marginal matching}\label{section: Martingale optimal transport duality for unbounded claims--exact}

Theorems \ref{thm:joint_distribution_weak_main_2} and \ref{thm: extended_weak_thm} extend the duality obtained in \cite{Yan}. In general we obtain an \emph{asymptotic} duality result with the dual and primal problems defined through a limiting procedure. In this section, we want to focus on establishing a duality result without any asymptotic approximation. As already seen from Theorem \ref{thm:joint_distribution_weak_main_2}, in a setting where there is a single marginal and the prediction set is absent, this type of duality result can be obtained without imposing further conditions on the payoff function $G$ other than uniform continuity. However, to achieve this goal in a more general setting, we will impose stricter conditions on the payoff function $G$ and prediction set $\set$. 
\begin{assumption}\label{assumption:2.1}
There exist constants $L>0$ and $p>1$ such that $G$ is uniformly continuous w.r.t. sup norm $\|\cdot\|$ and subject to 
\begin{equation*}
|G(S)|\le L(1+\|S\|^p), \qquad S\in \DD([0,T_n],\R_+^d)
\end{equation*}
Moreover, let $\upsilon, \hat{\upsilon}\in \DD([0,T_n],\R_+^d)$ be of the form 
\begin{align*}
\upsilon_t =& \sum_{i=1}^n\sum_{j=0}^{m_{i}-1}\upsilon_{i,j}\indicators{[t_{i,j}, t_{i,j+1})}(t)+ v_{n, m_n-1}\indicators{T_n}(t),\\
\hat{\upsilon}_t =&\sum_{i=1}^n\sum_{j=1}^{m_{i}-1}\upsilon_{i,j}\indicators{[\hat{t}_{i,j}, \hat{t}_{i,j+1})}(t)+v_{n, m_n-1}\indicators{T_n}(t)
\end{align*}
where $t_{i,0}=\hat{t}_{i,0} = T_i$ $\forall 0\le i\le n-1$, $t_{i,m_i -1}=\hat{t}_{i,m_i-1} = T_i$ $\forall 1\le i\le n$. Then,
\begin{equation}\label{eq: time-continuity of G}
 |G(\upsilon)-G(\hat{\upsilon})|\le L\|\upsilon\|^p\sum^n_{i=1}\sum_{i=1}^{m_i}|\Delta t_{i,j}-\Delta \hat{t}_{i,j}|
\end{equation}
where as usual $\Delta t_{i,j}:=t_{i,j}-t_{i,j-1}$ and $\Delta \hat{t}_{i,j}:=\hat{t}_{i,j}-\hat{t}_{i,j-1}$.
\end{assumption}
Note that Assumption \ref{assumption:2.1} is close in spirit to Assumption 2.1 in \cite{Yan}. Despite their proximity, our assumption here is strictly weaker, which can be seen from the fact that it includes European options having intermediate maturities, in contrast to Assumption 2.1 in \cite{Yan}. 
\begin{definition}\label{def: time invariant set}
We say $\set$ is \textit{time invariant} if for any non-decreasing continuous function $f:[0,T_n]\to [0,T_n]$ such that $f(0)=0$ and $f(T_i) = T_i$ for any $i=1,\ldots,n$, $S\in \set$ implies $(S_{f(t)})_{t\in [0,T_n]}\in \set$. 
\end{definition}

\begin{theorem}\label{theorem: new_section_main}
Let $\vec{\mu}$ satisfy Assumption \ref{ass:assumption_on_measure} and $\set$ be  closed and time invariant. Then, under Assumption \ref{assumption:key assumption}, for any $G$ that satisfies Assumption \ref{assumption:2.1}
the following robust pricing-hedging duality holds 
\begin{equation}
\widetilde{V}^{(p)}_{\vec{\mu}, \set}(G)=V^{(p)}_{\vec{\mu}, \set}(G)=P_{\vec{\mu}, \set}(G)=\widetilde{P}_{\vec{\mu}, \set}(G), \label{eq: new_section_main}
\end{equation}
where $p$ is the same as in Assumption \ref{ass:assumption_on_measure}.
\end{theorem}

\section{First proofs}\label{sec:first_proofs}

We present in this section proof of all the results except Theorem \ref{theorem:Main} which is shown in Sections \ref{section:discretisation of the dual} and \ref{section:discretisation of the primal}. We start by describing a discretisation of a continuous path, often referred to as the ``Lebesgue discretisation" which will often used. In particular, it will be central to Section \ref{section:discretisation of the dual} but is also employed in the proofs of Lemma \ref{thm: third_main}, \ref{thm:continuity in probaility},  \ref{thm: weak_main_new_section} and Theorem \ref{theorem: new_section_main} below.

\begin{definition}\label{defn:stopping time}
For a positive integer $N$ and any $S \in \Omega$, we set $\tau^{(N)}_0(S)=0$ and $m^{(N)}_0(S)=0$, then define
\begin{eqnarray*}
\tau^{(N)}_k(S)=\inf\Big\{t\ge\tau^{(N)}_{k-1}(S):|S_t-S_{\tau^{(N)}_{k-1}(S)}|=\frac{1}{2^N} \Big\}\wedge T
\end{eqnarray*}
and let $m^{(N)}(S)=\min\{k\in \N: \tau^{(N)}_k(S)=T\}$. 
\end{definition}
Following the observation that $m^{(N)}(S)<\infty$ $\forall S\in \Omega$, we say the sequence of stopping times $0=\tau^{(N)}_0<\tau^{(N)}_1<\cdots<\tau^{(N)}_{m^{(N)}}=T$ forms a Lebesgue partition of $[0,T]$ on $\Omega$. Similar partitions were studied previously, see e.g.\ \citet{Vovk:12}. Their main appearances have been as tools to build pathwise version of the \Itos integral. They can also be interpreted, from a financial point of view, as candidate times for rebalancing portfolio holdings, see \citet{whalley1997asymptotic}.
\begin{remark}\label{remark:stopping time}
Note that $m^{(N-2)}(S) \le m^{(N)}(\tilde{S})$ for any $S, \tilde{S}\in \Omega$ such that $\|S-\tilde{S}\|< 2^{-N}$. To justify this, notice that for each $i< m^{(N-2)}(S)$, $\{\tilde{S}_{t}\,:\, t\in (\tau^{(N-2)}_{i-1}(S), \tau^{(N-2)}_{i}(S)]\cap \{k/2^N\,:\, k\in \N_+\}$ has at least three elements, which implies that for each $i< m^{(N-2)}(S)$ there exist at least one $j< m^{(N)}(\tilde{S})$ such that $\tau^{(N)}_j(\tilde{S})\in (\tau^{(N-2)}_{i-1}(S), \tau^{(N-2)}_{i}(S)]$. In consequence, for any sequence $(\P^{(k)})_{k\ge 1}$ converging to $\P$ weakly and bounded non-increasing function $\phi:\N\to \R$
\begin{align}\label{eq:remark_stopping time}
\E_{\P}\Big[\phi\big(m^{(D)}(\S)\big)\Big]\le \liminf_{k\to \infty}\E_{\P^{(k)}}\Big[\phi\big(m^{(D-2)}(\S)\big)\Big].
\end{align}
\end{remark}

\subsection{Proof of Corollary \ref{cor: calculus of variations} and Remark \ref{remark: remark_v_penalty}}

Note that any $(X,\ts)$ that super-replicates $G-N\lambda_{\set}$ also super-replicates $G-1/N$ on $\set^{\frac{1}{N^2}}$. It follows that 
\begin{align*}
\duall^{sp}(G) =& \inf\Big\{\PP(X)\,:\, \exists (X,\ts)\in \AA^{sp}_{\XX} \text{ s.t.\ } \text{$(X,\ts)$ super-replicates $G$ on $\set^{\epsilon}$ for some $\epsilon> 0$}\Big\}\\
\le& \frac{1}{N} + \inf\Big\{\PP(X)\,:\, \exists (X,\ts)\in \AA^{sp}_{\XX} \text{ s.t.\ } \text{$(X,\ts)$ super-replicates $G-N\lambda_{\set}$ on $\info$}\Big\}.
\end{align*}
Since it holds for any $N$, we have
\begin{align*}
\duall^{sp}(G) \le& \inf_{N\ge 0}\inf\Big\{\PP(X)\,:\, \exists (X,\ts)\in \AA^{sp}_{\XX} \text{ s.t.\ } \text{and $(X,\ts)$ super-replicates $G-N\lambda_{\set}$ on $\info$}\Big\}\\
=& \inf_{N\ge 0}\widetilde{V}^{sp}_{\XX,\PP,\info}(G-N\lambda_{\set}).
\end{align*}
Note that by the same argument above we have 
\begin{equation}\label{eq: v_penalty}
\duall(H)\le \inf_{N\ge 0}\widetilde{V}_{\XX,\PP,\info}(H-N\lambda_{\set}) \;\;\text{ and }\;\; \duall^{sp}(H)\le \inf_{N\ge 0}\widetilde{V}^{sp}_{\XX,\PP,\info}(H-N\lambda_{\set})
\end{equation}
hold for every bounded measurable $H$.

Notice that
\begin{align*}
&\inf\Big\{\PP(X)\,:\, \exists (X,\ts)\in \AA^{sp}_{\XX} \text{ s.t.\ } \text{$(X,\ts)$ super-replicates $G-N\lambda_{\set}$ on $\info$}\Big\}\\
=& \inf_{X\in \Lin(\XX)}\big\{\PP(X)+\inf\Big\{x\,:\, \exists \ts\in \AA^{sp} \text{ s.t.\ } \text{$\ts$ super-replicates $G-N\lambda_{\set}-X-x$ on $\info$}\big\}\Big\}\\
=& \inf_{X\in \Lin(\XX)}\big\{\PP(X)+ \dualf^{sp}(G - N\lambda_{\set} - X)\big\}\\
=& \inf_{X\in\Lin(\XX)}\Big\{\PP(X)+\primalf(G-X-N\lambda_{\set})\Big\},
\end{align*}
where the last equality is justified by Theorem \ref{theorem:Main} as $\lambda_{\set}$ and $X$ are bounded and uniformly continuous. Hence, we have
\begin{align*}
\duall^{sp}(G) \le \inf_{N\ge 0,\, X\in\Lin(\XX)}\Big\{\PP(X)+\primalf(G-X-N\lambda_{\set})\Big\}.
\end{align*}
On the other hand, given any $(X,\ts)\in \AA_{\XX}$ and $\epsilon>0$ such that $(X,\ts)$ super-replicates $G$ on $\set^{\epsilon}$, by the admissibility of $(X,\ts)\in \AA_{\XX}$ and boundedness of $X$ and $G$, if $N>0$ is sufficiently large then
\begin{align*}
X(S)+ \int_{0}^{T_n}\ts_u(S)\cdot dS_u\ge G(S)-N,
\end{align*}
and hence $(X,\ts)$ super-replicates $G-N\lambda_{\set}$. It follows that
\begin{align*}
\duall(G) = &\inf\Big\{\PP(X)\,:\, \exists (X,\ts)\in \AA_{\XX} \text{ s.t.\ } \text{$(X,\ts)$ super-replicates $G$ on $\set^{\epsilon}$ for some $\epsilon> 0$}\Big\}\\
\ge &\inf_{N\ge 0}\inf\Big\{\PP(X)\,:\, \exists (X,\ts)\in \AA_{\XX} \text{ s.t.\ } \text{$(X,\ts)$ super-replicates $G-N\lambda_{\set}$ on $\info$}\Big\}\\
=& \inf_{N\ge 0,\, X\in \Lin(\XX)}\big\{\PP(X)+ \dualf(G-X-N\lambda_{\set})\big\}\\
=& \inf_{N\ge 0,\, X\in\Lin(\XX)}\Big\{\PP(X)+\primalf(G-X-N\lambda_{\set})\Big\},
\end{align*}
where the last equality is again justified by Theorem \ref{theorem:Main}.
As $\duall(G)\le \duall^{sp}(G)$, this establishes the equality in \eqref{eq: second_main_2_1}.

Note that by the same argument above we can argue that
$$\duall(H)\ge \inf_{N\ge 0}\widetilde{V}_{\XX,\PP,\info}(H-N\lambda_{\set}) \;\;\text{ and }\;\; \duall^{sp}(H)\ge \inf_{N\ge 0}\widetilde{V}^{sp}_{\XX,\PP,\info}(H-N\lambda_{\set})$$
hold for every bounded measurable $H$. Therefore, combining this with \eqref{eq: v_penalty}, we show \eqref{eq: remark_v_penalty}.

\subsection{Proof of Theorem \ref{thm: second main_2} }

To establish \eqref{eq: approximate weak duality}, we consider
a $(\XX,\ts)\in\AA_{\XX}$ that super-replicates $G$ on $\set^{\epsilon}$ for some $\epsilon>0$, i.e.\
\begin{align*}
X(\S)+\int_0^{T_n}\ts_ud\S_u\ge G(\S) \text{ on } \set^{\epsilon}.
\end{align*}
Since $X$ is bounded, it follows from the definition of admissibility that there exists $M>0$ such that 
\begin{align}
X(\S)+\int_0^{T_n}\ts_ud\S_u\ge G(\S)-M\lambda_{\set}(\S).\label{eq:weak_main_1}
\end{align}
Next, for each $N\ge 1$, we pick $\P^{(N)}\in \MSUPX^{1/N}$ such that 
$$\E_{\P^{(N)}}[G(\S)]\ge \sup_{\P\in \MSUPX^{1/N}}\E_{\P}[G(\S)]-\frac{1}{N}.$$
Since $\ts$ is progressively measurable in the sense of \eqref{eq:pm}, the integral $\int_0^{\cdot} \ts_u(\S)\cdot d\S_u$, defined pathwise via integration by parts, agrees a.s. with the stochastic integral under any $\P^{(N)}$. Then, by \eqref{eq:admissible}, the stochastic integral is a $\P^{(N)}$--super-martingale and hence $\E_{\P^{(N)}}\Big[\int_0^{T_n} \ts_u(\S)\cdot d\S_u\Big]\le 0$. Therefore, from \eqref{eq:weak_main_1}, we can derive that 
\begin{align}
\E_{\P^{(N)}}[X(\S)]\ge \E_{\P^{(N)}}\big[G(\S)-M\lambda_{\set}(\S)\big]\ge \sup_{\P\in \MSUPX^{1/N}}\E_{\P}[G(\S)]-\frac{1}{N}-\frac{M}{N}.\label{eq:weak_main_3}
\end{align}
Also note that $X$ takes the form of $a_0 + \sum_{i=1}^m a_i X_i$. Then by definition of $\MSUPX^{\eta}$ 
\begin{align*}
\big|\PP(X)-\E_{\P^{(N)}}[X(\S)]\big|\to 0 \text{ as } N\to \infty.
\end{align*}
This, together with \eqref{eq:weak_main_3}, yields 
\begin{align*}
\PP(X)\ge \primalxa(G).
\end{align*}
As $(X,\ts)\in\AA_{\XX}$ is arbitrary, we therefore establish \eqref{eq: approximate weak duality}.

To show \eqref{eq: second_main_2_2}, we first deduce from Theorem \ref{theorem:Main} and \eqref{eq: second_main_2_1} that
\begin{equation}
\begin{split}
\duall(G)=& \inf_{X\in \Lin(\XX),\,N\ge 0}\Big\{\primalf(G-X-N\lambda_{\set})+\PP(X)\Big\}\\
= & \lim_{N\to \infty} \inf_{X\in \Lin_N(\XX)}\Big\{\sup_{\P\in \MSUP_{\info}}\E_{\P} [G-X-N\lambda_{\set}]+\PP(X)\Big\}\label{eq:third_main_application1}\\
=& \lim_{N\to \infty}\sup_{\P\in \MSUPX^{1/N}}\E_{\P}[G] \\
=& \primalxa(G),
\end{split}
\end{equation}
where the crucial third equality follows from \eqref{eq: third_main_2_1} in Lemma \ref{thm: third_main} below.

Last, we show that $\MSUPXs^{\eta}\neq \emptyset$ for any $\eta > 0$. By the above and equality between $\primalf =\primalfs$ in Theorem \ref{theorem:Main} we have 
\begin{align*}
&\inf_{X\in \Lin(\XX),\,N\ge 0}\Big\{\sup_{\P\in \underline{\MM}_{\info}}\E_{\P}[G-X-N\lambda_{\set}]+\PP(X)\Big\} \\
=& \inf_{X\in \Lin(\XX),\,N\ge 0}\Big\{\sup_{\P\in \MM_{\info}}\E_{\P}[G-X-N\lambda_{\set}]+\PP(X)\Big\} =  \duall(G) = \primalxa(G).
\end{align*}
Then, taking $G= 0$, as $\MSUPX^{\eta} \neq \emptyset$ for any $\eta > 0$,
\begin{align*}
\inf_{X\in \Lin(\XX),\,N\ge 0}\Big\{\sup_{\P\in \underline{\MM}_{\info}}\E_{\P}[-X-N\lambda_{\set}]+\PP(X)\Big\} = \primalxa(0) = 0.
\end{align*}
Therefore, it follows from the equivalence in Lemma \ref{thm: third_main}, with $\MM_s=\MSUPs$, that $\MSUPXs^{\eta}\neq \emptyset$ for any $\eta > 0$. In addition, by \eqref{eq: third_main_2_2} in Lemma \ref{thm: third_main} below,
\begin{align*}
\duall(G) =& \inf_{X\in \Lin(\XX),\,N\ge 0}\Big\{\sup_{\P\in \underline{\MM}_{\info}}\E_{\P}[G-X-N\lambda_{\set}]+\PP(X)\Big\}\\
 =&  \lim_{N\to\infty}\sup_{\P\in \MSUPXs^{1/N}}\E_{\P}[G] =  \primalxsa(G).
\end{align*}
This completes the proof of Theorem \ref{thm: second main_2}. It remains to argue the following which is stated in a general form and also used in subsequent proofs. 

\begin{lemma}\label{thm: third_main}
Let $\set$ be a measurable subset of $\info$, $\XX$ satisfy Assumption \ref{assumption: XX} and $\MM_s$ be a non-empty convex subset of $\MSUP_{\info}$. Then the following two are equivalent:
\begin{enumerate}
\item[(i)] for any $\eta>0$, $\MSUPX^{\eta}\bigcap \MM_s\neq \emptyset$. 
\item[(ii)] $\inf_{X\in \Lin(\XX),\,N\ge 0}\Big\{\sup_{\P\in \MM_s}\E_{\P}[-X-N\lambda_{\set}]+\PP(X)\Big\} = 0$. 
\end{enumerate}
Further, under (i) or (ii), for any uniformly continuous and bounded $G:\Omega\to \R$ we have: 
\begin{equation}
\inf_{X\in \Lin(\XX),\,N\ge 0}\Big\{\sup_{\P\in \MM_s}\E_{\P}[G-X-N\lambda_{\set}]+\PP(X)\Big\} =\lim_{N\to\infty}\sup_{\P\in \MSUPX^{1/N}\cap \MM_s}\E_{\P}[G].\label{eq: third_main_2_1}
\end{equation}
Moreover, for any $\alpha, \beta\ge 0$ and $D\in\N$. 
\begin{align}\label{eq: third_main_2_2}
\begin{split}
&\inf_{X\in \Lin(\XX),\,N\ge 0}\Big\{\sup_{\P\in \MM_s}\E_{\P}[G(\S)-\alpha\wedge(\beta\sqrt{m^{(D)}(\S)})-X(\S)-N\lambda_{\set}(\S)]+\PP(X)\Big\}\\ 
\le&\,\lim_{N\to\infty}\sup_{\P\in \MSUPX^{1/N}\cap \MM_s}\E_{\P}[G(\S)-\alpha\wedge(\beta\sqrt{m^{(D-2)}(\S)})], 
\end{split}
\end{align}
where $m^{(D)}$ is defined in Definition \ref{defn:stopping time}.
\end{lemma}
\begin{proof}
Choose $\kappa>2\vee (\|G\|_{\infty}+\alpha)$. We first observe that
\begin{align}
&\inf_{X\in \Lin(\XX),\,N\ge 0}\Big\{\sup_{\P\in \MM_s}\E_{\P}[G-\alpha\wedge(\beta\sqrt{m^{(D)}})-X-N\lambda_{\set}]+\PP(X)\Big\}\nonumber \\
= & \lim_{N\to \infty} \inf_{X\in \Lin_N(\XX)}\Big\{\sup_{\P\in \MM_s}\E_{\P} [G-\alpha\wedge(\beta\sqrt{m^{(D)}})-X-N\lambda_{\set}]+\PP(X)\Big\}.\label{eq: third_main_0}
\end{align} 
Define the function $\GG: \Lin_N(\XX)\times \MM_s\to \R$ by
\begin{align*}
\GG(X, \P):=&\lim_{\epsilon\searrow 0}\inf_{\tilde{\P}\in \MSUPs_{\info},\, d_p(\tilde{\P},\P)<\epsilon}\E_{\tilde{\P}}\Big[G(\S)-\alpha\wedge(\beta\sqrt{m^{(D-2)}})-x(\S)-N\lambda_{\set}(\S)\Big]+\PP(X)\\
=& \lim_{\epsilon\searrow 0}\inf_{\tilde{\P}\in \MSUPs_{\info},\, d_p(\tilde{\P},\P)<\epsilon}\E_{\tilde{\P}}\Big[-\alpha\wedge(\beta\sqrt{m^{(D-2)}})\Big]+\E_{\P}[G-N\lambda_{\set}-X]+\PP(X).
\end{align*}
Then by \eqref{eq:remark_stopping time} in Remark \ref{remark:stopping time}, for any sequence $(\P^{(k)})_{k\ge 1}$ converging to $\P$ weakly,
\begin{align*}
\E_{\P}\Big[-\alpha\wedge(\beta\sqrt{m^{(D)}(\S)})\Big]\le \liminf_{k\to \infty}\E_{\P^{(k)}}\Big[-\alpha\wedge(\beta\sqrt{m^{(D-2)}(\S)})\Big].
\end{align*}
and hence
\begin{align*}
&\lim_{N\to \infty} \inf_{X\in \Lin_N(\XX)}\Big\{\sup_{\P\in \MM_s}\E_{\P} [G-\alpha\wedge(\beta\sqrt{m^{(D)}})-X-N\lambda_{\set}]+\PP(X)\Big\}\\
\le&\, \lim_{N\to \infty} \inf_{X\in \Lin_N(\XX)}\sup_{\P\in \MM_s}\E_{\P}[\GG(X, \P)],
\end{align*}
with equality when $\alpha = \beta = 0$. 

The next step is to interchange the order of the infimum and supremum. Notice that when we fix $\P$, $\GG$ is affine in the first variable and continuous due to bounded convergence theorem. In addition, by definition $\GG$ is lower-semi continuous in the second variable. Furthermore, $\GG$ is convex in the second variable. To justify this, we notice that $\P\mapsto \E_{\P}\Big[-\alpha\wedge(\beta\sqrt{m^{(D-2)}(\S)})\Big]$ is a linear functional and it follows that for each $\epsilon>0$ and $\lambda\in [0,1]$
\begin{align*}
&\inf_{\tilde{\P}\in \MM_s,\, d_p(\tilde{\P},\lambda\P^{(1)}+(1-\lambda)\P^{(2)})<\epsilon}\E_{\tilde{\P}}\Big[-\alpha\wedge(\beta\sqrt{m^{(D-2)}(\S)})\Big]\\
\le& \lambda\inf_{\tilde{\P}\in \MM_s,\, d_p(\tilde{\P},\P^{(1)})<\epsilon}\E_{\tilde{\P}}\Big[-\alpha\wedge(\beta\sqrt{m^{(D-2)}(\S)})\Big]+(1-\lambda)\inf_{\tilde{\P}\in \MM_s,\, d_p(\tilde{\P},\P^{(2)})<\epsilon}\E_{\tilde{\P}}\Big[-\alpha\wedge(\beta\sqrt{m^{(D-2)}(\S)})\Big].
\end{align*}
Since $\Lin_N(\XX)$ is convex and compact, it follows that we can now apply Min-Max Theorem (see Corollary 2 in \citet{min_max_terkelsen1972}) to $\GG$ and derive
\begin{align*}
\lim_{N\to \infty} \inf_{X\in \Lin_N(\XX)}\sup_{\P\in \MM_s}\E_{\P}[\GG(X, \P)] = \lim_{N\to \infty}\sup_{\P\in \MM_s} \inf_{X\in \Lin_N(\XX)}\E_{\P}[\GG(X, \P)].
\end{align*}
Therefore, we have
\begin{align}\label{eq: third_main_4}
\begin{split}
&\lim_{N\to \infty} \inf_{X\in \Lin_N(\XX)}\Big\{\sup_{\P\in \MM_s}\E_{\P} [G-\alpha\wedge(\beta\sqrt{m^{(D)}})-X-N\lambda_{\set}]+\PP(X)\Big\}\\
\le&\, \lim_{N\to \infty}\sup_{\P\in \MM_s}\inf_{X\in \Lin_N(\XX)}\Big\{\E_{\P} [G-\alpha\wedge(\beta\sqrt{m^{(D-2)}})-X-N\lambda_{\set}]+\PP(X)\Big\},
\end{split}
\end{align}
with equality when $\alpha = \beta = 0$. 

Now first consider the case: $\alpha = \beta =0$. In this case, it follows from above that
\begin{align}
& \lim_{N\to \infty} \inf_{X\in \Lin_N(\XX)}\Big\{\sup_{\P\in  \MM_s}\E_{\P} [G-X-N\lambda_{\set}]+\PP(X)\Big\}\nonumber\\
= & \lim_{N\to \infty}\sup_{\P\in \MM_s}\Big\{\inf_{X\in \Lin_N(\XX)}\E_{\P} [G-X-N\lambda_{\set}]+\PP(X)\Big\}.\label{eq: third_main_1}
\end{align}
Suppose that for any $\eta>0$, $\MSUPX^{\eta}\bigcap \MM_s\neq \emptyset$. Then we see that for any $N$, $\P\in \MSUPX^{1/N^2}$ and $X\in \Lin_N(\XX)$,
\begin{align*}
|\E_{\P} [X] - \PP(X)| \leq \sum_{i=1}^m|a_i|\big|\E_{\P}[X_i] - \PP(X_i)\big|\le \frac{N}{N^2} = \frac{1}{N},
\end{align*}
where $X$ takes the form of $a_0 + \sum_{i=1}^m a_i X_i$, for some $m\in \N$, $X_i\in \XX$ and $a_i\in \R$ such that $\sum_{i=0}^m|a_i| \le N$. In addition, $\lambda_{\set} \le \frac{1}{N^2}\indicator{\S\in \set^{1/N^2}} + \indicator{\S\notin \set^{1/N^2}}$ leads to
$$ \E_{\P} [N\lambda_{\set}] \le \frac{1}{N}\P(\S \in \set^{1/N^2})+N\P(\S\notin \set^{1/N^2}) \le \frac{2}{N}.$$
Therefore, we can deduce that
\begin{align}
& \lim_{N\to \infty}\sup_{\P\in \MSUPX^{1/N^2}\cap \MM_s}\Big\{\inf_{X\in \Lin_N(\XX)}\E_{\P} [G-X-N\lambda_{\set}]+\PP(X)\Big\} \nonumber\\
\ge& \lim_{N\to \infty}\sup_{\P\in \MSUPX^{1/N^2}\cap \MM_s}\Big\{\E_{\P}[G] - \frac{3}{N}\Big\} = \lim_{N\to\infty}\sup_{\P\in \MSUPX^{1/N}\cap \MM_s}\E_{\P}[G]\label{eq: third_main_2}.
\end{align}
Consequently, by taking $G=0$, using \eqref{eq: third_main_0}--\eqref{eq: third_main_2} and noting that considering a $\sup$ over a larger set increases its value, we have
\begin{align*}
\inf_{X\in \Lin(\XX),\,N\ge 0}\Big\{\sup_{\P\in \MM_s}\E_{\P}[-X-N\lambda_{\set}]+\PP(X)\Big\} \ge 0, 
\end{align*}
which leads to 
\begin{align*}
\inf_{X\in \Lin(\XX),\,N\ge 0}\Big\{\sup_{\P\in \MM_s}\E_{\P}[-X-N\lambda_{\set}]+\PP(X)\Big\} = 0.
\end{align*}

On the other hand, if 
\begin{align}
\inf_{X\in \Lin(\XX),\,N\ge 0}\Big\{\sup_{\P\in \MM_s}\E_{\P}[-X-N\lambda_{\set}]+\PP(X)\Big\} = 0,\label{eq: third_main_3}
\end{align}
then we will argue in the following that in the $\sup$ term of \eqref{eq: third_main_1} it suffices to consider probability measures $\P \in \MSUPX^{2\kappa/N}\cap \MM_s$. Suppose $\P\in  (\MM_s\setminus\MSUPX^{2\kappa/N})$, then either there exist $X\in \XX$ such that $\E_{\P}[X] -\PP(X)> 2\kappa/N$ or $\P(\S\notin\set^{2\kappa/N})\ge 2\kappa/N$. In the former case, since $NX\in \Lin(\XX)$,
\begin{align*}
\E_{\P}[G - NX -N\lambda_{\set}]+ \PP(NX) \le \E_{\P}[G] - N(\E_{\P}[X] -\PP(X)) < -\kappa,
\end{align*}
and in the latter case, $\E_{\P}[G - N\lambda_{\set}] < \kappa - 2\kappa = -\kappa$, while
\begin{align*}
&\lim_{N\to \infty}\sup_{\P\in \MM_s}\Big\{\inf_{X\in \Lin_N(\XX)}\E_{\P} [G-X-N\lambda_{\set}]+\PP(X)\Big\}\\
=& \inf_{X\in \Lin(\XX),\,N\ge 0}\Big\{\sup_{\P\in \MM_s}\E_{\P}[G-X-N\lambda_{\set}]+\PP(X)\Big\}\\
\ge& \inf_{X\in \Lin(\XX),\,N\ge 0}\Big\{\sup_{\P\in \MM_s}\E_{\P}[-\kappa-X-N\lambda_{\set}]+\PP(X)\Big\} = -\kappa,
\end{align*}
where the last equality follows from \eqref{eq: third_main_1}.
This argument also implies that $\MSUPX^{2\kappa/N} \cap \MM_s \neq \emptyset$ for any $N\in \N$. Therefore we have the equivalence between  
$$\forall \eta>0\ \MSUPX^{\eta}\bigcap \MM_s\neq \emptyset 
\;\text{ and }\;\inf_{X\in \Lin(\XX),\,N\ge 0}\Big\{\sup_{\P\in \MM_s}\E_{\P}(-X-N\lambda_{\set})+\PP(X)\Big\} = 0.$$ 

Now consider the general case: $\alpha, \beta\ge 0$. We begin to verify \eqref{eq: third_main_2_1} and \eqref{eq: third_main_2_2}. Since $\MSUPX^{\eta}\bigcap \MM_s\neq \emptyset$ $\forall \eta>0$, 
\begin{align*}
\sup_{\P\in \MM_s}\E_{\P}[X-N\lambda_{\set}]-\PP(X)\ge 0, \quad \forall\, X\in \XX,\ N\in \R_+. 
\end{align*}
Hence for every $N$ and $X\in \Lin_N(X)$ 
\begin{align*}
&\sup_{\P\in \MM_s}\E_{\P} [G-\alpha\wedge(\beta\sqrt{m^{(D)}})-X-N\lambda_{\set}]+\PP(X)\\
\ge&\, -\|G\|_{\infty} -\alpha + \sup_{\P\in \MM_s}\E_{\P}[X-N\lambda_{\set}]-\PP(X)\ge -\kappa,
\end{align*}
and therefore
\begin{align*}
\lim_{N\to \infty} \inf_{X\in \Lin_N(\XX)}\Big\{\sup_{\P\in \MM_s}\E_{\P} [G-\alpha\wedge(\beta\sqrt{m^{(D)}})-X-N\lambda_{\set}]+\PP(X)\Big\}\ge -\kappa.
\end{align*}
Then, by using the same argument as above, we can argue that in the $\sup$ term of \eqref{eq: third_main_4} it suffices to consider probability measures $\P \in \MSUPX^{2\kappa/N}\cap \MM_s$ and hence we have
\begin{align*}
&\lim_{N\to \infty} \inf_{X\in \Lin_N(\XX)}\Big\{\sup_{\P\in \MM_s}\E_{\P} [G-\alpha\wedge(\beta\sqrt{m^{(D)}})-X-N\lambda_{\set}]+\PP(X)\Big\}\\
\le & \lim_{N\to \infty} \inf_{X\in \Lin_N(\XX)}\Big\{\sup_{\P\in \MSUPX^{2\kappa/N}\cap \MM_s}\E_{\P} [G-\alpha\wedge(\beta\sqrt{m^{(D-2)}})-X]+\PP(X)\Big\}\\
\le & \lim_{N\to \infty}\sup_{\P\in \MSUPX^{2\kappa/N}\cap \MM_s}\E_{\P} [G-\alpha\wedge(\beta\sqrt{m^{(D-2)}})],
\end{align*}
where the second inequality follows from the fact that $-X\in \Lin_N(\XX)$ for every $X\in \Lin_N(\XX)$. This completes the verification of \eqref{eq: third_main_2_2}. In the case that $\alpha = \beta = 0$, combining the inequality above with \eqref{eq: third_main_2}, we then conclude that 
\begin{align*}
\inf_{X\in \Lin(\XX),\,N\ge 0}\Big\{\sup_{\P\in \MM_s}\E_{\P}[G-X-N\lambda_{\set}]+\PP(X)\Big\} =\lim_{N\to\infty}\sup_{\P\in \MSUPX^{1/N}\cap \MM_s}\E_{\P}[G].
\end{align*}
\end{proof}

\subsection{Proof of Theorem \ref{thm: countable puts}}

From Theorem \ref{thm: second main_2}, as worked out in Example \ref{example: finite options}, we know that
\begin{align}
V_{\XX,\PP,\info}(G) = \widetilde{P}_{\XX,\PP,\info}(G) = \lim_{N\to\infty}\sup_{\P\in \MSUP^{1/N}_{\XX,\PP,\info}}\E_{\P}[G].\label{eq:hp_pf_39_1}
\end{align}
Now for every positive integer $N$, we pick $\P^{(N)}\in \MSUP^{1/N}_{\XX,\PP,\info}$ such that
$$ \E_{\P^{(N)}}[G] + 1/N \ge \sup_{\P\in \MSUP^{1/N}_{\XX,\PP,\info}}\E_{\P}[G].$$
We write
\begin{align*}
p^{(N)}_{k,i,j} := \E_{\P^{(N)}}[(K^{(i)}_{k,j} - S^{(i)}_{k,j})^+] 
\end{align*}
for any $i=1,\ldots d$, $j=1,\ldots,n$, $k=1,\ldots,m(i,j)$, and define $\tilde{p}^{(N)}_{k,i,j}$'s by
\begin{align*}
\tilde{p}^{(N)}_{k,i,j} = \sqrt{N}\big(p_{k,i,j} - (1- 1/\sqrt{N})p^{(N)}_{k,i,j}\big).
\end{align*}
Note that
\begin{align}
|\tilde{p}^{(N)}_{k,i,j} - p_{k,i,j}| = (\sqrt{N} -1)|p_{k,i,j} - p^{(N)}_{k,i,j}|\le \frac{\sqrt{N}}{N} =\frac{1}{\sqrt{N}} \quad \forall i,j,k. \label{eq: finite_puts_1}
\end{align}
Then, it follows from Assumption \ref{assump: countable puts} that when $N$ is large enough there exists a $\tilde{\P}^{(N)}\in\MSUP_{\info}$ such that
\begin{align*}
\tilde{p}^{(N)}_{k,i,j} := \E_{\tilde{\P}^{(N)}}[(K^{(i)}_{k,j} - S^{(i)}_{k,j})^+] \;\quad \forall i,j,k.
\end{align*} 
Now we consider $\Q := (1-1/\sqrt{N})\P^{(N)} + \tilde{\P}^{(N)}/\sqrt{N}$. It follows that
\begin{align*}
\E_{\Q}[(K^{(i)}_{k,j} - S^{(i)}_{k,j})^+] =& (1-1/\sqrt{N})\E_{\P^{(N)}}[(K^{(i)}_{k,j} - S^{(i)}_{k,j})^+]+\frac{1}{\sqrt{N}}\E_{\tilde{\P}^{(N)}}[(K^{(i)}_{k,j} - S^{(i)}_{k,j})^+]\\
=& (1-1/\sqrt{N})p^{(N)}_{k,i,j} + \tilde{p}^{(N)}_{k,i,j}/\sqrt{N} = p_{k,i,j}
\end{align*}
and hence $\Q\in \MSUP_{\XX,\PP,\info}$. 
In addition, 
\begin{align*}
\big|\E_{\Q}[G] - \E_{\P^{(N)}}[G]\big| \le  \frac{1}{\sqrt{N}}(\E_{\P^{(N)}}[|G|]+\E_{\tilde{\P}^{(N)}}[|G|]) \le \frac{2\|G\|_{\infty}}{\sqrt{N}}.
\end{align*} 

Therefore, we have
\begin{align*}
\sup_{\P\in \MSUP^{1/N}_{\XX,\PP,\info}}\E_{\P}[G] \le \sup_{\P\in \MSUP_{\XX,\PP,\info}}\E_{\P}[G] - \frac{2\|G\|_{\infty}}{\sqrt{N}} -\frac{1}{N},
\end{align*}
which leads us to conclude
\begin{align*}
\widetilde{P}_{\XX,\PP,\info}(G)=\lim_{N\to\infty}\sup_{\P\in \MSUP^{1/N}_{\XX,\PP,\info}}\E_{\P}[G] \le P_{\XX,\PP,\info}(G).
\end{align*}
Together with \eqref{eq:hp_pf_39_1} and \eqref{eq:primallessthandual} this completes the proof.

\subsection{Proof of Theorem \ref{thm:joint_distribution_weak_main_2}}
Let 
$$\YY = \{f(\S^{(1)}_{T},\ldots, \S^{(d)}_{T})\,:\,f\in \CC(\R_+^d,\R) \text{ s.t.\ } \sup_{\vec{x}\neq \vec{y}}\frac{|f(\vec{x})-f(\vec{y})|}{|\vec{x}-\vec{y}|}\le 1,\,\,\|f\|_{\infty}\le 1\}.$$
Then, as $\Lin(\XX)$ is dense in $\Lin(\YY)$,
\begin{align}
V_{\XX,\PP,\info}(G) = V_{\YY, \PP, \info}(G). \label{eq: jd_0}
\end{align}

We now consider $\YY_M := \{f(\S^{(1)}_{T}\wedge M, \ldots, \S^{(d)}_{T}\wedge M)\,:\, f\in \YY\}$. $\YY_M$ is a subset of $\YY$ for each $M\in \N$, and in consequence, 
\begin{align}
V_{\YY,\PP,\info}(G) \le& V_{\YY_M,\PP,\info}(G), \quad \forall\, M\in \N.\label{eq: jd_1}
\end{align}
Observe that, for any $X\in \XX$,
$$X(\S^{(1)}_{T}\wedge M,\ldots,\S^{(d)}_{T}\wedge M) \le X(\S^{(1)}_{T},\ldots,\S^{(d)}_{T}) + \frac{2}{M-1}\sum_{i=1}^d \S^{(i)}_T$$
and hence
\begin{align*}
\Big|\int_{\R_+^d}X(s_1\wedge M,\ldots,s_d\wedge M)\pi(\td s_1,\ldots,\td s_d)- \int_{\R_+^d}X(s_1,\ldots,s_d)\pi(\td s_1,\ldots,\td s_d)\Big|\le \frac{2d}{M}.
\end{align*}

Note that by definition $\YY_M$ is closed and convex. Also, by \AzAs theorem, $\YY_M$ is compact. Hence $\Lin_1(\YY_M)$ satisfies Assumption \ref{assumption: XX}. Therefore, applying Theorem \ref{thm: second main_2} to $\YY_M$, we have
\begin{align}
V_{\YY_M,\PP,\info}(G) = \widetilde{\underline{P}}_{\YY_M,\PP,\info}(G) \quad \forall M > 0. \label{eq: jd_3}
\end{align}
Then, by putting $\eqref{eq: jd_0}$, $\eqref{eq: jd_1}$ and $\eqref{eq: jd_3}$ together
\begin{align*}
V_{\XX,\PP, \info}(G) \le& \lim_{M\to\infty}\widetilde{\underline{P}}_{\YY_M,\PP,\info}(G)
= \lim_{M\to \infty}\lim_{N\to \infty}\sup_{\P\in\MSUPs^{1/N}_{\YY_M,\PP,\info}}\E_{\P}[G] \le \lim_{M\to \infty}\sup_{\P\in\MSUPs^{1/M}_{\YY_M,\PP,\info}}\E_{\P}[G].
\end{align*}
For every $N \in \N$, take $\P^{(N)}\in \MSUPs^{1/N}_{\YY_N,\PP,\info}$ such that
\begin{align*}
\sup_{\P\in\MSUPs^{1/N}_{\YY_N,\PP,\info}}\E_{\P}[G] \le \E_{\P^{(N)}}[G] + \frac{1}{N}.
\end{align*}
Let $\pi^{(N)}$ be the law of $(\S^{(1)}_{T},\ldots,\S^{(d)}_{T})$ under $\P^{(N)}$. It is a probability measure on $\R^d_+$ with mean equal to $1$. It follows that the family $\{\pi^{(N)}\}_{N\ge 1}$ is tight. By Prokhorov theorem, there exists $\{\pi^{(N_k)}\}_{k\ge 1}$, a subsequence of $\{\pi^{(N)}\}_{N\ge 1}$, converging to some $\tilde{\pi}$. In the following, we are going to argue that $\tilde{\pi}$ is in fact $\pi$. For any $X\in \YY$ and $N\in \N$,
\begin{align*}
&\,\Big|\E_{\P^{(N)}}[X(\S_T)] - \int_{\R_+^d}X(s_1,\ldots,s_d)\pi(\td s_1,\ldots,\td s_d)\Big|\\
\le &\, \Big|\E_{\P^{(N)}}[X(\S^{(1)}_T\wedge N,\ldots,\S^{(d)}_T\wedge N)] - \int_{\R_+^d}X(s_1\wedge N,\ldots,s_d\wedge N)\pi(\td s_1,\ldots,\td s_d)\Big|\\
&\, + 2\E_{\P^{(N)}}\big[\indicator{\max_{1\le d \le n}\{\S^{(i)}_T > N\}}\big] + 2\int_{\R_+^d}\indicator{\max_{1\le d \le n}\{s_i > N\}}\pi(\td s_1,\ldots,\td s_d)\\
\le &\, \frac{1}{N} + \frac{2d}{N} + \frac{2d}{N} = \frac{4d+1}{N}.
\end{align*}
By weak convergence of $\pi^{(N)}$, along a subsequence of $\{\pi^{(N)}\}_{N\ge 1}$, for every $X\in \YY$
\begin{align*}
\int_{\R_+^d}X(s_1,\ldots,s_d)\pi^{(N)}(\td s_1,\ldots,\td s_d) \to \int_{\R_+^d}X(s_1,\ldots,s_d)\tilde{\pi}(\td s_1,\ldots,\td s_d) \;\;\text{ as } N \to \infty.
\end{align*} 
Therefore, for every $X\in \YY$
\begin{align*}
\int_{\R_+^d}X(s_1,\ldots,s_d)\pi(\td s_1,\ldots,\td s_d)= \int_{\R_+^d}X(s_1,\ldots,s_d)\tilde{\pi}(\td s_1,\ldots,\td s_d),
\end{align*}
which implies that $\pi = \tilde{\pi}$ as $\YY$ is rich enough to guarantee uniqueness of $\pi$.

It follows that
\begin{align*}
V_{\XX,\PP,\info}(G) \le \lim_{N\to \infty}\sup_{\P\in\MSUPs^{1/N}_{\YY_N,\PP,\info}}\E_{\P}[G]
\le \limsup_{k\to\infty}\sup_{\P\in\MSUPs_{\pi^{(N_k)},\info}}\E_{\P}[G] \le \underline{P}_{\pi,\info}(G),
\end{align*}
where the last inequality follows from the following lemma.

\begin{lemma}\label{thm:continuity in probaility}
Assume $\pi^{(N)}$ and $\pi$ are probability measures on $\R^d_+$ such that $\pi^{(N)}$ and $\pi$ satisfies \eqref{eq:assumption_on_pi} and $\pi^{(N)}$ converges to $\pi$ weakly. Then, for any bounded and uniformly continuous $G$, $\alpha, \beta\ge 0$ and $D\in\N$. 
\begin{align*}
\limsup_{N\to \infty}\underline{P}_{\pi^{(N)},\info}\big(G(\S)-\alpha\wedge(\beta\sqrt{m^{(D)}(\S)})\big)\le \underline{P}_{\pi,\info}\big(G(\S)-\alpha\wedge(\beta\sqrt{m^{(D-2)}(\S)})\big),
\end{align*}
where $m^{(D)}$ is defined in Definition \ref{defn:stopping time}.
\end{lemma}
\begin{proof}
Choose $f_e:\R_+\to \R_+$ such that $|G(\omega)-G(\upsilon)|\le f_e(\|\omega-\upsilon\|)$ for any $\omega, \upsilon\in \Omega$, $|X^{(c)}_i(\vec{x})- X^{(c)}_i(\vec{y})|\le f_e(|\vec{x}-\vec{y}|)$ for any $\vec{x}, \vec{y}\in \R_+^d$ and $\lim_{x\searrow 0}f_e(x)=0$. Now fix $N$ and $\P^{(N)}\in \underline{\MM}_{\pi^{(N)},\info}$. By definition of $\underline{\MM}$, there exists a complete probability space $(\Omega^{W},\FF^{W}_{T},\F^{W}, P^{W})$ together a finite dimensional Brownian motion $(W_t)_{t\geq 0}$ and the natural filtration $\FF^{W}_t=\sigma\{W_s|s\le t\}$, and a continuous martingale $M$ defined on $(\Omega^{W},\FF^{W}_{T},\F^{W},P^{W})$ such that $\P^{(N)}=P^{W}\circ M^{-1} $. 

Write $\epsilon_N:=d_p(\pi^{(N)},\pi)$. $\pi^{(N)}$ converges to $\pi$ weakly is equivalent to saying that $\epsilon_N\to 0$ as $N\to \infty$. Fix $N$. If $\epsilon_N=0$, then it is trivially true that $\underline{P}_{\pi^{(N)},\info}(G)= \underline{P}_{\pi,\info}(G)$. Therefore, we only consider the case that $\epsilon_N>0$. By Strassen's theorem, Corollary of Theorem 11 on page 438 in \citet{Strassen} or theorem 4 on page 358 in \citet{probability}, we can find a $\FF^{W}_{T}$ measurable random variable $\Lambda$ such that $\Lambda^{(d+i)} = X^{(c)}_i(\Lambda^{(1)},\ldots,\Lambda^{(1)})$ for every $i\le K$,
\begin{equation}\label{eq:cor}
(\Lambda^{(1)},\ldots,\Lambda^{(d)})\sim_{P^{W}}\pi\quad \text{and} \quad P^{W}(|\Lambda^{(i)}-M^{(i)}_{T}|>2\epsilon_N)<2\epsilon_N \quad \forall i\le d.
\end{equation}  

We now construct a continuous martingale from $\Lambda$ by taking conditional expectation, i.e.\
\begin{equation*}
\Gamma_t=E^{W}[\Lambda|\FF_t^W], \quad t\in[0,T],
\end{equation*}
where $E^{W}$ is the expectation with respect to $P^W$.
Note that by uniform continuity of $X^{(c)}_i$
\begin{align*}
|\Lambda^{(d+i)}-M^{(d+i)}_{T}|\le f_e(2\epsilon_N) \quad\forall i\le K,  \;\;\text{ whenever } |\Lambda^{(j)}-M^{(j)}_{T}|\le 2\epsilon_N \quad \forall j\le d
\end{align*}
Hence, for every $i\le K$
\begin{equation}\label{eq:cor_2}
P^{W}\big(|\Lambda^{(i+d)}-M^{(i+d)}_{T}|>f_e(2\epsilon_N)\big)\le 2d\epsilon_N.
\end{equation}
Observe that $E^W[\Lambda^{(i)}]=E^W[M_{T}^{(i)}]=1$ and $\Lambda^{(i)}\geq 0$ $P^W$-a.s.\ $\forall\, i$. Then, using \eqref{eq:cor}, 
\begin{align*}
E^W[|\Lambda^{(i)}-M^{(i)}_{T}|]=&\;2E^W[(\Lambda^{(i)}-M^{(i)}_{T})^+]-E^W[\Lambda^{(i)}-M^{(i)}_{T}]\\
=&\; 2E^W[(\Lambda^{(i)}-M^{(i)}_{T})^+]\\
\le&\;  4\epsilon_N+2E^W[\Lambda^{(i)}\indicators{\{|\Lambda^{(i)}-M^{(i)}_{T}|>2\epsilon_N\}}]\\
\le&\; 4\epsilon_N+2E^W[\Lambda^{(i)}\indicators{\{|\Lambda^{(i)}-M^{(i)}_{T}|>2\epsilon_N\}}\indicators{\{\Lambda>1/\sqrt{\epsilon_N}\}}]+4\sqrt{\epsilon_N}\\
\le&\; 4\epsilon_N+2\int_{\{x_i\ge \frac{1}{\sqrt{\epsilon_N}}\}\cap \R_+^d} x_i  \pi(\td x_1,\ldots,\td x_d)+4\sqrt{\epsilon_N}, \quad \forall i=1,\ldots,d.
\end{align*}
Similarly, for every $i\le K$,
\begin{align*}
E^W[|\Lambda^{(d+i)}-M^{(d+i)}_{T}|] =& 2E^W[(\Lambda^{(d+i)}-M^{(d+i)}_{T})^+]\\
\le& 2f_e(2\epsilon_N) + 2E^W[\Lambda^{(i)}\indicators{\{|\Lambda^{(d+i)}-M^{(d+i)}_{T}|>f_e(2\epsilon_N)\}}]\\
\le& 2f_e(2\epsilon_N) + 4d\frac{\|X^{(c)}_i\|_{\infty}}{\PP(X^{(c)}_i)}\epsilon_N.
\end{align*}
Now define $\eta_N$ by
\begin{align*}
\eta_N = 2f_e(2\epsilon_N)+4\epsilon_N+4d\sum_{i=1}^K\frac{\|X^{(c)}_i\|_{\infty}}{\PP(X^{(c)}_i)}\epsilon_N+ 2\sum_{i=1}^d\int_{\{x_i\ge \frac{1}{\sqrt{\epsilon_N}}\}\cap \R_+^d} x_i  \pi(\td x_1,\ldots,\td x_d)+4\sqrt{\epsilon_N}
\end{align*}
and note that $\eta_N \to 0$ as $N\to \infty$. Then by Doob's martingale inequality
\begin{align}
P^W(\|\Gamma-M\|\ge \eta_{N}^{1/2})\le&\; \eta_{N}^{-1/2}\sum_{i=1}^{d+K}E^W[|\Lambda^{(i)}-M^{(i)}_{T}|]\le (d+K)\eta_{N}^{1/2}.\label{eq:lifted_close_1}
\end{align}
It follows that
\begin{align*}
\big|E^W[G(\Gamma)-G(M)]\big|\le&\, 2(d+K)\|G\|_{\infty}\eta_{N}^{1/2}+E^W\big[\big|G(\Gamma)-G(M)\big|\indicators{\{\|\Gamma-M\|< \eta_{N}^{1/2}\}}\big]\\
\le&\,2(d+K)\|G\|_{\infty}\eta_{N}^{1/2}+f_e(\eta_{N}^{1/2}).
\end{align*} 
Note that by \eqref{eq:remark_stopping time} in Remark \ref{remark:stopping time} for $N$ sufficiently large,
\begin{align*}
E^W[\alpha\wedge(\beta\sqrt{m^{(D)}(M)})]\ge E^W[\alpha\wedge(\beta\sqrt{m^{(D-2)}(\Gamma)})] - \alpha\eta_{N}^{1/2}.
\end{align*}
As $\P^{(N)}\in \underline{\MM}_{\pi^{(N)},\info}$ is arbitrary, 
\begin{align*}
\underline{P}_{\pi^{(N)},\info}\big(G-\alpha\wedge(\beta\sqrt{m^{(D)}})\big)\le \underline{P}_{\pi,\info}\big(G-\alpha\wedge(\beta\sqrt{m^{(D-2)}})\big)-\Big((d+K)\|G\|_{\infty}\eta_{N}^{1/2}+f_e(\eta_{N}^{1/2})+\alpha\eta_{N}^{1/2}\Big).
\end{align*}

Therefore, we can conclude that
$$\limsup_{N\to \infty}\underline{P}_{\pi^{(N)},\info}\big(G-\alpha\wedge(\beta\sqrt{m^{(D)}})\big)\le \underline{P}_{\pi,\info}\big(G-\alpha\wedge(\beta\sqrt{m^{(D-2)}})\big),\quad \text{ as required.}$$
 
\end{proof}

\subsection{Proof of Theorem \ref{thm: weak_main}}
From Theorem \ref{thm: second main_2}, we know that 
\begin{align*}
\duall(G)\ge\primalxa(G).
\end{align*}
We also have the observation that $\primalls(G)\leq \primalxa(G)$. Then to establish Theorem~\ref{thm: weak_main}, it suffices to show that
\begin{align} 
\duall(G)\le \primalls(G).\label{eq:weak_thm 2}
\end{align}
This follows as a special case ($\alpha=\beta=0$) of the following crucial lemma which also be used to prove Theorem \ref{theorem: new_section_main} below. 
\begin{lemma}\label{thm: weak_main_new_section}
Let $\set$ be a measurable subset of $\info$, $\XX$ be given by \eqref{eq:Xputs} and $\PP$ be such that, for any $\eta>0$, $\MMUE{\eta}\neq\emptyset$, where $\vec{\mu}$ is defined via \eqref{eq:def_mus}. Then for any uniformly continuous and bounded $G$ and $\alpha,\beta\ge 0$
\begin{align*}
 \duall\Big(G-\alpha\wedge(\beta\sqrt{m^{(D)}})\Big) \le \widetilde{P}_{\vec{\mu}, \set}\Big(G-\alpha\wedge(\beta\sqrt{m^{(D-8)}})\Big). 
\end{align*}
where $m^{(D)}$ is defined in Definition \ref{defn:stopping time}.
\end{lemma}
\begin{proof}
Recall that 
\begin{align*}
\Lips[N]:=\Big\{&f\in \CC(\R_+^d,\R)\,:\,\sup_{\vec{x}\neq \vec{y}}\frac{|f(\vec{x})-f(\vec{y})|}{|\vec{x}-\vec{y}|}\le N,\,\,\|f\|_{\infty}\le N,\\
&\hspace{1cm}\text{ and }f(x_1,\ldots,x_d)=f(x_1\wedge N^2,\ldots,x_d\wedge N^2)\,\,\forall (x_1,\ldots,x_d)\in \R^d_+ \Big\}
\end{align*}
and $\Lips=\displaystyle \cup_{N>0}\Lips[N]$. 

Let $\ZZ_{M} = \{f(\S^{(i)}_{T_n})\,:\, f\in \Lip[M],\,i=1,\ldots,d\}$ and
$\YY_{M} = \{f(\S^{(i)}_{T_j})\,:\, f\in \Lip[M],\,i=1,\ldots,d,\,j=1,\ldots,n-1\}.$ We also write
$$ \ZZ = \bigcup_{M\ge 0} \ZZ_{M} \;\text{ and }\; \YY = \bigcup_{M\ge 0} \YY_{M}.$$

Notice that given any $f\in C_b(\R_+,\R)$, $\epsilon>0$ and a measure $\mu$ on $\R_+$ which has finite first moment, there is some $u:\R_+\to \R$ taking the form $a_0+\sum_{i=1}^na_i(s-K_i)^+$ such that $u\ge f$ and $\int (u-f)d\mu<\epsilon$. It follows that
\begin{align}
&\duall\Big(G-\alpha\wedge(\beta\sqrt{m^{(D)}})\Big) \nonumber\\
=& \widetilde{V}_{\ZZ\cup\YY,\PP,\set}\Big(G-\alpha\wedge(\beta\sqrt{m^{(D)}})\Big)\label{eq:new_new_inequality_1}\\
=& \inf_{X\in \Lin(\ZZ\cup\YY),\,N\ge 0}\Big\{\dualf\Big(G-X-\alpha\wedge(\beta\sqrt{m^{(D)}})-N\lambda_{\set}\Big)+\PP(X)\Big\}\label{eq:new_new_inequality_2}\\
\le& \inf_{X\in \Lin(\ZZ\cup\YY),\,N\ge 0}\Big\{\primalfs\Big(G-X-\alpha\wedge(\beta\sqrt{m^{(D-2)}})-N\lambda_{\set}\Big)+\PP(X)\Big\}\label{eq:new_new_inequality_3}\\
=& \inf_{Y\in \Lin(\YY),\,N\ge 0}\inf_{M\ge 0}\inf_{Z\in \Lin(\ZZ_M)}\Big\{\primalfs\Big(G-Y-Z-\alpha\wedge(\beta\sqrt{m^{(D-2)}})-N\lambda_{\set}\Big)+\PP(Y+Z)\Big\}\label{eq:new_inequality_1}\\
\le & \inf_{Y\in  \Lin(\YY)}\;\inf_{M\ge 0,\, N\ge 0}\;\lim_{L\to \infty}\sup_{\P\in \MSUPs^{1/L}_{\ZZ_{M},\PP, \info}}\E_{\P}[G-\alpha\wedge(\beta\sqrt{m^{(D-4)}})-Y-N\lambda_{\set}+\PP(Y)]\label{eq:new_inequality_2}\\
\le & \inf_{Y\in  \Lin(\YY)}\;\inf_{M\ge 0,\, N\ge 0}\;\sup_{\P\in \MSUPs^{1/M}_{\ZZ_{M},\PP, \info}}\E_{\P}[G-\alpha\wedge(\beta\sqrt{m^{(D-4)}})-Y-N\lambda_{\set}+\PP(Y)].\label{eq:new_inequality_3}
\end{align}
where the equality between \eqref{eq:new_new_inequality_1} and \eqref{eq:new_new_inequality_2} follows from Remark \ref{remark: remark_v_penalty}, the inequality between \eqref{eq:new_new_inequality_2} and \eqref{eq:new_new_inequality_3} is justified by Theorem \ref{thm: real_main_thm}. Finally the inequality between \eqref{eq:new_inequality_1} and \eqref{eq:new_inequality_2} is given by Lemma \ref{thm: third_main}. To justify this, we first observe that $\Lip[1]$ is a convex and compact subset of $\CC(\R_+,\R)$. Then, since $\Lin_1(\ZZ_M) = \ZZ_M$ for any $M$,  $\Lin_1(\ZZ_M)$ satisfies Assumption \ref{assumption: XX}. Therefore, by keeping $Y$ and $N$ fixed and applying Lemma \ref{thm: third_main} to 
$$ \inf_{Z\in \Lin(\ZZ_{M})}\Big\{\primalfs(G-\alpha\wedge(\beta\sqrt{m^{(D-2)}})-Y-Z-N\lambda_{\set})+\PP(Y)+\PP(Z)\Big\}, $$
we establish the inequality. 

For any $\P\in\MSUPs_{\ZZ_M, \PP, \info}^{1/M}$, let $\epsilon_{\P} = \max\{d_p(\mu_n^{(i)}, \LL_{\P}(\S^{(i)}_{T_n}))\,:\, 1\le i \le d\}$. Since $d_p$, the \Levy--Prokhorov's metric on probability measures on $\R_+^d$, is given by
\begin{equation*}
d_p(\mu,\nu):=\sup_{f\in \mathfrak{G}^b_1(\R_+^d)}\Big|\int f d\nu-\int f d\mu\Big|, 
\end{equation*} 
where $\mathfrak{G}^b_1(\R_+^d):=\big\{f\in C(\R^d_+,\R): \|f\|\le 1 \text{ and } |f(\vec{x})-f(\vec{y})|\le |\vec{x}-\vec{y}| \;\forall \vec{x}\neq \vec{y}\big\}$, we can pick $g\in \mathfrak{G}^b_1(\R_+)$ such that 
\begin{align*}
\Big|\int_{\R_+}g(x)\mu_n^{(i)}(\td x)-\E_{\P}[g(\S_{T_n}^{(i)})]\Big| > \epsilon_{\P}/2 \quad \text{ for some } i=1,\ldots,d,
\end{align*}
and define $\hat{g}\in \mathfrak{G}^b_M(\R_+)$ via $\hat{g}(x)=Mg(x\wedge M^2)$. Then,
\begin{align}
&\Big|\int_{\R_+}\hat{g}(x)\mu_n^{(i)}(\td x)-\E_{\P}[\hat{g}(\S_{T_n}^{(i)})]\Big| \label{eq:example_2 1}\\
\ge&  M\Big|\int g d\mu_n^{(i)} - \E_{\P}[g(\S_{T_n}^{(i)})] \Big|-(M+1)\mu_n^{(i)}(\{|x|>M^2\})-(M+1)\P(|\S_{T_n}^{(i)}|>M^2) \nonumber\\
\ge& M\epsilon_{\P}/2-\frac{2(M+1)}{M^2}. \nonumber
\end{align}
By definition of $\MSUPs_{\ZZ_M, \PP, \info}^{1/M}$, 
\begin{align*}
\Big|\int_{\R_+}\hat{g}(x)\mu^{(i)}_n(\td x)-\E_{\P}[\hat{g}(\S^{(i)}_{T_n})]\Big| \le 1/M.
\end{align*}
Hence, $\epsilon_{\P} \le 1/M^2 + 2(M+1)/M^3 \le 2/M$ when $M$ is sufficiently large. It follows that $\P\in \MSUPs_{\mu_n, \info, 2/M}$ and hence $\MSUPs_{\ZZ_M, \PP, \info}^{1/M} \subseteq \MSUPs_{\mu_n, \info, 2/M}$ when $M$ is sufficiently large. In consequence
\begin{align*}
&\inf_{M\ge 0}\sup_{\P\in \MSUPs^{1/M}_{\ZZ_{M},\PP,\info}}\E_{\P}[G-\alpha\wedge(\beta\sqrt{m^{(D-4)}})-Y-N\lambda_{\set}+\PP(Y)]\\
\le&\, \inf_{M\ge 0}\sup_{\P\in \MSUPs_{\mu_n, \info, 2/M}}\E_{\P}[G-\alpha\wedge(\beta\sqrt{m^{(D-4)}})-Y-N\lambda_{\set} +\PP(Y)].
\end{align*}

For every $M\in \N_+$, take $\P^{(M)}\in \underline{\MM}_{\mu_n,\info,2/M}$ such that 
\begin{align*}
&\E_{\P^{(M)}}\Big[G-\alpha\wedge(\beta\sqrt{m^{(D-4)}})-Y-N\lambda_{\set} +\PP(Y)\Big]\\
\ge&\, \sup_{\P\in \MSUPs_{\mu_n, \info, 2/M}}\E_{\P}[G-\alpha\wedge(\beta\sqrt{m^{(D-4)}})-Y-N\lambda_{\set} +\PP(Y)]-\frac{1}{M}.
\end{align*}
Let $\pi^{(M)}_n$ be the law of $(\S^{(1)}_{T_n},\ldots,\S^{(d)}_{T_n})$ under $\P^{(M)}$. It is a probability measure on $\R^{d}_+$ with mean $1$. It follows that the family $\{\pi^{(M)}_n\}_{M\ge 1}$ is tight. By Prokhorov theorem, there exists a subsequence $\{\pi^{(M_k)}_n\}$ converging to some $\pi_n$.
Note that the marginal distributions of $\pi_n$ are $\mu_n^{(i)}$'s. By Lemma \ref{thm:continuity in probaility}, it follows that
\begin{align*}
&\lim_{M\to \infty}\sup_{\P\in \MSUPs_{\mu_n, \info, 2/M}}\E_{\P}[G-\alpha\wedge(\beta\sqrt{m^{(D-4)}})-Y-N\lambda_{\set} +\PP(Y)]\\
\le& \sup_{\P\in \underline{\MM}_{\mu_n,\info}}\E_{\P}[G-\alpha\wedge(\beta\sqrt{m^{(D-6)}})-Y-N\lambda_{\set} +\PP(Y)].
\end{align*}

With the result above, we continue with \eqref{eq:new_inequality_3}: 
\begin{align*}
\duall(G)\le& \inf_{Y\in  \Lin(\YY)}\;\inf_{M\ge 0,\, N\ge 0}\;\sup_{\P\in \MSUPs^{1/M}_{\ZZ_{M},\PP, \info}}\E_{\P}[G-\alpha\wedge(\beta\sqrt{m^{(D-4)}})-Y-N\lambda_{\set}+\PP(Y)]\\
\le & \inf_{Y\in  \Lin(\YY),\, N\ge 0}\; \sup_{\P\in \underline{\MM}_{\mu_n,\info}}\E_{\P}[G-\alpha\wedge(\beta\sqrt{m^{(D-6)}})-Y-N\lambda_{\set} +\PP(Y)]\\
\le & \inf_{M\ge 0} \inf_{Y\in \Lin(\YY_M), N\ge 0}\sup_{\P\in \MM_{\mu_n,\info}}\E_{\P}[G-\alpha\wedge(\beta\sqrt{m^{(D-6)}})-Y-N\lambda_{\set} +\PP(Y)]\\
\le & \inf_{M\ge 0} \lim_{N\to \infty}\sup_{\P\in \MM_{\mu_n,\info}\cap\MSUP^{1/N}_{\YY_{M},\PP,\set}}\E_{\P}[G-\alpha\wedge(\beta\sqrt{m^{(D-8)}})]\\
\le & \inf_{M\ge 0} \sup_{\P\in \MM_{\mu_n,\info}\cap\MSUP^{1/M}_{\YY_{M},\PP,\set}}\E_{\P}[G-\alpha\wedge(\beta\sqrt{m^{(D-8)}})].
\end{align*}
To justify the second last inequality, we first notice that $\MSUP^{\eta}_{\YY_{M},\PP,\set}\neq \emptyset$ for any $M\in \N$ and $\eta >0$ since $\MMUE{\eta}\neq\emptyset$ for any $\eta>0$, and hence it follows from Lemma \ref{thm: third_main}, with $\MM_s=\MM_{\mu_n,\info}$, that 
\begin{align*}
&\inf_{M\ge 0} \inf_{Y\in \Lin(\YY_M), N\ge 0}\sup_{\P\in \MM_{\mu_n,\info}}\E_{\P}[G-\alpha\wedge(\beta\sqrt{m^{(D-6)}})-Y-N\lambda_{\set} +\PP(Y)]\\
\le &\, \inf_{M\ge 0} \lim_{N\to \infty}\sup_{\P\in \MM_{\mu_n,\info}\cap\MSUP^{1/N}_{\YY_{M},\PP,\set}}\E_{\P}[G-\alpha\wedge(\beta\sqrt{m^{(D-8)}})] \\
\le &\, \inf_{M\ge 0} \sup_{\P\in \MM_{\mu_n,\info}\cap\MSUP^{1/M}_{\YY_{M},\PP,\set}}\E_{\P}[G-\alpha\wedge(\beta\sqrt{m^{(D-8)}})].
\end{align*}
Next we are going to argue that $\MSUP_{\YY_M, \PP, \set}^{1/M}\cap\MM_{\mu_n, \info}\subseteq \MSUP_{\vec{\mu}, \set, 2/M}$ when $M$ is large enough. Fix $\P\in\MSUP_{\YY_M, \PP, \set}^{1/M}\cap\MM_{\mu_n, \info}$ and let $\tilde{\epsilon}_{\P} = \max\{d_p(\mu_j^{(i)}, \LL_{\P}(\S^{(i)}_{T_j}))\,:\, i \le d,\, j\le n\}$. We can pick $g\in \mathfrak{G}^b_1(\R_+)$ such that 
\begin{align*}
\Big|\int_{\R_+}\hat{g}(x)\mu_j^{(i)}(\td x)-\E_{\P}[\hat{g}(\S_{T_j}^{(i)})]\Big| > \tilde{\epsilon}_{\P}/2 \quad \text{ for some } i\le d,\; j\le n-1.
\end{align*}
By following the same argument as above, we can show that
$\tilde{\epsilon}_{\P} \le 1/M^2 + 2(M+1)/M^3$. Hence, $\MSUP_{\YY_M, \PP, \set}^{1/M}\cap \MM_{\mu_n,\info}\subseteq \MSUP_{\vec{\mu}, \set, 2/M}$ when $M$ is sufficiently large. Therefore, we have 
\begin{align*}
\duall(G) \le&\, \inf_{M\ge 0} \sup_{\P\in \MM_{\mu_n,\info}\cap\MSUP^{1/M}_{\YY_{M},\PP,\set}}\E_{\P}[G-\alpha\wedge(\beta\sqrt{m^{(D-8)}})]\\
\le& \lim_{N\to \infty}\sup_{\P\in \MMUE{2/N}}\E_{\P}[G-\alpha\wedge(\beta\sqrt{m^{(D-8)}})]= \primalls(G-\alpha\wedge(\beta\sqrt{m^{(D-8)}})).
\end{align*}

\end{proof}

\subsection{Proof of Theorem \ref{thm: extended_weak_thm}}
We start with a key lemma, analogous to the one obtained in \citet{Yan}.
\begin{lemma}\label{lemma:primal estimate}
Consider 
\begin{equation}\label{eq:small claim}
\alpha_D(S):=\big(\max_{1\le i\le d}\|S^{(i)}\|^p+1\big)\indicators{\{\max_{1\le i\le d}\|S^{(i)}\|+1\ge D\}}+\frac{\max_{1\le i\le d}\|S^{(i)}\|^p}{D}.
\end{equation}
Then, given that $(\mu_j^{(i)})$ satisfies Assumption \ref{ass:assumption_on_measure}, for any $\P\in \MM$ such that $\LL_{\P}(\S^{(i)}_{T_j})=\mu^{(i)}_j$ $\forall i\le d,\,j\le n$ 
\begin{equation}
\E_{\P}[\alpha_D(\S)]\le e_2(\vec{\mu}_n,D),
\end{equation}
where $e_2(\vec{\mu}_n,D):=\Big(\frac{p}{p-1}\Big)^p\sum_{i=1}^d\Big(2\int_{|x|\ge(\frac{p-1}{p})(D-1)}|x|^p\mu^{(i)}_n(\td x)+\frac{1}{K}\int|x|^p\mu^{(i)}_n(\td x)\Big)\to 0$ as $D\to \infty$.
\end{lemma}
\begin{proof}
First define $h_D:\R\to \R$ by 
\begin{eqnarray*}
h_{D}(x)=pD^{p-1}\Big(|x|-\Big(\frac{p-1}{p}\Big)D\Big)\indicator{(\frac{p-1}{p})D\le|x|< D}+|x|^p\indicator{|x|\ge D}.
\end{eqnarray*}
Notice that $h_K$ is convex and satisfies 
\begin{eqnarray*}
|x|^p\indicator{|x|\ge D}\le h_{D}(x)\le |x|^p\indicator{|x|\ge (\frac{p-1}{p})D}, \quad \text{ for any } D\ge 1.
\end{eqnarray*}
For any $\P\in \MM$ such that $\LL_{\P}(\S^{(i)}_{T_j})=\mu^{(i)}_j$ $\forall i\le d,\,j\le n$, $\{h_D(\S_t):=(h^{(1)}_D(\S_t),\ldots,h^{(d)}_D(\S_t))\}_{t\ge 0}$ is a sub-martingale under $\P$ since $h_D$ is convex. Therefore by Doob's inequality
\begin{align*}
\E_{\P}[\alpha_D(\S)]\le&\; \E_{\P}\big[2\|h_{D-1}(\S)\|\big]+\frac{1}{D}\E_{\P}[\|\S\|^p]\\
\le& \; \sum_{i=1}^d\Big(\frac{p}{p-1}\Big)^p\Big(\E_{\P}[2h^{(i)}_{D-1}(\S_{T_n})]+\frac{1}{K}\E_{\P}[|\S^{(i)}_{T_n}|^p]\Big)\\
\le& \; \sum_{i=1}^d\Big(\frac{p}{p-1}\Big)^p\bigg(2\int_{|x|\ge(\frac{p-1}{p})(D-1)}|x|^p\mu^{(i)}_n(\td x)+\frac{1}{K}\int|x|^p\mu^{(i)}_n(\td x)\bigg) =e_2(\vec{\mu}_n,D).
\end{align*}
\end{proof}

We now proceed with the proof the Theorem \ref{thm: extended_weak_thm}.
We first show that $\primalls(G)\le \widetilde{V}^{(p)}_{\XX^{(p)},\PP,\set}(G)$.

Given $(\XX,\ts)\in\AA^{(p)}_{\XX}$ such that $(\XX,\ts)$ super-replicates $G$ on $\set^{\epsilon}$ for some $\epsilon>0$, since $X$ is bounded, it follows from the definition of $\AA^{(p)}$ that there exists $M_1>0$ such that 
\begin{align}
X(\S^{(1)},\ldots,\S^{(d)})+\int_0^{T_n}\ts_ud\S_u\ge G(\S)-M_1(1+\sup_{0\le t\le T_n}|\S_t|^p)\indicator{\S\notin \set^{\epsilon}} \;\text{ on } \info.\label{eq:extended_weak_main_1}
\end{align}
Next, for each $N\ge 1$, we pick $\P^{(N)}\in \MMUE{1/N}$ such that 
$\displaystyle \E_{\P^{(N)}}[G(\S)]\ge \sup_{\P\in \MMUE{1/N}}\E_{\P}[G(\S)]-\frac{1}{N}$. We first notice that $X(\S^{(1)},\ldots,\S^{(d)})$ is of the form
$ \sum_{i=1}^d\sum_{j=1}^n f_{i,j}(\S^{(i)}_j)$ 
for some $f_{i,j}$ such that $\forall i\le d,\,j\le n$ $f_{i,j}$ is continuous and bounded by $M_2(1+|\S^{(i)}_{T_j}|)$ for some $M_2$. Since by Jensen's inequality, for any $\P\in \MSUP$ such that $\LL_{\P}(\S^{(i)}_{T_n})=\mu_n^{(i)}$ $\forall i\le d$
$$\E_{\P}[|\S^{(i)}_{T_j}|^p]\le \E_{\P}[|\S^{(i)}_{T_n}|^p]\le\int_{[0,\infty)}x^p\mu^{(i)}_n(d x)<\infty\quad \forall i\le d,\,j\le n,$$
it follows from weak convergence of measures, definition of $\MMUE{\epsilon}$ and Lemma \ref{lemma:primal estimate} that 
\begin{align}
\big|\PP(X)-\E_{\P^{(N)}}[X(\S^{(1)},\ldots,\S^{(d)})]\big|\to 0 \quad\text{ as } N\to \infty.\label{eq:extended_weak_main_3}
\end{align}
Since $\ts$ is progressively measurable in the sense of \eqref{eq:pm}, the integral $\int_0^{\cdot} \ts_u(\S)\cdot d\S_u$, defined pathwise via integration by parts, agrees a.s. with the stochastic integral under any $\P^{(N)}$. Then, by \eqref{eq:admissible}, the stochastic integral is a $\P^{(N)}$ super-martingale and hence $\E_{\P^{(N)}}\Big[\int_0^{T_n} \ts_u(\S)\cdot d\S_u\Big]\le 0$. Therefore, by Lemma \ref{lemma:primal estimate}
\begin{align*}
\E_{\P^{(N)}}[X(\S)]\ge& \E_{\P^{(N)}}\big[G(\S)-M_1(1+|\S|^p)\indicator{\S\notin \set^{\epsilon}}\big]\\
\ge& \sup_{\P\in \MMUE{1/N}}\E_{\P}[G(\S)]-\frac{1}{N}-M_1\E_{\P^{(N)}}\big[(1+|\S|^p)\indicator{\S\notin \set^{\epsilon}}\indicator{\|\S\|\le N^{1/2p}}\big]\\
&-M_1\E_{\P^{(N)}}\big[(1+|\S|^p)\indicator{\S\notin \set^{\epsilon}}\indicator{\|\S\|> N^{1/2p}}\big]\\
\ge& \sup_{\P\in \MMUE{1/N}}\E_{\P}[G(\S)]-\frac{1}{N}-\frac{M_1(1+\sqrt{N})}{N}-\frac{M_1}{N}-e_2(\vec{\mu}_n, N^{1/2p}).
\end{align*}

This, together with \eqref{eq:extended_weak_main_3}, yields 
\begin{align*}
\PP(X)\ge \primalls(G).
\end{align*}
As $(X,\ts)\in\AA_{\XX}$ is arbitrary, we therefore establish $\primalls(G)\le \widetilde{V}^{(p)}_{\XX^{(p)},\PP,\set}(G)$. 

Let $D>1$ and define $G_{D}$ by $G_{D}=G\wedge D\vee(-D)$.
Then it is clear that $G_D$ is bounded and uniformly continuous. Therefore, by Theorem \ref{thm: weak_main}
\begin{alignat*}{1}
\widetilde{V}^{(p)}_{\XX^{(p)},\PP,\set}(G_{L+D})= &\lim_{\eta\searrow 0}\sup_{\P\in\MMUE{\eta}}\E_{\P}[G_{L+D}(\S)] \nonumber\\
\le&\lim_{\eta\searrow 0}\sup_{\P\in\MMUE{\eta}}\E_{\P}[G(\S)]+2L\lim_{\eta\searrow 0}\sup_{\P\in\MMUE{\eta}}\E_{\P}\Big[(1+\|\S\|^p)\indicator{\|\S\|\ge (\frac{D}{L})^{1/p}}\Big],
\end{alignat*}
where the second inequality follows from $G_{L+D}(\S)\le G(\S)+2L\big(1+\|\S\|^p\indicator{\|\S\|\ge (\frac{D}{L})^{1/p}}\big)$.\\
We know from Assumption \ref{assumption:key assumption} that any $S\in \info$ satisfies $\|S^{(i)}\|\le \kappa$ $\forall i >d$, where $\kappa$ is the smallest number such that $X^{(c)}_i/\PP(X^{(c)}_i)$'s are bounded by $\kappa$.
It follows from Lemma~\ref{lemma:primal estimate} that for any $D\ge L\kappa^p$ and $\P\in \MMUE{\eta}$
\begin{align*}
&\E_{\P}\Big[(1+\|\S\|^p)\indicator{\|S\|\ge (\frac{D}{L})^{1/p}}\Big]\\
=&\E_{\P}\Big[\max_{1\le i\le d}\|\S^{(i)}\|^p\indicator{\max_{1\le i\le d}\|\S^{(i)}\|^p\ge (\frac{D}{L})^{1/p}}\Big]
\le e_2(\vec{\mu}_n,D/L) \to 0, \quad \text{ as } D \to \infty,
\end{align*}
and therefore,
\begin{eqnarray*}
\limsup_{D\to\infty}\widetilde{V}^{(p)}_{\XX^{(p)},\PP,\set}(G_{D+L})\le\primalls(G).
\end{eqnarray*}
On the other hand, by the linearity of the market,
\begin{eqnarray*}
\widetilde{V}^{(p)}_{\XX^{(p)},\PP,\set}(G)&\le& \widetilde{V}^{(p)}_{\XX^{(p)},\PP,\set}(G_{D+L})+\widetilde{V}^{(p)}_{\XX^{(p)},\PP,\set}(2L(1+\|\S\|^p)\indicator{\|\S\|\ge (\frac{D}{L})^{1/p}}).
\end{eqnarray*}
Since $\MMUE{\eta}\neq\emptyset$ for any $\eta >0$, $\widetilde{V}^{(p)}_{\XX^{(p)},\PP,\set}(\alpha_D)\ge \primalls(\alpha_D)\ge 0$. Then it follows from lemma 4.1 in \citet{Yan} and the obvious fact that
$\widetilde{V}^{(p)}_{\XX^{(p)},\PP,\Omega}(\alpha_D)\ge \widetilde{V}^{(p)}_{\XX^{(p)},\PP,\set}(\alpha_D)$ that 
\begin{eqnarray*}
\limsup_{D\to\infty}\widetilde{V}^{(p)}_{\XX^{(p)},\PP,\set}(\alpha_D)= 0.
\end{eqnarray*}
Hence we conclude that
\begin{equation*}
\primalls(G)\le \widetilde{V}^{(p)}_{\XX^{(p)},\PP,\set}(G)\le \limsup_{D\to\infty}\widetilde{V}^{(p)}_{\XX^{(p)},\PP,\set}(G_{D+L})\le \primalls(G)
\end{equation*}
and therefore we have equalities throughout.

\subsection{Proof of Theorem \ref{theorem: new_section_main}}
We first make two simple observations.
\begin{remark}\label{lemma: limit of enlargement is itself}
If $\set$ is a non-empty closed subset of $\Omega$ with respect to sup norm, then 
\begin{align*}
\set = \bigcap_{\epsilon>0}\set^{\epsilon} = \bigcap_{\epsilon>0}\overbar{\set^{\epsilon}},
\end{align*}
where $\overbar{\set^{\epsilon}}$ is the closure of $\set^{\epsilon}$.
\end{remark}
\begin{lemma}\label{lemma: enlargement also time invariant}
If $\set$ is time invariant, then for every $\epsilon>0$ $\set^{\epsilon}$ is also time invariant.
\end{lemma}
\begin{proof}
Fix $\epsilon>0$, $S\in \set^{\epsilon}$ and a non-decreasing continuous function $f:[0,T_n]\to [0,T_n]$ such that $f(0)=0$ and $f(T_i) = T_i$ for any $i=1,\ldots,n$. By definition, there exist $S^{(N)}\in \set$ such that
$$ \|S^{(N)}-S\|\le \epsilon +\frac{1}{N}$$ 
Now write $\tilde{S}_t = S_{f(t)}$ and $\tilde{S}^{(N)}_t = S^{(N)}_{f(t)}$. Note that $\tilde{S}^{(N)}\in \set$ as $\set$ is time invariant. Then it is clear that
$$ \|\tilde{S}^{(N)}-\tilde{S}\|=\|S^{(N)}-S\|\le \epsilon +\frac{1}{N},$$
which implies that $\tilde{S}\in \set^{\epsilon}$. Since $S\in \set^{\epsilon}$ and $f$ are arbitrary, we can therefore conclude that $\set^{\epsilon}$ is time invariant. 
\end{proof}

We now proceed with the proof the Theorem \ref{theorem: new_section_main}. As argued in the proof of Theorem \ref{thm: extended_weak_thm} above, by Lemma \ref{lemma:primal estimate}, it suffices to argue  \eqref{eq: new_section_main} for bounded $G$. Further, note that the inequalities
$\widetilde{V}^{(p)}_{\vec{\mu}, \set}(G)\ge V^{(p)}_{\vec{\mu}, \set}(G)\ge P_{\vec{\mu}, \set}(G)$ hold in general. In addition, according to Theorem \ref{thm: extended_weak_thm}, $\widetilde{V}^{(p)}_{\vec{\mu}, \set}(G) = \widetilde{P}_{\vec{\mu}, \set}(G)$. Therefore, we only need to show $\widetilde{P}_{\vec{\mu}, \set}(G) = P_{\vec{\mu}, \set}(G)$. Our proof of this equality is divided into six steps. First, using Lemma \ref{thm: weak_main_new_section}, we argue that it suffices to consider measures with ``good control" on the expectation of $m^{(D)}(\S)$. Next, we perform three time changes within each trading period $[T_i, T_{i+1}]$. The resulting time change of $\S$, denoted $\ddot{\S}$, allows for a  ``good control" over its quadratic variation process. At the same time, we keep $G(\S)$ and $G(\ddot{\S})$ ``close" and given a measure $\P\in \MM_{\vec{\mu}, \set, \eta}$ with good control on $\E_{\P}[m^{(D)}(\S)]$, since $\set^{\eta}$ is time invariant, the law of the time-changed price process $\ddot{\S}$ remains an element of $\MM_{\vec{\mu}, \set, \eta}$. Then, in Step 5, given a sequence of models with improved calibration precisions, we show tightness of quadratic variation process of the time-changed price process $\ddot{\S}$ under these measures. This then leads to tightness of images measures via $\ddot{\S}$. In Step 6, we deduce the duality $\widetilde{P}_{\vec{\mu}, \set}(G) = P_{\vec{\mu}, \set}(G)$ from tightness and conclude.

\textbf{Step 1: Reducing to measures $\P$ with good control on $\E_{\P}[m^{(D)}(\S)]$.}

Let $G$ be bounded and satisfy Assumption \ref{assumption:2.1}. Choose $\kappa\in \R_+$ such that $\|G\|\le \kappa$ and let $f_e:\R^{d+K}_+\to \R_+$ be the modulus of continuity of $G$, i.e.\ 
$$|G(\omega)-G(\upsilon)|\le f_e(|\omega-\upsilon|) \;\;\text{ for any }\omega, \upsilon\in \Omega$$ 
with $\lim_{x\to 0} f_e(x) = 0$. 
Fix $D\in \N$.  
Consider a random variable
\begin{equation*}
X_{D}(\S) = \sqrt{\sum_{j=1}^{m^{(D)}(\S)}\sum_{i=1}^{d+K}|\S^{(i)}_{\tau^{(D)}_j(\S)}- \S^{(i)}_{\tau^{(D)}_{j-1}(\S)}|^2} \ge 2^{-D}\sqrt{m^{(D)}(\S)-1}\ge 2^{-D}(\sqrt{m^{(D)}(\S)} -1),
\end{equation*}
where $\tau^{(D)}_i$'s and $m^{(D)}$ are defined in Definition \ref{defn:stopping time}.

Then by Lemma 5.4 in \citet{dolinsky2014martingale}
$$ 0\le V^{(p)}_{\vec{\mu}, \set}(X_{D}(\S)) \le V^{(p)}_{\vec{\mu}, \info}(X_{D}(\S))\le 3d V^{(p)}_{\vec{\mu}, \info}(\|\S\|^p)<\infty.$$ 
It follows that from the linearity of the market and the estimate above
\begin{align*}
\widetilde{V}^{(p)}_{\vec{\mu}, \set}(G(\S))\le& \widetilde{V}^{(p)}_{\vec{\mu}, \set}\Big(G(\S)-\kappa 2^D\wedge\frac{X_{D}(\S)}{2^D}\Big)+\widetilde{V}^{(p)}_{\vec{\mu}, \set}(X_{D}(\S)/2^D)\\
\le&\, \widetilde{V}^{(p)}_{\vec{\mu}, \set}\Big(G(\S)-\kappa 2^D\wedge\frac{\sqrt{m^{(D)}(\S)}}{2^{2D}}\Big)+c_2/2^D\\
\le&\, \widetilde{V}_{\vec{\mu}, \set}\Big(G(\S)-\kappa 2^D\wedge\frac{\sqrt{m^{(D)}(\S)}}{2^{2D}}\Big)+c_2/2^D\\
\le&\, \widetilde{P}_{\vec{\mu}, \set}\Big(G(\S)-\kappa 2^D\wedge\frac{\sqrt{m^{(D-8)}(\S)}}{2^{2D}}\Big)+c_2/2^D\\
=&\, \lim_{N\to\infty}\sup_{\P\in\MM_{\vec{\mu}, \set, 1/N}}\E_{\P}\Big[G(\S)-\kappa 2^D\wedge\frac{\sqrt{m^{(D-8)}(\S)}}{2^{2D}}\Big]+c_2/2^D.
\end{align*}
where $c_2$ is a constant and the last inequality follows from Lemma \ref{thm: weak_main_new_section}.

Next we denote $\widetilde{\MSUP}_{\info}$ the set of $\P\in \MSUP_{\info}$ such that 
\begin{align}
\E_{\P}\Big[\kappa 2^D\wedge\frac{\sqrt{m^{(D-8)}(\S)}}{2^{2D}}\Big]\le 2\kappa+ 2. \label{eq: new_section_control_on_stopping_times}
\end{align}
We notice that if $\P\in \MM_{\vec{\mu}, \set, 1/N}$ such that $\P\notin \widetilde{\MSUP}_{\info}$, then
\begin{align*}
\E_{\P}\Big[G(\S)-\kappa 2^D\wedge\frac{\sqrt{m^{(D-8)}(\S)}}{2^{2D}}\Big] <\kappa- 2\kappa-2 =-\kappa-2.
\end{align*}
While for $N$ sufficiently large,
\begin{align*}
\sup_{\P\in \MM_{\vec{\mu}, \set, 1/N}}\E_{\P}\Big[G(\S)-\kappa 2^D\wedge\frac{\sqrt{m^{(D-8)}(\S)}}{2^{2D}}\Big]  \ge& \widetilde{P}_{\vec{\mu}, \set}\Big(G(\S)-\kappa 2^D\wedge\frac{\sqrt{m^{(D-8)}(\S)}}{2^{2D}}\Big) \\
\ge& V^{(p)}_{\vec{\mu}, \set}(G(\S)) - c_2/2^D \ge-\kappa-1 \quad \text{ for a large $D$ }.
\end{align*}
It follows that 
\begin{align*}
\lim_{N\to\infty}\sup_{\P\in \MM_{\vec{\mu}, \set, 1/N}}\E_{\P}\Big[G(\S)-\kappa 2^D\wedge\frac{\sqrt{m^{(D-8)}(\S)}}{2^{2D}}\Big] =\lim_{N\to\infty}\sup_{\P\in \widetilde{\MSUP}_{\info}\cap\MM_{\vec{\mu}, \set, 1/N}}\E_{\P}\Big[G(\S)-\kappa 2^D\wedge\frac{\sqrt{m^{(D-8)}(\S)}}{2^{2D}}\Big].
\end{align*}
In particular, $\widetilde{\MSUP}_{\info}\cap\MM_{\vec{\mu}, \set, 1/N} \neq \emptyset$ for $N$ large enough. 

\textbf{Step 2: First time change: ``squeezing paths and adding constant paths''.}

Now for every $N\in \N$, take $\P^{(N)}\in \widetilde{\MSUP}_{\info}\cap\MM_{\vec{\mu}, \set, 1/N}$ such that
\begin{equation*}
\E_{\P^{(N)}}[G(\S)] \ge \sup_{\P\in \tilde{\MSUP}_{\info}\cap\MM_{\vec{\mu}, \set, 1/N}}\E_{\P}[G(\S)]-1/N. 
\end{equation*}
Define an increasing function $f: [0,T_n]\mapsto [0,T_n]$ by
\begin{align*}
f(t) = \sum_{i=1}^n\Big(T_{i}\wedge\Big(T_{i-1}+\frac{(T_i - T_{i-1})(t-T_{i-1})}{T_i - T_{i-1} -1/D}\Big)\Big) \indicator{T_{i-1}< t\le T_i}
\end{align*}
and then a process $(\tilde{\S}_t)_{t\in [0,T_n]}$ by a time change of $\S$ via $f$, i.e.\ $\tilde{\S}_t = \S_{f(t)}$. 
It follows from \eqref{eq: time-continuity of G} that
\begin{align}
|G(\S) - G(\tilde{\S})| \le& |G(\S) - G(F^{(D)}(\S))|+ |G(\tilde{\S}) - G(F^{(D)}(\tilde{\S}))| + |G(F^{(D)}(\S)) - G(F^{(D)}(\tilde{\S}))|\nonumber\\
\le& 2f_e(2^{-D+9}) + \frac{2Ln\|\S\|}{D}. \label{eq: new_section_original_tilde_difference}
\end{align}
In addition,  $\S_{T_i} =\ST_{T_i}$ $\forall i\le n$. In particular, $\LL_{\P^{(N)}}(\S_{T_i}) = \LL_{\P^{(N)}}(\ST_{T_i})$ $\forall i\le n$ and further $\P^{(N)} \circ (\ST_t)^{-1} \in \MM_{\vec{\mu}, \set, 1/N}$ as $\set^{1/N}$ is time invariant, by Lemma \ref{lemma: enlargement also time invariant}. 

\textbf{Step 3: Second time change: introducing lower bound on time step.}

For ease of notation it is helpful to rename the elements of the set 
$$\{\tau^{(N)}_j\,:\, j\le m^{(N)}_j\}\cup \{T_i\,:\, i = 1,\ldots, n\}$$
as follows. We define a sequence of stopping times $\tau_{i,j}^{(N)}:\Omega\to [T_{i-1},T_{i}]$ and $m^{(N)}_i:\Omega \to \N_+$ in a recursive manner: set $m^{(N)}_0(\S) = 0$ and $\tau^{(N)}_{0,-1}(\S) = 0$, and $\forall i=1,\ldots,n$, set $\tau^{(N)}_{i,0}(S)=T_{i-1}$ and let
\begin{eqnarray*}
&\tau^{(N)}_{i,1}(\S)=\inf\Big\{t\ge T_{i-1}:|\S_t-\S_{\tau^{(N)}_{i-1,m^{(N)}_{i-1}(\S)-1}(\S)}|=\frac{1}{2^N}\}\wedge T_i,\\
&\tau^{(N)}_{i,k}(\S)=\inf\Big\{t\ge\tau_{i,k-1}(\S):|\S_t-\S_{\tau^{(N)}_{i,k-1}(\S)}|=\frac{1}{2^N}\}\wedge T_i,\\
&m^{(N)}_i(\S)=m^{(N)}_{i-1}+\min\{k\in \N: \tau^{(N)}_{i,k}(\S)=T_i\}.
\end{eqnarray*}
It follows that for any $S\in \info$
\begin{align}
m^{(D-8)}(S)\le m^{(D-8)}_n(S)  \le m^{(D-8)}(S) + n-1.\label{eq: difference of two m's}
\end{align}

Set $\Theta = 2\lceil \kappa^2 2^{6D} \rceil+n$ and $\delta = 1/(4D\Theta^2)$. We now define a sequence of stopping times $\sigma_{i,j}:\Omega \to [0,T_n]$. Fix any $S\in \Omega$ as follows. Firstly, set $\sigma_{i,0}(S) = T_{i-1}$ and $\sigma_{i,\Theta+1}(S) = T_{i}$. Then, for $j\le \Theta$, 
$\sigma_{i,j}(S) = \Big(\tau^{(D-8)}_{i, j}(S) + \delta j\Big)\wedge \big(T_i-1/(2D)\big)$ if $j< m^{(D-4)}_i(S)$, and $\sigma_{i,j}(S) = T_i-1/(2D)$ otherwise. 

Then it follows from the definition that $T_{i-1}=\sigma_{i,0}(\S)\le \sigma_{i,1}(\S)\le \ldots \le \sigma_{i,\Theta}(\S)<\sigma_{i,\Theta+1}(\S)= T_i$. We also note that since $\ST$ is always constant on $[T_i-1/D, T_i]$, $\tau^{(D-8)}_{i, j}(\ST)\le T_i-1/D$ and hence for $j\le \Theta \wedge  (m^{(D-8)}_{i}(\ST)-1)$
$$\sigma_{i,j}(\ST) \le \tau^{(D-8)}_{i, m^{(D-4)}_i-1}(\ST) + \delta (\Theta - 1) \le T_i-\frac{1}{D} + \frac{1}{4D\Theta}< T_i - \frac{1}{2D}.$$
Therefore, $\forall j =1,\ldots, \Big(\Theta\wedge  (m^{(D-8)}_{i}(\ST)-1)\Big)$ 
\begin{equation}
\sigma_{i,j}(\ST) - \sigma_{i,j-1}(\ST) = \delta + \big(\tau^{(D-8)}_{i,j}(\ST) - \tau^{(D-8)}_{i,j-1}(\ST)\big)\ge \delta.\label{eq: property_ST}
\end{equation}

Define a process $\check{\S}$ by
\begin{align*}
\check{\S}_t = \sum_{i=0}^{n-1}\sum_{j=0}^{\Theta-1}\Big\{&\ST_{\tau^{(D-8)}_{i, j}(\ST)+ (t- \sigma_{i,j}(\ST)-\delta)^+}\indicators{[\sigma_{i,j}(\ST),\sigma_{i,j+1}(\ST))}(t) \\
&\qquad + \ST_{\big(\tau^{(D-8)}_{i-1,\Theta}(\ST)+\frac{1}{T_i-t}-\frac{1}{T_i-\sigma_{i,\Theta}(\ST)}\big)\wedge T_i}\indicators{[\sigma_{i,\Theta}(\ST), T_i]}(t) \Big\}.
\end{align*}
Equivalently, $\check{\S}$ can be obtained by time changing $\ST$ via an increasing and continuous process $g:[0, T_n]\times \info\to [0,T_n]$, defined by
\begin{align*}
g_t(S) = \sum_{i=0}^{n-1}\sum_{j=0}^{\Theta-1}\Big\{&\Big(\tau^{(D-8)}_{i,j}(S)+(t- \sigma_{i,j}(S)-\delta)^+\Big)\indicators{[\sigma_{i,j}(S),\sigma_{i,j+1}(S))}(t)\\
& + T_i\wedge\Big(\tau^{(D-8)}_{i,\Theta-1}(S)+(\sigma_{i,\Theta}(S)- \sigma_{i,\Theta-1}(S)-\delta)^+ +\frac{1}{T_i-t}-\frac{1}{T_i-\sigma_{i,\Theta}(\ST)}\Big) \indicators{[\sigma_{i,\Theta}(S), T_i]}(t) \Big\}.
\end{align*}
In particular, it follows from \eqref{eq: property_ST} that
\begin{align*}
g_t(\ST) = \sum_{i=0}^{n-1}\sum_{j=0}^{\Theta-1}\Big\{&\Big(\tau^{(D-8)}_{i,j}(\ST)+(t- \sigma_{i,j}(\ST)-\delta)^+\Big)\indicators{[\sigma_{i,j}(\ST),\sigma_{i,j+1}(\ST))}(t)\\
& + T_i\wedge\Big(\sigma_{i,\Theta}(\ST)+\frac{1}{T_i-t}-\frac{1}{T_i-\sigma_{i,\Theta}(\ST)}\Big) \indicators{[\sigma_{i,\Theta}(S), T_i]}(t) \Big\}.
\end{align*}
Furthermore, $g$ is adapted to $\F$ and hence predictable with respect to $\F^{\P^{(N)}}$ -- the usual augmentation of $\F$ (since g is continuous). Therefore, it is clear that
$\check{\S}$ is a local martingale with respect to $\F^{\P^{(N)}}$ under $\P^{(N)}$. Moreover, $\ST_{T_i} = \check{\S}_{T_i} = \S_{T_i}$ for any $i\le n$. This implies that $\check{\S}$ is a martingale with respect to $\F^{\P^{(N)}}$ under $\P^{(N)}$ and further $\P^{(N)} \circ (\check{\S}_t)^{-1} \in \MM_{\vec{\mu}, \set, 1/N}$.

Observe that for any $S\in \Omega$ such that $m^{(D-8)}_n(S)\le \Theta$ it follows from \eqref{eq: time-continuity of G} -- the time continuity property of $G$
\begin{align}
\begin{split}
|G(\ST(S))-G(\check{\S}(S))|\le&\, |G(\ST(S)) - G(F^{(D-8)}(\ST(S)))|+ |G(\check{\S}(S)) - G(F^{(D-8)}(\check{\S}(S)))|\\
 &\,+ |G(F^{(D-8)}(\S)(S)) - G(F^{(D-8)}(\check{\S}(S)))|\\
\le& 2f_e(2^{-D+9})+2nL\|\ST\|\Theta\delta\le 2f_e(2^{-D+9}) +2nL\|\S(S)\|/D,\label{eq: new section second time change}
\end{split}
\end{align}
when $D$ is sufficiently large, where $F^{(D-8)}$ is defined in \eqref{eq: naive_piecewise_constant_approximation}. From \eqref{eq: new_section_control_on_stopping_times}, Markov inequality gives
\begin{equation}\label{eq: new_section_markov_inequality}
\P^{(N)}(\{S\in \info:\, m^{(D-8)}(S)\ge \Theta-n+2\})\le \frac{2\kappa+2}{\kappa D}.
\end{equation}
and hence by \eqref{eq: difference of two m's}
\begin{equation}\label{eq: new_section_markov_inequality_imply}
\P^{(N)}(\{S\in \info:\, m^{(D-8)}_n(S)\ge \Theta+1\})\le \frac{2\kappa+2}{\kappa D}.
\end{equation}
Furthermore, by \eqref{eq: new section second time change} and \eqref{eq: new_section_markov_inequality_imply}
\begin{align}
\big|\E_{\P^{(N)}}[G(\ST)] - \E_{\P^{(N)}}[G(\check{\S})]\big|
\le&\, 2\kappa\P^{(N)}(m^{(D-8)}(\ST)>\Theta)+2f_e(2^{-D+9}) +2nL\E_{\P^{(N)}}[\|\S\|]/D \nonumber\\
\le&\, \frac{4\kappa+4}{2^D} +2f_e(2^{-D+9}) +2LnV^{(p)}_{\mu_n, \info}(\|\S\|)/D. \label{eq: new_section_tilde_check_difference}
\end{align}

\textbf{Step 4: Third time change: controlling increments of quadratic variation.}

We say $\omega\in \CC([0,T_n],\R)$ admits quadratic variation if 
$$ \sum_{k=0}^{m^{(N)}(\omega)-1}\Big(\omega_{\tau^{(N)}_{k}(\omega)}-\omega_{\tau^{(N)}_{k+1}(\omega)}\Big)^2 \text{ converges to a limit as $N\to \infty$ for any $i\le d+K$.} $$
We let $\langle \omega \rangle$ be that limit if $\omega$ admits quadratic variation and zero otherwise. In addition, for $S\in \Omega$, we say $S$ admits quadratic variation if $S^{(i)}$ admits quadratic variation for any $i\le d+K$. 

It follows from Theorem 4.30.1 in \citet{rogers2000diffusions} that for any $\P\in \MM$, $\langle \S \rangle := \big(\langle \S^{(1)} \rangle, \ldots, \langle \S^{(d+K)} \rangle\big)$ agrees with the classical definition of quadratic variation of $\S$ under $\P$ $\P$-a.s.\ 

Now Doob's inequality gives $\forall i\le d$
\begin{equation}
\E_{\P^{(N)}}[\|\check{\S}^{(i)}\|^p]\le \Big(\frac{p}{p-1}\Big)^{p}\int_{[0,\infty)} x^p\mu^{(i)}_n(\td x).
\end{equation}
And, by BDG-inequalities, we know there exist constants $c_p, C_p\in (0,\infty)$ such that 
\begin{equation}
c_p\E_{\P^{(N)}}\big[\langle \check{\S}^{(i)} \rangle^{p/2}_{T_n}\big]\le \E_{\P^{(N)}}[\|\check{\S}^{(i)}\|^p] \le C_p\E_{\P^{(N)}}\big[\langle \check{\S}^{(i)} \rangle^{p/2}_{T_n}\big].
\end{equation}
It follows that
\begin{equation}\label{eq: new_section markov inequality 2}
\E_{\P^{(N)}}\Big[\sum_{i=1}^{d+K}\langle \check{\S}^{(i)} \rangle^{p/2}_{T_n}\Big]\le K_1,
\end{equation}
where $K_1 := \Big(\frac{1}{c_p}\Big)\Big(\Big(\frac{p}{p-1}\Big)^{p}\sum_{i=1}^d\int_{[0,\infty)} x^p\mu^{(i)}_n(\td x)+K\kappa^p\Big)$.

In the following we want to modify $\check{\S}$ on
$$\tilde{\info} := \{S\in \info\,:\, \check{\S}(S) \text{ admits quadratic variation }\} = \{S\in \info\,:\, S \text{ admits quadratic variation }\}$$ 
using time change technique to obtain another process $\ddot{\S}$, the law of which is in $\MSUP_{\vec{\mu},\set,1/N}$. In fact, $\ddot{\S}$ is a time change of $\check{\S}$ on each interval $[\sigma_{i,j}(\ST), \sigma_{i,j+1}(\ST))$. Then by continuity of $G$, it follows that
\begin{equation*}
|G(\check{\S}(S)) - G(\ddot{\S}(S))|\le f_e(2^{-D+9}) \qquad  \forall\, S\in \tilde{\info}\cap \{\tilde{S}\in \info\,:\,m^{(D-8)}_n(\ST(\tilde{S})) \le \Theta\}.
\end{equation*}
This, together with \eqref{eq: new_section_markov_inequality_imply}  and the fact that $\P(\tilde{\info}) = 1$ for any $\P\in \MM_{\info}$, yields 
\begin{align*}
\big|\E_{\P^{(N)}}[G(\check{\S}) - G(\ddot{\S})]\big|\le&\, f_e(2^{-D+9}) + 2\kappa\P^{(N)}(\{S\in \info:\, m^{(D-8)}_n(\ST(S))\ge \Theta+1\})\\
\le&\,f_e(2^{-D+9}) + \frac{4\kappa+4}{D}.
\end{align*}

Hence, by \eqref{eq: new_section_original_tilde_difference} and \eqref{eq: new_section_tilde_check_difference}, 
\begin{equation}
\big|\E_{\P^{(N)}}[G(\S) - G(\ddot{\S})]\big|\le 5f_e(2^{-D+9})  + \frac{2Ln\|\S\|}{D}+ \frac{8\kappa+8}{2^D}  +\frac{2LnV^{(p)}_{\mu_n, \info}(\|\S\|)}{D}. \label{eq: new_section_key_difference}
\end{equation}
First, for every $i,j,k$, define  $\rho^{(i,j,k)}:\Omega \to [0,T_n]$ by $ \rho^{(i,j,k)}(S) = \sigma_{i,j}(\check{\S}(S))+\delta(1-2^{-k+1})$ $\forall\, S\in \Omega$. Then, $\forall\,i = 1,\ldots,n$, $j = 0,1,\ldots,$ and $k = 1,2,\ldots$, consider change of time $\theta^{(i,j,k)}:\info \times [0,T_n] \to [0,T_n]$ defined as follows: if $S\in \tilde{\info}$, $\theta^{(i,j,k)}_t(S) = t$ $\forall\, t\le \rho^{(i,j,k)}(S)$ and for $t>\rho^{(i,j,k)}(S)$   
\begin{align*}
\theta^{(i,j,k)}_t(S) =& \inf\{ u\ge \rho^{(i,j,k)}(S)\,:\sum_{l=1}^{d+K}\big(\langle\check{\S}^{(l)}(S)\rangle_u-
\langle\check{\S}^{(l)}(S)\rangle_{\rho^{(i,j,k)}}\big)> 2^k(t-\rho^{(i,j,k)})/\delta\}\wedge \rho^{(i,j,k+1)} \wedge \sigma_{i,j+1}(\check{\S}(S)),
\end{align*}
otherwise $\theta^{(i,j,k)}_t(S) = t$ on $[0,T_n]$. 

Now, by considering $\ddot{\S}$ -- a time change of $\check{\S}$ via $\theta^{(i,j,k)}$'s, defined by $\ddot{\S}_t := \check{\S}_{(\theta^{(i,j,k)}_t(\S))^{-1}}$ on $[\rho^{(i,j,k)}(\S), \rho^{(i,j,k+1)}(\S))$ $\forall i,j,k$, we see that for any $i,j,k$ and any $S\in \tilde{\info}$ the quadratic variation of $\ddot{\S}(S)_t$ grows linearly at rate $2^k$ on $[\rho^{(i,j,k)}(S), \rho^{(i,j,k+1)}(S))$ with $\rho^{(i,j,k+1)}(S) - \rho^{(i,j,k)}(S) = 2^{-k}$ if $\sigma_{i,j+1}(\ST(S)) - \sigma_{i,j}(\ST(S)) > 0$ and $0$ otherwise. It follows that $\ddot{\S}$ is a continuous process on $\tilde{\info}$ and furthermore
\begin{equation*}
\sum_{l=1}^{d+K}\big(\langle\ddot{\S}^{(l)}(S)\rangle_t - \langle\ddot{\S}^{(l)}(S)\rangle_s\big) \le 2^k|t-s|/\delta \;\;  \quad \forall s,t \text{ s.t.\ } \sigma_{i,j}(\ST(S))\le s \le t\le \sigma_{i,j+1}(\ST(S)),
\end{equation*} 
whenever $S\in \tilde{\info}$ is such that $\sum_{i=1}^{d+K}\langle \check{\S}^{(i)}(S) \rangle_{T_n} \le k$. Therefore, on $\{S\in \tilde{\info},\,:\,\sum_{i=1}^{d+K}\langle \check{\S}^{(i)}(S) \rangle_{T_n} \le k\}$,
\begin{equation}
\sum_{l=1}^{d+K}\big(\langle\ddot{\S}^{(l)}\rangle_t - \langle\ddot{\S}^{(l)}\rangle_s\big) \le 2^{k+1}|t-s|/\delta \quad \forall s,t\in [0,T_n] \text{ with } |t-s|\le \delta. \label{new section: key_eq}
\end{equation}
Hence by Markov inequality
\begin{align*}
\P^{(N)}\Big(\sum_{i=1}^{d+K}\langle \ddot{\S}^{(i)} \rangle_{T_n} > k\Big) =& \P^{(N)}\Big(\sum_{i=1}^{d+K}\langle \check{\S}^{(i)} \rangle_{T_n} > k\Big)\\
\le& \sum_{i=1}^{d+K}\P^{(N)}\Big(\langle \check{\S}^{(i)} \rangle_{T_n} > k/(d+K)\Big)\\
\le&  \frac{\E_{\P^{(N)}}\big[\sum_{i=1}^{d+K}\langle \check{\S}^{(i)} \rangle^{p/2}_{T_n}\big](d+K)^{p/2}}{k^{p/2}}\le (d+K)^{p/2}K_1k^{-p/2}. 
\end{align*}
\textbf{Step 5: Tightness of measures through tightness of quadratic variation processes.}

Together with \eqref{new section: key_eq}, by \AzAs theorem, this implies that  
$\{\P^{(N)}\circ (\langle \ddot{\S} \rangle_{t})^{-1}\}_{N\in \N}$ is tight (in $\CC([0,T_n], \R^d)$. Then by Lemma 6.4.13 in \citet{jacod2002limit}, $\{\P^{(N)}\circ (\ddot{\S}_t)^{-1}\}_{N\in \N}$ is tight (in $\D([0,T_n],\R^{d})$), which by Theorem 6.3.21 in \citet{jacod2002limit} implies that $\forall \epsilon>0, \eta>0$, there are $N_0\in \N$ and $\theta>0$ with
\begin{align*}
N\ge N_0 \,\Rightarrow\, \P^{(N)}(w_{T_n}^{\prime}(\S, \theta)\ge \eta)\le \epsilon,
\end{align*}
where $w_{T_n}^{\prime}$ is defined by
\begin{align*}
w_{T_n}^{\prime}(S, \theta) = \inf\Big\{\max_{i\le r} \sup_{t_{i-1}\le s\le t< t_{i}}|S_t- S_s|:\, 0 = t_0<\ldots<t_r = T_n,\, \inf_{i<r}(t_{i}-t_{i-1})\ge \theta\Big\}. 
\end{align*}
Note that for $S$ such that $w_{T_n}^{\prime}(S, \theta) > 0$, there exist $t_0,\ldots,t_r$ with $0 = t_0<\ldots<t_r = T_n$ and $\inf_{i<r}(t_{i}-t_{i-1})\ge \theta$ such that
\begin{align*}
\max_{i\le r} \sup_{t_{i-1}\le s\le t< t_{i}}|S_t- S_s| \le 2w_{T_n}^{\prime}(S, \theta).
\end{align*}
which by continuity of $S$ implies that
\begin{align*}
w_{T_n}(S, \theta):=\sup\{|S_t- S_s|:\, 0\le s<t\le T_n,\, t-s\ge\theta\} \le 4w_{T_n}^{\prime}(S, \theta).
\end{align*}
Then we have
\begin{align*}
N\ge N_0 \,\Rightarrow\, \P^{(N)}(w_{T_n}(\ddot{\S}, \theta)\ge 4\eta)\le \epsilon,
\end{align*}
which then by Theorem 6.1.5 in \citet{jacod2002limit} implies that $\{\P^{(N)}\circ ( \ddot{\S}_{t})^{-1}\}$ is tight (in $\CC([0,T_n],\R^{d})$). 

\textbf{Step 6: Tightness gives exact duality.}

Then there exists a converging subsequence $\{\P^{(N_k)}\circ ( \ddot{\S}_{t})^{-1}\}$ such that $\P^{(N_k)}\circ ( \ddot{\S}_{t})^{-1} \to \P$ weakly for some probability measure $\P$ on $\Omega$. Consequently,
\begin{align*}
\lim_{k\to \infty}\E_{\P^{(N_k)}}[G(\ddot{\S})]=\E_{\P}[G(\S)].
\end{align*}
In addition, if $\P$ is an element of $\MSUP_{\vec{\mu}, \set}$, then
\begin{align}
\begin{split}
V^{(p)}_{\vec{\mu}, \set}(G)\le \widetilde{V}^{(p)}_{\vec{\mu}, \set}(G)\le& \lim_{N\to\infty}\sup_{\P\in\MM_{\vec{\mu}, \info, 1/N}}\E_{\P}\Big[G(\S)-\kappa 2^D\wedge\frac{\sqrt{m^{(D-8)}(\S)}}{2^{2D}}\Big]+c_2/2^D\\
\le& \liminf_{N\to \infty}\E_{\P^{(N)}}[G(\S)] + c_2/2^{D}\\
\le& \liminf_{N\to \infty}\E_{\P^{(N)}}[G(\ddot{\S})] + e(D)\\
\le& \lim_{k\to \infty}\E_{\P^{(N_k)}}[G(\ddot{\S})] + e(D)\\
\le& \E_{\P}[G(\S)] + e(D) \le P_{\vec{\mu}, \set}(G) + e(D). \label{eq: new_section_before_conclusion}
\end{split}
\end{align}
where $e(x) := 5f_e(2^{-x+9})  + \frac{2Ln\|\S\|}{x}+ \frac{c_2+ 8\kappa+8}{2^x}  +\frac{2LnV^{(p)}_{\mu_n, \info}(\|\S\|)}{x}$ and the third inequality follows from \eqref{eq: new_section_key_difference}.

It remains to argue that $\P$ is an element of $\MSUP_{\vec{\mu}, \set}$. First, it is straightforward to see that $\S$ is a $\P$--martingale and $\LL_{\P}(S_{T_i}) = \mu_i$ for any $i\le n$. To show that $\P(\{\S\in \set\}) = 1$, notice that by portemanteau theorem, for every $\epsilon>0$
$$\P(\{\S\in \overbar{\set^{\epsilon}}\}) \ge \limsup_{k\to \infty} \P^{(N_k)}(\{\S\in \overbar{\set^{\epsilon}}\}) \ge \limsup_{k\to \infty} \P^{(N_k)}(\{\S\in \set^{1/N_k}\}) =1.$$
Therefore, it follows from Remark \ref{lemma: limit of enlargement is itself} and monotone convergence theorem that 
\begin{equation*}
\P(\{\S\in \set\}) = \lim_{\epsilon>0}\P(\{\S\in \overbar{\set^{\epsilon}}\}) =1,
\end{equation*}
and hence $\P\in \MSUP_{\vec{\mu}, \set}$.

To conclude, as $D$ is arbitrary, \eqref{eq: new_section_before_conclusion} yields that
\begin{equation*}
V^{(p)}_{\vec{\mu}, \set}(G)\le P_{\vec{\mu}, \set}(G),
\end{equation*}
which then implies that
\begin{equation*}
\widetilde{V}^{(p)}_{\vec{\mu}, \set}(G)=V^{(p)}_{\vec{\mu}, \set}(G)=P_{\vec{\mu}, \set}(G)=\widetilde{P}_{\vec{\mu}, \set}(G).
\end{equation*}

\section{Discretisation of the dual}\label{section:discretisation of the dual}

This and the subsequent section, are devoted to the proof of $\eqref{eq:thm 2}$ which in turn implies Theorem \ref{theorem:Main}. The strategy of the proof is inspired by \citet{Yan} and proceeds via discretisation, of the dual side in this section and of the primal side in Section \ref{section:discretisation of the primal}. The duality between discrete counterparts is obtained using classical probabilistic results of \citet{follmer1997optional}.

\subsection{A discrete time approximation through simple strategies}\label{section:1}

The proof of $\eqref{eq:thm 2}$ is based on a discretisation method involving a discretisation of the path space into a countable set of piece-wise constant functions. These are obtained as a ``shift" of the ``Lebesgue discretisation" of a path. Recall from Definition \ref{defn:stopping time} that 
for a positive integer $N$ and any $S \in \Omega$, $\tau^{(N)}_0(S)=0$,  $m^{(N)}_0(S)=0$, 
\begin{eqnarray*}
\tau^{(N)}_k(S)=\inf\Big\{t\ge\tau^{(N)}_{k-1}(S):|S_t-S_{\tau^{(N)}_{k-1}(S)}|=\frac{1}{2^N} \Big\}\wedge T
\end{eqnarray*}
and $m^{(N)}(S)=\min\{k\in \N: \tau^{(N)}_k(S)=T\}$.

Now denote by $\AA_{N}$ the set of $\ts\in \AA$ for which we only allow trading in the risky assets to take place at the moments $0=\tau_0^{(N)}(S)<\tau_1^{(N)}(S)< \cdots<\tau_{m^{(N)}(S)}^{(N)}(S)=T$ and $|\ts|\le N$.  Set
\begin{align*}
\dualf^{(N)}(G):=\inf\Big\{x\,:\, \exists \ts\in \AA_{N} \text{ s.t.\ } 
\text{$\ts$ super-replicates $G-x$}\Big\}.
\end{align*}
Then it is obvious from the definition of $\dualf^{(N)}$ that $\dualf^{(N_1)}(G)\ge \dualf^{(N_2)}(G)\ge\dualf(G)$ for any $N_2\ge N_1$, and in fact, the following result states that $\dualf^{(N)}(G)$ converges to $\dualf(G)$ asymptotically.

\begin{theorem}\label{theorem:second main}
Under the assumptions of Theorem~\ref{theorem:Main},
\begin{equation*}
\lim_{N\to\infty}\dualf^{(N)}(G)=\dualf(G).
\end{equation*}
\end{theorem}

\begin{theorem}\label{thm: real_main_thm}
For any $\alpha, \beta \ge 0$, $D\in\N$
\begin{align*}
\dualf(G- \alpha\wedge (\beta \sqrt{m^{(D)}})) \le \primalf(G - \alpha\wedge (\beta \sqrt{m^{(D-2)}})),
\end{align*}
where $m^{(D)}$ is defined in Definition \ref{defn:stopping time}.
\end{theorem}

\subsection{A countable class of piecewise constant functions}\label{section:def}

In this section, we construct a countable set of piecewise constant functions which can give approximations to any continuous function $S$ to a certain degree. It will be achieved in three steps. The first step is to use the Lebesgue partition defined in the last section to discretise a continuous function into a piecewise constant function whose jump times are the stopping times. Due to the arbitrary nature of jump times and jump sizes,  $F^{(N)}(S)$, the piecewise constant function generated through this procedure, will take values in an uncountable set.  To overcome this, in the subsequent two steps, we restrict the jump times and the jump sizes to a countable set and hence define a class of approximating schemes.

\textbf{Step 1. }
Let $\tau_k^{(N)}(S)$ and $m^{(N)}(S)$ be defined as in Subsection~\ref{section:1}. To simplify notations, in this section we often suppress their dependences on $S$ and $N$ and write
\begin{align*}
 m=m^{(N)}(S), \quad \tau_k=\tau_k^{(N)}(S) \qquad \text{ for any } k,N.
\end{align*}

Our first naive approximation $F^{(N)}:\Omega\to \D([0,T],\R^{d+K})$ is as follows:
\begin{align}\label{eq: naive_piecewise_constant_approximation}
F^{(N)}_t(S)=
\sum\limits_{k=0}^{m-1}S_{\tau_{k}}\indicators{[\tau_k,\tau_{k+1})}(t)+S_{T} \indicators{\{T\}}(t) \;\; \text{ for $t\in [0, T]$, $S\in \Omega$. } 
\end{align}
Note that $F^{(N)}(\S)$ is piecewise constant and $\|F^{(N)}(\S)-\S\|\le 1/2^{N}$.

\textbf{Step 2. }
Define a map $\p^{(N)}:\R^d_+\to A^{(N)}:=\{2^{-N}k\,:\,k=(k_1,\ldots,k_{d+K})\in\N^{d+K}\}$ as
\begin{align*}
\p^{(N)}(x)_i:=2^{-N}\lceil 2^{N}x_i\rceil,\;\; i=1,\ldots,d+K.
\end{align*}
We then define our second approximation $\Fcs: \Omega\to \D([0,T],\R^{d+K})$ by
\begin{align*}
\Fcs_t(S)=&(S_0 - \p^{(N+1)}(S_{\tau_1}))+ \sum\limits_{k=0}^{m-2}\p^{(N+k+1)}(S_{\tau_{k+1}})
\indicators{[\tau_k,\tau_{k+1})}(t)\\
&+\p^{(N+m)}(S_{\tau_{m}})\indicators{[\tau_{m-1},T]}(t)   \qquad \text{ $t\in [0, T]$}.
\end{align*} 

\textbf{Step 3. }

We now construct the shifted jump times $\Th^{(N)}_k:\Omega\to \Q_+\cup\{T\}$. Firstly, set $\Th^{(N)}_0=0$. Then, for any $S\in \Omega$ and $k=1,\cdots,m^{(N)}(S)$ let
\begin{align*}
\Delta\Th^{(N)}_k = 
\begin{cases}
\frac{p_k}{q_k}\;\; \text{ with } (p_k,q_k)=\argmin\{p+q\,: (p,q)\in \N^2,\,\tau^{(N)}_{k-1}-\Th^{(N)}_{k-1}< \frac{p}{q}\le \Delta\tau^{(N)}_k\}  
&\mbox{if } k<m^{(N)}(S)\\
T-\Th^{(N)}_{m^{(N)}-1} &\mbox{otherwise, }
\end{cases}
\end{align*}
where $\Delta\tau^{(N)}_k:=\tau^{(N)}_k-\tau^{(N)}_{k-1}$. Lastly, define $\Th^{(N)}_k :=\sum_{i=1}^k\Delta \Th^{(N)}_i$.
Here we also suppress the dependences of these shifted jump times on $S$ and $N$ and write
\begin{align*}
\Th_k=\Th_k^{(N)}(S) \qquad \text{ for any } k,N.
\end{align*}

Clearly $0=\Th_0<\Th_1<\Th_2\cdots<\Th_{m}=T$, $\tau_{k-1}<\Th_k\le \tau_k$ $\forall\, k<m$ and
$\Th_{m}=\tau_{m}=T$. These $\Th$'s are the shifted versions of $\tau$'s, and are uniquely defined for any $S$. We are going to use $\Th$'s to define a class of approximating schemes.

We can define an approximation $\Fhs:\Omega\to \D([0,T],\R^{d+K})$ by
\begin{align*}
\Fhs_t(S)=&(S_0 - \p^{(N+1)}(S_{\tau_1}))+\sum\limits_{k=0}^{m-2}\p^{(N+k+1)}(S_{\tau_{k+1}})
\indicators{[\Th_k,\Th_{k+1})}(t)\\
&+\p^{(N+m)}(S_{\tau_{m}})\indicators{[\Th_{m-1},T]}(t)   \qquad \text{ $t\in [0, T]$}.
\end{align*} 

Notice that $\Fhs(\S)$ is piecewise constant and 
\begin{align}
\|\Fhs(\S)-\S\|\le& \|\Fhs(\S)-\Fcs(\S)\|+\|\Fcs(\S)-F^{(N)}(\S)\|+\|F^{(N)}(\S)-\S\|\nonumber\\
\le& \frac{2}{2^{N-1}}+\frac{2}{2^N}+\frac{1}{2^N}< \frac{1}{2^{N-3}}.\label{eq:error_estimation}
\end{align}

\begin{definition}
Let $\hat{\D}^{(N)}\subset  \D([0,T],\R^{d+K})$ be the set of functions $f=(f^{(i)})_{i=1}^{d+K}$ which satisfy the following,
\begin{enumerate}
\item for any $i=1,\ldots,d+K$, $f^{(i)}(0)=1$,
\item $f$ is piecewise constant with jumps at times $t_1,\cdots,t_{l-1}\in \Q_+$ for some $l<\infty$,
\\where $t_0=t_{l_0}=0<t_1<t_2<\cdots<t_{l-1}<T$,
\item for any $k=1,\ldots,l-1$ and $i=1,\ldots,d+K$, $f^{(i)}(t_{k})-f^{(i)}(t_{k-1})=j/2^{N+k}$, for $j\in\Z$ with $|j|\le 2^{k}$, 
\item $\inf_{t\in [0,T],\, 1\le i\le d+K}f^{(i)}(t)\ge -2^{-N+3}$,
\item $\|f^{(i)}\| \le \kappa + 1$ for $i = d+1, \ldots, d+K$, where $\kappa = \max_{1\le j\le K}\frac{\|X^{(c)}_j\|_{\infty}}{\PP(X^{(c)}_j)}$,
\item if $f^{(i)}(t_k) = -2^{-N+3}$ for some $i\le d+K$ and $k\le l-1$, then $f(t_j) = f(t_k)$ $\forall\, k< j<l$,
\item if $f^{(i)}(t_k) =  \kappa+1$ for some $i> d$ and $k\le l-1$, then $f(t_j) = f(t_k)$ $\forall\, k< j<l$.
\end{enumerate}
\end{definition}
It is clear that $\hat{\D}^{(N)}$ is countable.

\subsection{A countable probabilistic structure}\label{subsection:probabilistic}
Let $\hat{\Omega}:=\D([0,T],\R^{d+K})$ be the space of all right continuous functions $f:[0,T]\to \R^{d+K}$ with left-hand limits. Denote by $\SH=(\SH_t)_{0\le t \le T}$ the canonical process on the space $\hat{\Omega}$.

The set $\DDO$ is a countable subset of $\hat{\Omega}$. There exists a local martingale measure $\PN$ on $\hat{\Omega}$ which satisfies $\PN(\DDO)=1$ and $\PN(\{f\})>0$ for all $f\in \DDO$. In fact, such a local martingale measure $\PN$ on $\DDO$ can be constructed `by hand'. Indeed, we can construct a continuous Markov chain that undergoes transitions in the finite number of allowed values in the way that the mean is preserved, with jump times decided via an exponential clock. Let
$\hat{\F}^{(N)}:=\{\hat{\F}^{(N)}_t\}_{0\le t\le T}$ be the filtration generated by the process $\SH$ and satisfying the usual assumptions (right continuous and contains $\PN$-null sets). 

In the last section, we saw definitions of $\hat{\tau}^{(N)}_k$ on $\Omega$. Here we extend their definitions to $\bigcup_{N\in\N}\DDO$. Define the jump times by setting $\Th_0(\SH)=0$ and for $k>0$,
\begin{equation}
\Th_k(\SH)=\inf\big\{t>\Th_{k-1}(\SH):\SH_t\neq\SH_{t-} \big\}\wedge T.\label{eq:new stopping times}
\end{equation} 
Next we introduce the random time before $T$ 
\begin{align*}
m(\SH):=\min\{k:\Th_k(\SH)= T\}.
\end{align*}
Observe that for $S\in \Omega$, $\Fh^{(N)}(S)\in\DDO$, $\Th_k(\Fh^{(N)}(S))=\Th_k(S)$ for all $k$ and $m(\Fh^{(N)}(S))=m^{(N)}(S)$.
It follows that the definitions are consistent.

In this context, a trading strategy $(\dts_t)_{t=0}^{T}$ on the filtered probability space $(\hat{\Omega},\hat{\F}^{(N)},\PN)$ is a predictable stochastic process. 
Thus, $\dts$ is a map from $\D([0,T],\R^{d+K})$ to $\DD([0,T],\R^{d+K})$. Choose $a\in \DD([0,T],\R^{d+K})$ such that $a\not\in \dts(\DDO)$ and then define a map $\phi:\D([0,T],\R^{d+K})\to \DD([0,T],\R^{d+K})$ by $\phi(\hat{S})=\dts(\hat{S})$ if $\hat{S}\in \DDO$, and equal to $a$ otherwise. Since $\PN$ has full support on $\DDO$, $\dts=\phi(\SH)$ $\PN$-a.s.. In particular, for any $A\in \BB(\R)$, the symmetric difference of  $\{\dts_t\in A\}$ and $\{\phi(\SH)_t\in A\}$ is a null set for $\PN$. Thus $\phi$ is a predictable map. Furthermore, since $\PN$ charges all elements in $\DDO$, for any $\upsilon, \tilde{\upsilon}\in \D([0,T],\R^{d+K})$ and $t\in[0,T]$. 
\begin{eqnarray*}
\upsilon_u=\tilde{\upsilon}_u \quad \forall u\in [0,t) \quad \Rightarrow \quad \phi(\upsilon)_t=\phi(\tilde{\upsilon})_t.
\end{eqnarray*} 
Indeed, suppose these exist $t\in[0,T]$ and $\upsilon,\tilde{\upsilon}\in \DDO$ such that $\upsilon_u=\tilde{\upsilon}_u$ for all $u\in [0,t)$ and $\phi(\upsilon)_t\neq\phi(\tilde{\upsilon})_t$. Since $\dts$ is predictable, we have
$$\hat{\F}^{(N)}_{t-} \ni \{\dts_t=\phi(v)_t\}\cap\{\S_u=v_u, u<t\}=\{\dts_t=\phi(v)_t\}\cap\{\S_u=v_u, u<t\}\cap \{\S_t\neq \tilde v_t\},$$
which is a contradiction since $\{\dts_t=\phi(v)_t\}\cap\{\S_u=v_u, u<t\}$ is not a null set and hence not in $\hat{\F}^{(N)}_{t-}$.
We conclude that any predictable process $\dts$ has a version $\phi$ that is progressively measurable in the sense of \eqref{eq:pm}. In what follows we always take this version.

In this section, we formally define the probabilistic super-replicating problem and later build a connection between the probabilistic super-replication problem on the discretised space and the path-wise discretised robust hedging problem. For the rest of the section, we write $\int_{t_1}^{t_2}$ to mean $\int_{(t_1,t_2]}$.

\vspace{5mm}

As $G$ is defined only on $\Omega$, to consider paths in $\hat{\Omega}$, we need to extend the domain of $G$ to $\hat{\Omega}$. For most of the financial contracts, the extension is natural. However, here we pursue a general approach. We first define a projection function $\proj:\hat{\Omega}\to \CC([0,T],\R^{d+K})$ by
\begin{align*}
\proj(\hat{S}) = 
\begin{cases}
 \hat{S} &\mbox{ if $\hat{S}$ is continuous} \\
 \sum_{k=0}^{m(\hat{S})-1}\Big(\frac{\hat{S}_{\Th_{k+1}}-\hat{S}_{\Th_{k}}}{\Th_{k+1}-\Th_k}(t-\Th_k)+\hat{S}_{\Th_k}\Big)\indicators{[\Th_k,\Th_{k+1})}(t) &\mbox{ if $\hat{S}\in \bigcup_{N\in\N}\DDO$ }\\
 \omega^{1} &\mbox{ otherwise,}
\end{cases}
\end{align*}
where $\omega^{1}$ is the constant path equal to $1$. In fact, when $\hat{S}\in \bigcup_{N\in\N}\DDO$, $\proj(\hat{S})$ is the minimum of $0$ and the linear interpolation function of 
$$\big((\Th_0(\hat{S}),\hat{S}_{\Th_0(\hat{S})}),\ldots, (\Th_{m(\hat{S})}(\hat{S}),\hat{S}_{\Th_{m(\hat{S})}(\hat{S})})\big).$$

We then can define $\GH:\hat{\Omega}\to \Omega$ via this explicit projection $\proj$ by $\GH(\hat{S}) = G(\proj(\hat{S})\vee 0)$, where
$\hat{S}\vee 0:= \big((\hat{S}^{(1)}_t\vee 0, \ldots, \hat{S}^{(d+K)}_t\vee 0)\big)_{0\le t\le T}$ for any $S\in \hat{\Omega}$. 

Note that $G$ and $\GH$ are equal on $\Omega$. In addition, for every $N\in \N$ and $\hat{S}\in \DDO$, we have 
\begin{equation}
\|\proj(\hat{S}) - \hat{S}\|\le 2^{-N+1}.\label{eq:error_estimation_new_2}
\end{equation}
Therefore, we can deduce that
\begin{align}
\|\proj(\Fh(S))\vee 0 - S\|\le& \|\proj(\Fh(S))\vee 0 - \Fh(S)\vee 0\| + \|\Fh(S)\vee 0 - S\|\nonumber\\
\le& 2^{-N+1}+2^{-N+3} \quad \forall S\in \Omega.\label{eq:error_estimation_new}
\end{align}
where the last inequality follows from \eqref{eq:error_estimation} and \eqref{eq:error_estimation_new_2}.

Similarly, for each $D\in \N$, we define $\hat{m}^{(D)}:\hat{\Omega}\to \Omega$ by $\hat{m}^{(D)}(\hat{S}) = m^{(D)}(\proj(\hat{S})\vee 0)$. Then by Remark \ref{remark:stopping time} and \eqref{eq:error_estimation_new}, when $N$ is sufficiently large,
\begin{equation} \label{eq: number of stopping times}
\hat{m}^{(D-2)}(\hat{F}^{(N)}(S))) \le m^{(D)}(S) \quad \forall\,S\in \Omega.
\end{equation}

\begin{definition}\label{Definition:2}
\
\begin{enumerate}
\item $\dts:\hat{\Omega}\to \DD([0,T],\R^{d+K})$ is $\PN$-admissible if $\dts$ is predictable and bounded by $N$, and the stochastic integral $(\int_0^{t}\dts_u(\SH)\cdot d\SH_u)_{0\le t\le T}$ is well defined under $\PN$, satisfying that $\exists\, M >0$ such that 
\begin{eqnarray}
\int_0^t\dts_u(\SH)\cdot d\SH_u\ge -M \quad \PN-\text{a.s.,} \quad t \in [0,T).
\end{eqnarray}
\item An admissible strategy $\dts$ is said to $\PN$-\textsl{super-replicate} $\GH$ if
\begin{equation}
\int_0^{T}\dts_u(\SH)\cdot d\SH_u\ge \GH(\SH), \quad \PN-\text{a.s..}
\end{equation}
\item The super-replicating cost of $\GH$ is defined as
\begin{align*}
\dualt:=\inf\{x\,:\, \exists \dts \text{ s.t.\ } \dts \text{ is $\PN$-admissible and $\PN$-super-replicates $\GH-x$}\}
\end{align*}
\end{enumerate}
\end{definition}

For the rest of the section we will establish connections between probabilistic super-hedging problems and discretised robust hedging problems. Our reasoning is close to the one in \citet{Yan}.

\begin{definition}\label{def:tradingstrategy}
Given a predictable stochastic process $(\dts_t)_{t=0}^{T}$ on $(\hat{\Omega},\hat{\F}^{(N)},\PN)$, we define $
\ts^{(N)}:\Omega\to \D([0,T],\R^{d+K})$
by
\begin{equation}
\ts^{(N)}_t(S):=\sum_{k=0}^{m-1}\dts_{\Th_{k}}(\Fhs(S))\indicators{(\tau_{k},\tau_{k+1}]}(t), \label{cts}
\end{equation}
where $\tau_k=\tau_k^{(N)}(S)$, $m=m^{(N)}(S)$ are given in Definition~\ref{defn:stopping time} and $\Th_k=\Th_k(\Fhs(S))$ are given in \eqref{eq:new stopping times}.
\end{definition}

\vspace{5mm}

\begin{lemma}\label{adm}
For any admissible process $\dts$ in the sense of Definition \ref{Definition:2}, $\ts^{(N)}$ defined in \eqref{cts} is progressively measurable in the sense of \eqref{eq:pm}.
\end{lemma}
\begin{proof}
To see $\ts^{(N)}$ is progressively measurable, we need to show 
\begin{equation*}
\ts^{(N)}_t(\omega)=\ts^{(N)}_t(\upsilon).
\end{equation*}
for any $\omega,\upsilon\in \Omega$ such that $\omega_u=\upsilon_u$  $\forall \, u\le t$ for some $t\in(0,T]$, the case $t=0$ being true by definition. Let $t\in (0,T]$ and set
\begin{equation}
k_t(\omega)=k_t^{(N)}(\omega):=\min\{i\ge 1:\tau_i^{(N)}(\omega)\ge t\}-1.
\end{equation}
It is clear that $k_t(\omega)=k_t(\upsilon)$, $\tau_{k_t(\omega)}(\omega)=\tau_{k_t(\upsilon)}(\upsilon)$ and $\omega_u=\upsilon_u$ for all $u\le \tau_{k_t(\omega)}(\omega)$. \\
Write $\theta:=\tau_{k_t(\omega)}(\omega)$. It follows from the definition of $\Fhs$ and $\Th$'s that
\begin{eqnarray*}
&\Th_{k_t(\Fhs(\omega))}(\Fhs(\omega))=\Th_{k_t(\Fhs(\upsilon))}(\Fhs(\upsilon)),\\
&\Fhs_u(\omega)=\Fhs_u(\upsilon) \quad \forall u\in[0,\theta).
\end{eqnarray*}
From \eqref{cts},
\begin{equation}
\ts^{(N)}_t(\omega)=\dts_{\theta}\big(\Fhs(\upsilon)\big),\qquad
\ts^{(N)}_t(\upsilon)=\dts_{\theta}\big(\Fhs(\upsilon)\big).
\end{equation}
Therefore, by the progressive measurability of $\dts$ as argued above, we conclude that $\ts^{(N)}_t(\omega)=\ts^{(N)}_t(\upsilon)$.
\end{proof}

\vspace{5mm}

The following theorem is crucial. It states that the probabilistic super-replicating value is asymptotically larger than the value of the discretised robust hedging problem. Recall that $\lambda_{\info}(\omega):=\inf_{\upsilon\in \info}\|\omega-\upsilon\|\wedge 1$.

\begin{theorem}\label{thm:duals}
For uniformly continuous and bounded $G$, $\alpha, \beta \ge 0$ and $D\in\N$, we have 
\begin{equation}\label{eq:n-marginal}
\liminf_{N\to \infty}\dualf^{(N)}(G(\S) - \alpha\wedge (\beta\sqrt{m^{(D)}(\S)}))\le \liminf_{N \to \infty}\dualt\Big(\GH(\SH) - \alpha\wedge (\beta\sqrt{\hat{m}^{(D-2)}(\SH)}) - N\lambda_{\info}(\SH)\Big).
\end{equation}
\end{theorem}
\begin{proof}
See Appendix~\ref{appendix:duals}.
\end{proof}

\subsection{Duality for the discretised problems}\label{section:duality for discretised problem}

\begin{definition}\label{def:mm_extended}
\
\begin{enumerate}
\item Let $\PT$ be the set of all probability measures $\hat{\Q}$ which are equivalent to $\PN$. 
\item For any $\kappa\ge 0$, denote $\MNS(\kappa)$ by the set of all probability measures $\hat{\Q}\in \PT$ such that 
$$ \hat{\Q}\big(\{\omega\in \hat{\Omega}\,:\,\inf_{\upsilon\in \info}\|\hat{\S}(\omega)-\upsilon\|\ge 1/N\}\big)\le \frac{\kappa}{N}$$
and
$$ \E_{\hat{\Q}}\bigg[\sum_{k=1}^{m(\SH)}\sum_{i=1}^{d+K}|\E_{\hat{\Q}}[\SH^{(i)}_{\Th_{k}}|\hat{\F}_{\Th_{k-}}]
-\SH^{(i)}_{\Th_{k-1}}|\bigg]\le \frac{\kappa}{N}, $$
where $\Th_{k}=\Th_{k}(\SH)$ and $m=m(\SH)$ are as defined in \eqref{eq:new stopping times}.
\end{enumerate}
\end{definition}

\begin{lemma}\label{lemma:lemma4.10}
Suppose $\GH$ is bounded by $\kappa-1$ and $\MSUP_{\info}\neq \emptyset$. Then, there are at most finitely many $N\in \N$ such that $\MNS(2\kappa)= \emptyset$ and 
\begin{equation}
\liminf_{N \to \infty}\dualt\Big(\GH(\SH)-N\lambda_{\info}(\SH)\Big)\le \liminf_{N \to \infty}\sup_{\hat{\Q}\in \MNS(2\kappa)}\E_{\hat{\Q}}[\GH(\SH)].
\end{equation}
\end{lemma}
\begin{proof}
Since for any $\hat{\Q}\in \PT$ the support of $\hat{\Q}$ is $\DDO$, of which elements are piece-wise constant, the canonical process $\SH$ is therefore a semi-martingale under $\hat{\Q}$. Moreover, it has the following decomposition, $\SH=\hat{M}^{\hat{\Q}}+\hat{A}^{\hat{\Q}}$ where
\begin{align*}
&\hat{A}^{\hat{\Q}}_t=\sum_{k=1}^{m(\SH)}\Big[\E_{\hat{\Q}}[\SH_{\Th_{k}}|\hat{\F}_{\Th_{k}-1}]-\SH_{\Th_{k-1}}\Big]\indicators{[\Th_{k},\Th_{k+1})}(t), \;\; t<T,\\
&\hat{A}^{\hat{\Q}}_{T}:=\lim_{t \uparrow T} \hat{A}_t^{\hat{\Q}}
\end{align*}
is a predictable process of bounded variation and $\hat{M}^{\hat{\Q}}$ is a martingale under $\hat{\Q}$. Then, similar to \citet{dolinsky2014martingale}, it follows from Example 2.3 and Proposition 4.1 in \citet{follmer1997optional} that
\begin{align}
\dualt\Big(\GH(\SH)-N\lambda_{\info}(\SH)\Big)
=\sup_{\hat{\Q}\in \PT}\E_{\hat{\Q}}\bigg[\hat{G}(\SH)- N\lambda_{\info}(\SH)-N\sum_{k=1}^{m(\SH)}\sum_{i=1}^{d+K}|\E_{\hat{\Q}}[\SH^{(i)}_{\Th_{k}}|\hat{\F}_{\Th_{k-}}]
-\SH^{(i)}_{\Th_{k-1}}| \bigg].\label{eq: optional_decomposition_constraint}
\end{align}
By Theorem \ref{thm:duals}, 
\begin{align*}
\liminf_{N \to \infty}\dualt\Big(\GH(\SH)-N\lambda_{\info}(\SH)\Big)
\ge& \liminf_{N\to \infty}\dualf^{(N)}(G)\ge \primalf(G)> -\kappa.
\end{align*}
Then, in \eqref{eq: optional_decomposition_constraint},  it suffices to consider the supremum over $\MNS(2\kappa)$. In particular, $\MNS(2\kappa)\neq \emptyset$ for $N$ large enough.

\end{proof}

\section{Discretisation of the primal}\label{section:discretisation of the primal}

\subsection{Approximation of Martingale Measures}\label{section:martingale}

Next, we show that we can lift any discrete martingale measure in $\MNS(c)$ to a continuous martingale measure in $\MSUPs_{\info}$ such that the difference of expected value of $G$ under this continuous martingale measure and the expected value of $\GH$ under the original discrete martingale measure is within a bounded error, which goes to zero as $N\to \infty$. Through this, we connect the primal problems on the discretised space to the approximation of the primal problems on the space of continuous functions asymptotically. 
\begin{proposition}\label{thm:second last approximation} 
Under the assumptions of Theorem~\ref{theorem:Main}, if $G$ and $X^{(c)}_i/\PP(X^{(c)}_i)$'s are bounded by $\kappa-1$ for some $\kappa\ge 1$, then for any $\alpha, \beta \ge 0$, $D\in\N$
\begin{equation}
\limsup_{N \to \infty}\sup_{\hat{\Q}\in\MNS(2\kappa+\alpha)}\E_{\hat{\Q}}[\GH(\SH)-\alpha\wedge (\beta\sqrt{\hat{m}^{(D)}(\SH)})]\le \sup_{\P\in \MSUPs_{\info}}\E_{\P}[G(\S)-\alpha\wedge (\beta\sqrt{m^{(D-2)}(\S)})].\label{eq:second last approximation_statement}
\end{equation}
\end{proposition}

\begin{proof}
Let $f_e:\R^{d+K}_+\to \R_+$ be the modulus of continuity of $G$, i.e.\ 
$$|G(\omega)-G(\upsilon)|\le f_e(|\omega-\upsilon|) \;\;\text{ for any }\omega, \upsilon\in \Omega$$ 
and $\lim_{x\searrow 0}f_e(x)=0$. Recall from Lemma \ref{lemma:lemma4.10} that $\MNS(2\kappa+2\alpha)\neq \emptyset$ for $N$ large enough. Hence, to show \eqref{eq:second last approximation_statement}, it suffices to prove that for any $\hat{\Q}\in\MNS(2\kappa+2\alpha)$
\begin{align}
\E_{\hat{\Q}}[\GH(\SH)-\alpha\wedge (\beta\sqrt{\hat{m}^{(D)}(\SH)})] \le \sup_{\P\in \MSUPs_{\info}}\E_{\P}[G(\S)-\alpha\wedge (\beta\sqrt{m^{(D-2)}(\S)})] + g(1/N), \label{eq:second last approximation_suffice_to_show}
\end{align}
for some $g:\R_+\to \R_+$ such that $\lim_{x\searrow 0}g(x) = 0$. We now fix $N$ and $\hat{\Q}\in\MNS(2\kappa+2\alpha)$ and prove \eqref{eq:second last approximation_suffice_to_show} in four steps.\\
\textbf{Step 1. } We will first construct a semi-martingale $\hat{Z}=\hat{M}+\hat{A}$ on a Wiener space $(\Omega^{W},\FF^W,P^W)$ such that
\begin{align}
\big|\E_{\hat{\Q}}[\GH(\SH)] - E^W[\GH(\hat{Z})]\big|\le \kappa2^{-N+1} \label{eq: step_1_conclusion_1}
\end{align}
and
\begin{align}
P^W\big(\{\omega\in \Omega^W\,:\,\inf_{\upsilon\in \info}\|\hat{M}(\omega)+\hat{A}(\omega)-\upsilon\|\ge 1/N\}\big)\le \frac{2\kappa+2\alpha}{N}+2^{-N}, \label{eq: step_1_conclusion_2}
\end{align}
where $\hat{M}$ is constructed from a martingale and both have piece-wise constant paths.

Since the measure $\hat{\Q}$ is supported on $\hat{\D}^{(N)}$, the canonical process $\SH$ is a pure jump process under $\hat{\Q}$, with a finite number of jumps $\hat{\Q}$-a.s. Consequently there exists a deterministic positive integer $m_0$ (depending on $N$) such that 
\begin{equation}
\hat{\Q}(\hat{m}(\SH)>m_0)<2^{-N}.\label{eq:truncation}
\end{equation}
It follows that 
\begin{equation}
|\E_{\hat{\Q}}[\GH(\SH)]-\E_{\hat{\Q}}[\GH(\SH^{\Th_{m_0}})]|\le \kappa2^{-N+1}.\label{eq:trio2}
\end{equation}
Notice that by definition of $\DDO$, the law of $\SH^{\Th_{m_0}}$ under $\hat{\Q}$ is also supported on $\DDO$.

Let $(\Omega^{W},\FF^W,P^W)$ be a complete probability space together with a standard $m_0+2$-dimensional Brownian motion $\Big\{W_t=(W_t^{(1)},\cdots, W_t^{(m_0+2)})\Big\}_{t=0}^{\infty}$ and the natural filtration $\FF^{W}_t=\sigma\{W_s|s\le t\}$. With a small modification to Lemma 5.1 in \citet{Yan}, we can construct a sequence of stopping times (with respect to Brownian filtration) $\sigma_1\le\sigma_2\le \cdots\le \sigma_{m_0}$ together with $\FF^{W}_{\sigma_i}$-measurable random variable $Y_i$'s such that 
\begin{eqnarray}\label{eq:seeap}
\LL_{P^W}\big((\sigma_1,\ldots,\sigma_{m_0}, Y_{1},\ldots,Y_{m_0})\big)=\LL_{\hat{\Q}}\big((\Th_1, \ldots, \Th_{m_0}, \SH_{\Th_1}-\SH_{\Th_0},\ldots,  \SH_{\Th_{m_0}}-\SH_{\Th_{m_0-1}})\big).
\end{eqnarray} 
(Detailed construction is provided in the Appendix~\ref{appendix:construction}.)

Define $X_i$ as
\begin{align*}
X_i=E^{W}[Y_i|\FF^W_{\sigma_{i-1}}\vee \sigma(\sigma_i)], \quad i=1,\ldots,m_0.
\end{align*}
Note that $|X_i|\le 2^{-N}$. Also by construction of $\sigma_i$'s and $Y_i$'s, we have
\begin{equation*}
E^{W}[Y_i|\FF^W_{\sigma_{i-1}}\vee \sigma(\sigma_i)] =  E^{W}[Y_i| \vec{\sigma}_{i},\, \vec{Y}_{i-1}],
\end{equation*}
where $\vec{\sigma}_i:=(\sigma_1,\ldots,\sigma_i)$, $\vec{Y}_i:=(Y_1,\ldots,Y_i)$ and $E^{W}$ is the expectation with respect to $P^W$.

From these, we can construct a jump process $(\hat{A}_t)_{t=0}^{T}$ by
\begin{align*}
\hat{A}_t= \sum_{j=1}^{m_0} X_j\indicators{[\sigma_j,T]}.
\end{align*}
In particular, for $k\le m_0$
$$\hat{A}_{\sigma_k} =  \sum_{j=1}^{k}X_j.$$
Set a martingale $(M_t)_{t=0}^{T}$ as
\begin{align}
M_t= 1+E^{W}\Big[\sum_{j=1}^{m_0}(Y_j-X_j)|\FF_t^{W}\Big], \;\;\;t\in [0,T].\label{eq:Browian_increment}
\end{align}
Since all Brownian martingales are continuous, so is $M$. Moreover, Brownian motion increments are independent and therefore,
\begin{equation*}
M_{\sigma_k}=1+\sum_{j=1}^{k}(Y_j-X_j), \quad P^{W}-\text{a.s.,}\quad k\le m.
\end{equation*}
We now introduce a stochastic process $(\hat{M}_t)_{t=0}^{T}$, on the Brownian probability space, by setting $\hat{M}_t=M_{\sigma_{k}}$ for $t\in[\sigma_{k},\sigma_{k+1})$, $k<m_0$ and $\hat{M}_t=\hat{M}_{\sigma_{m_0}}$ for $t\in [\sigma_{m_0},T]$. Note that as $|Y_i-X_i|\le 2^{-N+1}$, for any $k\le m_0$ and $t\le T$
\begin{align*}
&\big|\hat{M}_{t\wedge \sigma_{k+1}\vee \sigma_k} - M_{t\wedge \sigma_{k+1}\vee \sigma_k}\big| \\
=& \Big|\sum_{j=k+1}^{m_0}E^W[(Y_j-X_j)|\FF_{t\wedge \sigma_{k+1}\vee \sigma_k}^W]\Big| \\
=& \Big|\sum_{j=k+2}^{m_0}E^W\big[E^W[(Y_j-X_j)|\FF^W_{\sigma_{j-1}}\vee \sigma(\sigma_j)]\FF_{t\wedge \sigma_{k+1}\vee \sigma_k}^W\big] + E^W[Y_{k+1}-X_{k+1}|\FF_{t\wedge \sigma_{k+1}\vee \sigma_k}^W]\Big|\\
=&\Big| E^W[Y_{k+1}-X_{k+1}|\FF_{t\wedge \sigma_{k+1}\vee \sigma_k}^W]\Big|
\le E^W[|Y_{k+1}-X_{k+1}||\FF_{t\wedge \sigma_{k+1}\vee \sigma_k}^W]\le 2^{-N+1}
\end{align*}
and hence
\begin{align}
\|\hat{M}-M\|< 2^{-N+2}. \label{eq: continuous_martingale_estimate_1}
\end{align}
We also notice that $\hat{Z}=\hat{M}+\hat{A}$ satisfies $\hat{Z}_0 = \SH_0$ and
\begin{eqnarray*}
\LL_{P^W}\big((\sigma_1,\ldots,\sigma_{m_0}, Y_{1},\ldots,Y_{m_0})\big)=\LL_{\hat{\Q}}\big((\Th_1, \ldots, \Th_{m_0}, \SH_{\Th_1}-\SH_{\Th_0},\ldots,  \SH_{\Th_{m_0}}-\SH_{\Th_{m_0-1}})\big).
\end{eqnarray*} 
It follows that
\begin{align}
E^W[\GH(\hat{Z})] = \E_{\hat{\Q}}[\GH(\SH^{\tau_{m_0}})].\label{eq:trio2_appendix}
\end{align}
In particular, by \eqref{eq:trio2} we see that \eqref{eq: step_1_conclusion_1} holds and also by \eqref{eq:truncation} and definition of $\hat{M}$ and $\hat{A}$ \eqref{eq: step_1_conclusion_2} holds.

\textbf{Step 2. } 
We will shortly construct a continuous martingale $M^{\theta_0}$ from $M$ such that $M^{\theta_0}$ is bounded below by $-2^{-N+2}-N^{-\frac{1}{2}}$ and 
\begin{align}
|E^W[\GH(M^{\theta_0})]-\E_{\hat{\Q}}\big[\GH(\SH)\big]|\le c^2N^{-\frac{1}{2}} + 2f_e(N^{-\frac{1}{2}} + 2^{-N+2}) +2^{-N}.
\end{align}
As the law of $\hat{Z}$ is the same as $\SH^{m_0}$ under $\hat{\Q}$, it follows from the fact that $\hat{\Q}$ is supported on $\hat{\D}^{(N)}$ and any $f\in \hat{\D}^{(N)}$ is above $-2^{-N+3}$ that 
\begin{align}
\hat{Z}\ge -2^{N+3}, \quad \text{ $P^W$-a.s.}. \label{eq:trio4_appendix}
\end{align}
Then, by combining this with \eqref{eq:error_estimation_new_2} and \eqref{eq: continuous_martingale_estimate_1}, we can deduce that
\begin{align*}
\|\proj(\hat{Z}) - M\|\le& \|\proj(\hat{Z}) - \hat{Z}\vee 0\| + \|\hat{Z}\vee 0 - \hat{Z}\|+ \|\hat{Z} - M\|\\
\le& 2^{-N+1} + 2^{N+3}+\|\hat{M}-M\| + \|\hat{A}\| \\
\le& 2^{-N+4} + N^{-\frac{1}{2}} ,\quad \text{ whenever } \max_{1\le i\le d+K}\sum_{k=1}^{m_0}|X^{(i)}_{k}|\le N^{-\frac{1}{2}}.
\end{align*}
It follows that 
\begin{align*}
|\GH(M)-\GH(\hat{Z})|=& |G(M\vee 0)- G(\proj(\hat{Z})\vee 0)|\\
\le&  f_e(2^{-N+4}+N^{-\frac{1}{2}}),\quad \text{ whenever } \max_{1\le i\le d+K}\sum_{k=1}^{m_0}|X^{(i)}_{k}|\le N^{-\frac{1}{2}},
\end{align*}
where we use the fact that $\|\proj(\hat{Z})\vee 0 - M\vee 0\|\le \|\proj(\hat{Z}) - M\|$.
Hence, since $\GH$ is bounded by $\kappa$ 
\begin{align*}
\big|E^W[\GH(M)]-E^W[\GH(\hat{Z})]\big|\le f_e(2^{-N+4}+N^{-\frac{1}{2}})+2\kappa P^W\Big(\max_{1\le i\le d+K}\sum_{k=1}^{m_0}|X^{(i)}_{k}|> N^{-\frac{1}{2}}\Big).
\end{align*}
Note that 
\begin{align*}
X_k = E^{W}\big[Y_k| \vec{\sigma}_{k},\, \vec{Y}_{k-1}\big] \stackrel{(d)}{=}&\,\E_{\hat{\Q}}\big[\SH_{\Th_k} - \SH_{\Th_{k-1}}|\vec{\Th}_k,\, \vec{\Delta\SH}_{\Th_{k-1}}\big]\\ 
=&\, \E_{\hat{\Q}}\big[\SH_{\Th_k} - \SH_{\Th_{k-1}}|\hat{\F}_{\Th_{k-}}\big]
\end{align*}
where $\Delta\SH_k=\SH_{\Tt_k}-\SH_{\Tt_{k-}}$ for $k\le m_0$ and hence
\begin{align*}
E^{W}\Big[\sum_{i=1}^{d+K}\sum_{k=1}^{m_0}|X^{(i)}_{k}|\Big]=\E_{\hat{\Q}}\bigg[\sum_{k=1}^{m_0}\sum_{i=1}^{d+K}|\E_{\hat{\Q}}[\SH^{(i)}_{\Th_{k}}|\hat{\F}_{\Th_{k-}}]
-\SH^{(i)}_{\Th_{k-1}}|\bigg].
\end{align*}
By Markov inequality and definition of $\MNS(2\kappa)$, we have
\begin{align}
P^W\Big(\sum_{i=1}^{d+K}\sum_{k=1}^{m_0}|X^{(i)}_{k}|> N^{-\frac{1}{2}}\Big)
\le& \sqrt{N}E^{W}\Big[\sum_{i=1}^{d+K}\sum_{k=1}^{m_0}|X^{(i)}_{k}|\Big]\nonumber\\
\le& \sqrt{N}\E_{\hat{\Q}}\bigg[\sum_{k=1}^{m}\sum_{i=1}^{d+K}|\E_{\hat{\Q}}[\SH^{(i)}_{\Th_{k}}|\hat{\F}_{\Th_{k-}}]
-\SH^{(i)}_{\Th_{k-1}}|\bigg]\le 2\kappa N^{-\frac{1}{2}}. \label{eq:left_out_estimate_1}
\end{align}
Therefore, we have 
\begin{align}
|E^W[\GH(M)]-E^W[\GH(\hat{Z})]|\le f_e(2^{-N+4}+N^{-\frac{1}{2}})+4\kappa^2N^{-\frac{1}{2}}. \label{eq:lifting_measure_2}
\end{align}

By \eqref{eq: continuous_martingale_estimate_1}, \eqref{eq:trio4_appendix} and \eqref{eq:left_out_estimate_1}
\begin{align}
&P^W\big(\inf_{0\le t\le T}\min_{1\le i\le d+K}M^{(i)}_t> -2^{-N+4}-N^{-\frac{1}{2}}\text{ and }\max_{d\le i\le d+K}\|M^{(i)}\|< \kappa+1+2^{-N+2}+N^{-\frac{1}{2}}\big)\nonumber\\
\ge&\, 1-2\kappa N^{-\frac{1}{2}}. \label{eq:lifting_measure_3}
\end{align}
Hence a stopped process $M^{\theta_0}$, with $$\displaystyle\theta_0:=\inf\big\{t\ge 0: \min_{1\le i\le d+K}M^{(i)}_t\le -2^{-N+4}-N^{-\frac{1}{2}}\text{ or }\max_{d\le i\le d+K}\|M^{(i)}\|\ge \kappa+1+2^{-N+2}+N^{-\frac{1}{2}}\big\},$$ 
satisfies 
\begin{align}
|E^{W}[\GH(M)]-E^{W}[\GH(M^{\theta_0})]|\le 4\kappa^2 N^{-\frac{1}{2}}.\label{eq:martingale_stopped}
\end{align}
By \eqref{eq:trio2}, \eqref{eq:lifting_measure_2} and \eqref{eq:martingale_stopped}, it follows that
\begin{align}
&|E^W[\GH(M^{\theta_0})]-\E_{\hat{\Q}}\big[\GH(\SH)\big]| \nonumber\\
\le& |E^W[\GH(M^{\theta_0})]-E^W[\GH(M)]|+|E^W[\GH(M)]-E^W[\GH(\hat{Z})]|+
|\E_{\hat{\Q}}[\GH(\SH^{\Th_{m_0}})]-\E_{\hat{\Q}}[\GH(\SH)]|\nonumber\\
\le& 4\kappa^2N^{-\frac{1}{2}}+4\kappa^2N^{-\frac{1}{2}}+f_e(2^{-N+4}+N^{-\frac{1}{2}})+\kappa 2^{-N+1} .\label{eq:trio4}
\end{align}
In addition, by \eqref{eq: continuous_martingale_estimate_1} and  \eqref{eq:left_out_estimate_1} we can deduce from \eqref{eq: step_1_conclusion_2} that
\begin{align*}
P^W\big(\{\omega\in \Omega^W\,:\,\inf_{\upsilon\in \info}\|M^{\theta_0}(\omega)-\upsilon\|\ge 1/N + N^{-\frac{1}{2}}+2^{-N+2}\}\big)\le \frac{2\kappa+2\alpha}{N}+2^{-N}+2\kappa N^{-\frac{1}{2}}
\end{align*}
which for simplicity we notice that it implies for $N$ large enough 
\begin{align}
P^W\big(\{\omega\in \Omega^W\,:\,\inf_{\upsilon\in \info}\|M^{\theta_0}(\omega)-\upsilon\|\ge 4\kappa N^{-\frac{1}{2}}\}\big)\le 4\kappa N^{-\frac{1}{2}} . \label{eq: step_2_conclusion_1}
\end{align}
Similarly, by \eqref{eq: continuous_martingale_estimate_1} and  \eqref{eq:left_out_estimate_1}, we have
\begin{align}
P^W( \|\hat{Z} - M^{\theta}\| \ge 2^{-N+2} + N^{-\frac{1}{2}})\le 2\kappa N^{-\frac{1}{2}}. \label{eq: step_2_conclusion_final}
\end{align}

\textbf{Step 3. }  The next step is to modify the martingale $M^{\theta_0}$ in such way that $\Gamma$, the new continuous martingale, is non-negative.

Write $\epsilon_{N}=2^{-N+4}+N^{-\frac{1}{2}}$ and define a non-negative $\FF^W_{T}$-measurable random variable $\Lambda$ by
$\Lambda=(M_{T\wedge{\theta_0}}+\epsilon_{N})/(1+\epsilon_{N})$.
Then 
\begin{align*}
|\Lambda-M_{T_n\wedge\theta_0}|=\Big|\epsilon_{N}\frac{1-M_{T\wedge\theta_0}}{1+\epsilon_{N}}\Big|
\le\epsilon_{N}(1+|M_{T\wedge\theta_0}|).
\end{align*}
Note that for any $i>d$ $\|\Lambda^{(i)}\|\le \kappa+1+2^{-N+2}+N^{-\frac{1}{2}}+\epsilon_N\le \kappa+2$ for $N$ large enough. We now construct a continuous martingale from the $\Lambda$ by taking conditional expectations: 
\begin{equation*}
\Gamma_t=E^{W}[\Lambda|\FF_t^W], \quad t\in[0,T]
\end{equation*}
and $\Lambda\ge 0$ implies that $\Gamma$ is non-negative and $\Gamma^{(i)}_0 = 1$ $\forall\, i\le d+K$. Hence $\P^{(N)} := P^W\circ (\Gamma_t)^{-1}\in \MSUPs$.

We first notice that
\begin{align*}
E^{W}[|M_{T\wedge\theta_0}^{(i)}|]=\E^{W}[M_{T\wedge\theta_0}^{(i)}-2(M^{(i)}_{T\wedge\theta_0})^-]\le \E^{W}[M_{T\wedge\theta_0}^{(i)}+2]=3 \quad \forall i=1,\ldots,d+K.
\end{align*}
Then by Doob's martingale inequality
\begin{align}
P^W(\|\Gamma-M^{\theta_0}\|\ge \epsilon_{N}^{1/2})\le&\; \epsilon_{N}^{-1/2}\sum_{i=1}^{d+K}E^W[|\Lambda^{(i)}-M^{(i)}_{T\wedge\theta_0}|]\le \epsilon_{N}^{-1/2}4(d+K)\epsilon_{N}=4(d+K)\epsilon_{N}^{1/2}.\label{eq:trio5}
\end{align}
This together with \eqref{eq:trio4} yields
\begin{align*}
&\big|E^W[G(\Gamma)]-\E_{\hat{\Q}}[G(\SH)]\big|\\
\le& E^W\big[|G(\Gamma)-\GH(M^{\theta_0})|\big]+\big|E^W[\GH(M^{\theta_0})]-\E_{\hat{\Q}}[\GH(\SH)]\big|\\
\le& E^W\big[|G(\Gamma)-G(M^{\theta_0}\vee 0)|\indicator{\|\Gamma-M^{\theta_0}\|<\epsilon_{N}^{1/2}}\big]+8\kappa(d+K)\epsilon_{N}^{1/2}+8\kappa^2N^{-\frac{1}{2}}+f_e(2^{-N+4}+N^{-\frac{1}{2}})+\kappa 2^{-N+1}\\
\le& f_e(\epsilon_N^{1/2})+8\kappa(d+K)\epsilon_{N}^{1/2}+9\kappa^2\epsilon_{N}^{1/2}+f_e(\epsilon_N^{1/2}) \le 2f_e(\epsilon_N^{1/2}) + 17\kappa^2(d+K)\epsilon^{1/2}_{N}.
\end{align*}
Finally, we can deduce from \eqref{eq: step_2_conclusion_1} and \eqref{eq:trio5} that
\begin{align}
P^W\big(\{\omega\in \Omega^W\,:\,\inf_{\upsilon\in \info}\|\Gamma(\omega)-\upsilon\|\ge 4\kappa N^{-\frac{1}{2}}+\epsilon_N^{1/2}\}\big)\le 4\kappa N^{-\frac{1}{2}} + 4\epsilon^{1/2}_{N}. \label{eq: step_3_conclusion_1}
\end{align}
and 
\begin{align}
P^W( \|\hat{Z} - \Gamma\| \ge 4\kappa N^{-\frac{1}{2}}+\epsilon_N^{1/2})\le 4\kappa N^{-\frac{1}{2}} + 4\epsilon^{1/2}_{N}. \label{eq: step_3_conclusion_final}
\end{align}

\textbf{Step 4. }  The last step is to construct a new process $\tilde{\Gamma}$ from $\Gamma$ such that the law of $\tilde{\Gamma}$ under $P^W$ is an element of $\MSUPs_{\info}$.

We write $\eta_N = 4\kappa N^{-\frac{1}{2}} + 4\epsilon^{1/2}_{N}$ and
\begin{align*}
p^{(N)}_{i} := \E_{\P^{(N)}}[X^{(c)}_i(\S_{T}^{(1)},\ldots,\S_{T}^{(d)})] 
\end{align*}
for any $i=1,\ldots K$, and define $\tilde{p}^{(N)}_{i}$'s by
\begin{align*}
\tilde{p}^{(N)}_{i} = \frac{\PP(X^{(c)}_i) - (1- \sqrt{\eta_N})p^{(N)}_{i}}{\sqrt{\eta_N}}.
\end{align*}
Note that as $\E_{\P^{(N)}}[X^{(c)}_i(\S_{T}^{(1)},\ldots,\S_{T}^{(d)})] =\PP(X^{(c)}_i)$, we can deduce that
\begin{align*}
&\big|\PP(X^{(c)}_i) - p^{(N)}_{i}\big| \\
\le& E^W[|X^{(c)}_i(\Gamma^{(1)}_T,\ldots,\Gamma^{(d)}_T) - \PP(X^{(c)}_i)\Gamma^{(d+i)}_T|]\\
\le&  \PP(X^{(c)}_i)\eta_N+E^W\Big[|X^{(c)}_i(\Gamma^{(1)}_T,\ldots,\Gamma^{(d)}_T) - \PP(X^{(c)}_i)\Gamma^{(d+i)}_T|\indicator{|X^{(c)}_i(\Gamma^{(1)}_T,\ldots,\Gamma^{(d)}_T)/\PP(X^{(c)}_i)-\Gamma^{(d+i)}_T|>\eta_N}\Big]\\
\le& \PP(X^{(c)}_i)\eta_N+2(\kappa+1)\PP(X^{(c)}_i)\eta_N, \quad \forall i=1,\ldots,K.
\end{align*}
It follows immediately that
\begin{align}
\big|\tilde{p}^{(N)}_{i} - \PP(X^{(c)}_i)\big| =& \Big(\frac{1}{\sqrt{\eta_N}} -1\Big)\Big|\PP(X^{(c)}_i) - p^{(N)}_{i}\Big| \nonumber\\
\le& \frac{2(\kappa+1)\PP(X^{(c)}_i)\eta_N}{\sqrt{\eta_N}} =2(\kappa+1)\PP(X^{(c)}_i)\sqrt{\eta_N} \quad \forall i\le K. \label{eq: finite_puts_1_appendix}
\end{align}
Then, it follows from Assumption \ref{assumption:key assumption} that when $N$ is large enough there exists a $\tilde{\P}^{(N)}\in\MSUPs_{\tilde{\info}}$ such that
\begin{align*}
\tilde{p}^{(N)}_{i} := \E_{\tilde{\P}^{(N)}}[X_i(\S_{T}^{(1)},\ldots,\S_{T}^{(d)})] \;\quad \forall i\le K.
\end{align*} 

Enlarge Wiener space $(\Omega^{W}, \FF^W, P^W)$ if necessary, then there are continuous martingales $\Gamma$ and $\tilde{M}$ which have laws equal to $\P^{(N)}$ and $\tilde{\P}^{(N)}$ respectively, and an $\FF^W_T$-measurable random variable $\xi\in\{0,1\}$ that is independent of $\Gamma$ and $\tilde{M}$, with 
 $$P^{W}(\xi = 1) = 1- \sqrt{\eta_N}\; \text{ and }\; P^{W}(\xi = 0)=\sqrt{\eta_N}. $$
Define $\FF^W_T$-measurable random variables $\tilde{\Lambda}^{(i)}$ by
\begin{align*}
&\tilde{\Lambda}^{(i)} = \Gamma_T^{(i)}\indicator{\xi = 1} + \tilde{M}_T^{(i)}\indicator{\xi = 0}\; \quad \forall i=1,\ldots,d,\\
&\tilde{\Lambda}^{(i)} = X_{i-d}(\tilde{\Lambda}^{(1)},\ldots,\tilde{\Lambda}^{(d)})/\PP(X^{(c)}_{i-1}) \; \quad \forall i>d.
\end{align*} 
We now construct a continuous martingale from $\tilde{\Lambda}$ by taking conditional expectations:
\begin{equation*}
\tilde{\Gamma}_t=E^{W}[\tilde{\Lambda}|\FF_t^W], \quad t\in[0,T].
\end{equation*}
It follows from the fact that $\xi$ is independent of $M$ and $\tilde{M}$
\begin{align*}
\tilde{\Gamma}^{(i)}_0 =& E^W[\tilde{\Gamma}^{(i)}_T|\FF_0^W]\\
=&(1-\sqrt{\eta_N})E^W[X_i(\Gamma_{T}^{(1)},\ldots,\Gamma_{T}^{(d)})/\PP(X^{(c)}_i) ]+ \sqrt{\eta_N} E^W[X_i(\tilde{M}_{T}^{(1)},\ldots,\tilde{M}_{T}^{(d)})/\PP(X^{(c)}_i)]\\
=& \frac{(1-\sqrt{\eta_N})p_i^{N}+\sqrt{\eta_N}\tilde{p}_i^{N}}{\PP(X^{(c)}_i)} =1 \;\;\quad \forall i>d
\end{align*}
and
\begin{align*}
\tilde{\Gamma}^{(i)}_0 =& E^W[\tilde{\Gamma}^{(i)}_T|\FF_0^W] = E^W[\tilde{\Lambda}^{(i)}_T|\FF_0^W] = (1-\eta_N)E^{W}[\Gamma^{(i)}_{T}]+ \eta_N E^{W}[\tilde{M}^{(i)}_{T}] = 1\;\quad \forall i\le d.
\end{align*}
Hence $\tilde{\P} := P^W \circ (\tilde{\Gamma}_t)^{-1} \in \MSUPs_{\info}$. Also by independence between $\xi$ and $(M, \tilde{M})$, we have
\begin{align*}
E^{W}[|\tilde{\Lambda}^{(i)} - \Gamma_{T}^{(i)}|] = \sqrt{\eta_N} E^{W}[|\tilde{M}_{T}^{(i)}- \Gamma_{T}^{(i)}|] \le 2\sqrt{\eta_N} \quad \forall i\le d
\end{align*}
and by \eqref{eq: step_3_conclusion_1}
\begin{align*}
P^W(|\Gamma_T - \tilde{\Lambda}^{(i)}|>\eta_N)\le \eta_N +\sqrt{\eta}\le 2\sqrt{\eta_N} \quad \forall i> d,
\end{align*}
which implies that 
\begin{align*}
E^W[|\tilde{\Lambda}^{(i)}- \Gamma^{(i)}_{T} |]=&\;2E^W[(\tilde{\Lambda}^{(i)}- \Gamma^{(i)}_{T})^+]-E^W[\tilde{\Lambda}^{(i)}- \Gamma^{(i)}_{T}]\\
=&\; 2E^W[(\tilde{\Lambda}^{(i)}- \Gamma^{(i)}_{T})^+]\\
\le&\;  2\eta_N+2E^W\Big[\Lambda^{(i)}\indicator{|\tilde{\Lambda}^{(i)}- \Gamma^{(i)}_{T}|>\eta_N}\Big]\\
\le&\; 2\eta_N + 4(\kappa+2)\sqrt{\eta_N} \le 14\kappa\sqrt{\eta_N}, \quad \forall i=d+1,\ldots,K.
\end{align*} 
Then by Doob's martingale inequality
\begin{align*}
P^W(\|\tilde{\Gamma}-\Gamma\|\ge \kappa\eta_{N}^{1/4})\le&\; \frac{1}{\kappa\eta_{N}^{1/4}}\sum_{i=1}^{d+K}E^W[|\tilde{\Lambda}^{(i)}-\Gamma^{(i)}_{T}|]\le 14(d+K)\eta_{N}^{1/4}
\end{align*}
and hence
\begin{align*}
\big|\E_{\tilde{\P}}[G(\S)]-\E_{\P^{(N)}}[G(\S)]\big| =& \big|E^W[G(\tilde{\Gamma})-G(\Gamma)]\big|\\
\le& f_e(\kappa\eta_{N}^{1/4})+ E^W\Big[|G(\Gamma)-G(\Gamma)|\indicator{\|\tilde{\Gamma}-\Gamma\|\ge \kappa\eta_{N}^{1/4}}\Big]\\
\le& f_e(\kappa\eta_{N}^{1/4}) + 28\kappa(d+K)\eta_{N}^{1/4}.
\end{align*} 
In addition, we can decuce from \eqref{eq: step_3_conclusion_final} that
\begin{align}
P^W( \|\hat{Z} - \tilde{\Gamma}\| \ge \kappa\eta_{N}^{1/4} +4\kappa N^{-\frac{1}{2}}+\epsilon_N^{1/2})\le 4\kappa N^{-\frac{1}{2}} + 4\epsilon^{1/2}_{N} +14(d+K)\kappa\eta_{N}^{1/4}.\label{eq: step_4_conclusion_2}
\end{align}
Notice that when $N$ is sufficiently large such that $\kappa\eta_{N}^{1/4} +4\kappa N^{-\frac{1}{2}}+\epsilon_N^{1/2} < 2^{-D-1}$, on the event $\big\{\omega\in\Omega^{W}\,:\, \|\hat{Z}(\omega) - \tilde{\Gamma}(\omega)\| < \kappa\eta_{N}^{1/4} +4\kappa N^{-\frac{1}{2}}+\epsilon_N^{1/2} \text{ and } \hat{Z}(\omega)\in \DDO\big\}$, we can deduce from \eqref{eq:error_estimation_new_2} and \eqref{eq:trio4_appendix} that
\begin{align*}
|\proj(\hat{Z})\vee 0 - \tilde{\Gamma}|\le&\, |\proj(\hat{Z})\vee 0 - \proj(\hat{Z})|+|\proj(\hat{Z}) - Z| + |Z - \tilde{\Gamma}|\\
<&\,2^{-N+3}+2^{-N+1} +2^{-D+1} \le 2^{-D}.
\end{align*}
and hence by Remark \ref{remark:stopping time} the inequality $\hat{m}^{(D)}(\hat{Z}) \ge m^{(D-2)}(\tilde{\Gamma})$ holds on $\big\{\omega\in\Omega^{W}\,:\, \|\hat{Z}(\omega) - \tilde{\Gamma}(\omega)\| < \kappa\eta_{N}^{1/4} +4\kappa N^{-\frac{1}{2}}+\epsilon_N^{1/2} \text{ and } \hat{Z}(\omega)\in \DDO\big\}$.

It follows that
\begin{align*}
\E_{\hat{\Q}}[\alpha\wedge (\beta\sqrt{\hat{m}^{(D)}(\SH)})] \ge& \E_{\hat{\Q}}[\alpha\wedge (\beta\sqrt{\hat{m}^{(D)}(\SH^{m_0})})]\\
=& E^W[\alpha\wedge (\beta\sqrt{\hat{m}^{(D)}(\hat{Z})})] \\
\ge& E^W[\alpha\wedge (\beta\sqrt{m^{(D-2)}(\tilde{\Gamma})})] - \alpha\Big( 4\kappa N^{-\frac{1}{2}} + 4\epsilon^{1/2}_{N} +14(d+K)\kappa\eta_{N}^{1/4}\Big).
\end{align*}
\end{proof}

\section{Appendix}

\subsection{Proof of Theorem~\ref{thm:duals}}\label{appendix:duals}
\begin{proof}
Fix $N\ge 6$. Choose $f_e:\R_+\to \R_+$ such that $|G|\le \kappa$, $|G(\omega)-G(\upsilon)|\le f_e(|\omega-\upsilon|)$ for any $\omega, \upsilon\in \Omega$ and $\lim_{x\to 0}f_e(x)=0$. Define $G^{(N)}:\Omega\to \R$ as
\begin{equation*}
G^{(N)}(S):=\GH(S)-f_e(2^{-N+4})-\frac{14(d+K)N}{2^N}.
\end{equation*}
and 
\begin{eqnarray*}
\dualf^{(N)}(G-\alpha\wedge (\beta\sqrt{m^{(D)}}))=\dualf^{(N)}(G^{(N)}-\alpha\wedge (\beta\sqrt{m^{(D)}}))+f_e(2^{-N+4})+\frac{14(d+K)N}{2^N}.
\end{eqnarray*}
Hence, to show \eqref{eq:n-marginal}, it suffices to show 
\begin{equation}\label{eq:n-marginal2} 
\dualf^{(N)}(G^{(N)}-\alpha\wedge (\beta\sqrt{m^{(D)}}))\le \dualt\Big(\GH-\alpha\wedge (\beta\sqrt{\hat{m}^{(D-2)}})- N\lambda_{\info}\Big).
\end{equation}
The rest of proof is structured to establish \eqref{eq:n-marginal2}. Given a probabilistic semi-static portfolio $\dts$ which super-replicates $\GH-\alpha\wedge (\beta\sqrt{\hat{m}^{(D-2)}})-N\lambda_{\info}-x$,  we will argue that the lifted progressively measurable trading strategy $\ts^{(N)}$ super-replicates $G^{(N)}-\alpha\wedge (\beta\sqrt{m^{(D)}})-x$ on $\info$. To simplify notations, throughout the rest of the proof, we fix $S\in \info$ and write $\Fh:=\Fh^{(N)}(S)$.

\textit{Super-replication:}

We first notice that for any $j<m-1$
\begin{align*}
&|(S_{\tau_{j+1}}-S_{\tau_{j}})-(\Fh_{\Th_j}-\Fh_{\Th_{j-1}})|\\
\le& |S_{\tau_{j+1}}-\Fh_{\Th_j}|+|S_{\tau_{j}}-\Fh_{\Th_{j-1}}|\le \frac{1}{2^{N+j+1}}+\frac{1}{2^{N+j}}=\frac{3}{2^{N+j+1}}.
\end{align*}
It follows that for any $k<m$,
\begin{align}
&\Big|\int_{0}^{\tau_{k}}\ts_u^{(N)}(S)\cdot dS_u-\int_{0}^{\Th_{k}} \dts_u(\Fh)\cdot d\Fh_u|\nonumber\\
\le & \Big|\sum_{j=0}^{k-1}\dts_{\Th_{j}}(\Fh)\cdot
(S_{\tau_{j+1}}-S_{\tau_j})-\sum_{j=0}^{k-1}\dts_{\Th_{j+1}}(\Fh)\cdot
(\Fh_{\Th_{j+1}}-\Fh_{\Th_{j}})\Big|\nonumber \\
\le &  \sum_{j=0}^{k-2} \Big|\dts_{\Th_{j+1}}(\Fh)\cdot\big((S_{\tau_{j+2}}-S_{\tau_{j+1}})- (\Fh_{\Th_{j+1}}-\Fh_{\Th_{j}})\big)\Big| + \frac{2(d+K)N}{2^{N-1}}\nonumber \\
\le & \sum_{j=0}^{\infty}\frac{N(d+K)}{2^{N+j+2}} + \frac{2(d+K)N}{2^{N-1}}\le \frac{5(d+K)N}{2^{N}}. \label{eq:integral_mimic_1}
\end{align}
In addition,
\begin{align}
&\Big|\int_{\tau_{m-1}}^{T}\ts_u^{(N)}(S)\cdot dS_u-\int_{\Th_{m-1}}^{T} \dts_u(\Fh)\cdot d\Fh_u\Big|\nonumber\\
=& \Big|\dts_{\Th_{m-1}}(\Fh)\cdot
(S_{T}-S_{\tau_{m-1}})-\dts_{\Th_{m}}(\Fh)\cdot
(\Fh_{\Th_{m}}-\Fh_{\Th_{m-1}})\Big|\le \frac{N(d+K)}{2^N}.\label{eq:integral_approximation}
\end{align}
Hence, 
\begin{align}
x+\int_{0}^{T}\ts_u^{(N)}(S)\cdot dS_u\ge& x+\int_{0}^{T} \dts_u(\Fh)\cdot d\Fh_u-\frac{5(d+K)N}{2^N}-\frac{(d+K)N}{2^N} \label{eq:integral_super-replication_1} \\
\ge& \GH(\Fh)- \alpha\wedge (\beta\sqrt{\hat{m}^{(D-2)}(\Fh)})-N\lambda_{\info}(\Fh)-\frac{6(d+K)N}{2^N} \label{eq:integral_super-replication_2}\\
\ge& \GH(\Fh)-\alpha\wedge (\beta\sqrt{\hat{m}^{(D-2)}(\Fh)})-N/2^{N-3}-\frac{6(d+K)N}{2^N} \label{eq:integral_super-replication_3}\\
\ge& G(S)-\alpha\wedge (\beta\sqrt{m^{(D)}(S)})-f_e(2^{-N+4})-\frac{14(d+K)N}{2^N}=G^{(N)}(S)\nonumber
\end{align}
where the inequality between \eqref{eq:integral_super-replication_1} and \eqref{eq:integral_super-replication_2} follows from the super-replicating property of $\dts$ and the fact that $\PN(f)>0$, $\forall f\in\DDO$, the inequality between \eqref{eq:integral_super-replication_2} and \eqref{eq:integral_super-replication_3} is justified by \eqref{eq:error_estimation}and the last inequality is given by \eqref{eq:error_estimation_new} and \ref{eq: number of stopping times}.

\textsl{Admissibility}:\\
Now, for a given $t<T$, let $k< m$ be the largest integer so that $\tau_{k}(S)\le t$. It follows from \eqref{eq:integral_mimic_1} and \eqref{eq:integral_approximation} that
\begin{align}
\int_{0}^{t}\ts^{(N)}_u(S)\cdot dS_u 
=&\int_{0}^{\tau_{k}}\ts_u^{(N)}(S)\cdot dS_u+\int_{\tau_k}^{t}\ts_u^{(N)}(S)\cdot dS_u \nonumber\\
\ge& \int_{0}^{\Th_{k}} \dts_u(\Fh)\cdot d\Fh_u - \frac{5(d+K)N}{2^{N}} - N(d+K)\max_{i}|S^{(i)}_{t}-S^{(i)}_{\tau_{k}}| \label{eq:integral_admissiblity_1}\\
\ge& -M-\frac{6(d+K)N}{2^N}. \label{eq:integral_admissibility_2}
\end{align}
where the inequality between \eqref{eq:integral_admissiblity_1} and \eqref{eq:integral_admissibility_2} follows from the admissibility of $\dts$ and the fact that $\PN(f)>0, \forall f\in \DDO$. Hence, $\pi^{(N)}$ is admissible.

\end{proof}

\subsection{Construction of $\sigma$'s and $Y$ in Theorem~\ref{thm:second last approximation}}\label{appendix:construction}

For an integer $m_0$ and given $x_1,\cdots,x_{m_0}$, introduce the notation 
\begin{equation*}
\vec{x}_m:=(x_1,\cdots,x_{m_0}).
\end{equation*}
Set 
\begin{align*}
&\T:=\Q_+\cup \{a\ge 0\,:\, a=T-b \text{ for some }  b\in\Q_+\}=\{t_l\}_{l}^{\infty},\\
&\mathcal{S}_k:=\{\frac{1}{2^{N+k}}(a_1,\ldots,a_{d+K})\,:\, a_j\in\Z,\, |a_j|\le 2^{k},\, j=1,\ldots,d+K\}, 
\end{align*}
where $(t_l)_{l\ge 1}$ is a decreasing sequence of strictly positive numbers $t_k\searrow 0$ with $t_1 = T$.  
For $k=1,\cdots,m_0$, define the functions $\Psi_k, \Phi_k:\T^{k}\times \mathcal{S}_1\times \ldots\times \mathcal{S}_{k-1}\to [0,1]$ by
\begin{equation}\label{ap:6}
\Psi_k(\vec{\alpha}_k;\vec{\beta}_{k-1}):=\hat{\Q}(\Th_k-\Th_{k-1}\ge\alpha_k|B),
\end{equation}
where
\begin{equation*}
B:=\big\{\Th_i-\Th_{i-1}=\alpha_i,\SH_{\Th_i}-\SH_{\Th_{i-1}}=\beta_i, i\le k-1\big\},
\end{equation*}
and 
\begin{equation}\label{ap:7}
\Phi_k(\vec{\alpha}_k;\vec{\beta}_{k-1};\beta):=\hat{\Q}(\SH_{\Th_k}-\SH_{\Th_{k-1}}=\beta|C), \quad\beta\in \mathcal{S}_k, 
\end{equation}
where
\begin{equation*}
C:=\big\{\Th_k\le T, \Th_j-\Th_{j-1}=\alpha_j,\SH_{\Th_i}-\SH_{\Th_{i-1}}=\beta_i, j\le k, i\le k-1\big\}.
\end{equation*}
As usual we set $\hat{\Q}(\cdot|\emptyset)\equiv 0$. Next, for $k\le m_0$, we define the maps $\Upsilon_k:\T^{k}\times\mathcal{S}_1\times \ldots\times \mathcal{S}_{k-1}\to [-\infty,\infty]$ and $\Theta_k:\T^{k}\times\mathcal{S}_1\times \ldots\times \mathcal{S}_{k-1}\to [-\infty,\infty]$, as the unique solutions of the following equations,
\begin{equation}\label{ap:4}
P^{W}\Big(W^{(1)}_{\alpha_{k}}<\Upsilon_k(\vec{\alpha}_k;\vec{\beta}_{k-1}; s_{k,l})\Big)=\sum_{j=1}^{l}
\Phi_k(\vec{\alpha}_{k};\vec{\beta}_{k-1}; s_{k,j})
\end{equation}
where $\{s_{k,1}, s_{k,2},\ldots, s_{k,l},\ldots\}$ is an enumeration of $\mathcal{S}_{k}$,  and 
\begin{equation}\label{ap:3}
P^{W}\Big(W^{(1)}_{t_{l}}-W^{(1)}_{t_{l+1}}<\Theta_k(\vec{\alpha}_k;\vec{\beta}_{k-1})\Big)=\frac{
\Psi_k(\vec{\alpha}_{k-1},t_l;\vec{\beta}_{k-1})}{\Psi_k(\vec{\alpha}_{k-1},t_{l+1};\vec{\beta}_{k-1})},
\end{equation}
where $l\in \N$ is given by $\alpha_k=t_l\in \T$. From the definitions it follows that $\Psi_k(\vec{\alpha}_{k-1},t_l;\vec{\beta}_{k-1})\le \Psi_k(\vec{\alpha}_{k-1},t_{l+1};\vec{\beta}_{k-1})$. Thus if $\Psi_k(\vec{\alpha}_{k-1},t_{l+1};\vec{\beta}_{k-1})=0$ for some $l$, then $\Psi_k(\vec{\alpha}_{k-1},t_l;\vec{\beta}_{k-1})=0$. We set $0/0\equiv 0$.\\
Set $\sigma_0\equiv 0$ and define the random variables $\sigma_1,\ldots, \sigma_{m_0}, Y_1,\ldots,Y_{m_0}$ by the following recursive relations
\begin{align*}
&\sigma_1=\sum_{k=1}^{\infty}t_k\indicator{W^{(1)}_{t_k}-W^{(1)}_{t_{k+1}}>\Theta_1(t_k)}\prod_{j=k+1}^{\infty}\indicator{W^{(1)}_{t_j}-W^{(1)}_{t_{j+1}}<\Theta_1(t_j)},\\
&Y_1=\sum_{j=1}^{\infty}s_{1,j}\indicator{\Upsilon_1(\sigma_1;s_{1,j-1})\le W^{(2)}_{\sigma_{1}}<\Upsilon_1(\sigma_1;s_{1,j})},
\end{align*}
and for $i\ge 1$
\begin{eqnarray}
&&\sigma_i=\sigma_{i-1}+\Delta_i\label{ap:1}\\
&&Y_i=\indicator{\sigma_i< T}\sum_{j=1}^{\infty}s_{i,j}\indicator{\Upsilon_i(\vec{\Delta}_{i};\vec{Y}_{i-1};s_{i,j-1})\le W^{(i+1)}_{\sigma_{i}}-W^{(i+1)}_{\sigma_{i-1}}<\Upsilon_i(\vec{\Delta}_{i};\vec{Y}_{i-1};s_{i,j})},\nonumber
\end{eqnarray}
where $\Delta_i=t_k$ on the set $B_i\cap C_{i,k}\cap D_{i,k}$ and zero otherwise. These sets are given by,
\begin{eqnarray*}
&&B_i:=\{\sigma_i\le T\}\\
&&C_{i,k}:=\{W^{(1)}_{t_k+\sigma_{i-1}}-W^{(1)}_{t_{k+1}+\sigma_{i-1}}>\Theta_i(\vec{\Delta}_{i-1},t_k;\vec{Y}_{i-1})\}\\
&&D_{i,k}:=\bigcap_{j=k+1}^{\infty}\{W^{(1)}_{t_j+\sigma_{i-1}}-W^{(1)}_{t_{j+1}+\sigma_{i-1}}<\Theta_i(\vec{\Delta}_{i-1},t_j;\vec{Y}_{i-1})\}.
\end{eqnarray*}
Since $t_k$ is decreasing with $t_1=T$, $\sigma_1\le \sigma_2\le\cdots\le \sigma_{m_0}$ and they are stopping times with respect to the Brownian filtration. Let $k\le m_0$ and $(\vec{\alpha}_k;\vec{\beta}_{k-1})\in \T^{k}\times\mathcal{S}_1\times \ldots\times \mathcal{S}_{k-1}$. There exists $l\in N$ such that $\alpha_k=t_l\in\T$. From $\eqref{ap:1}-\eqref{ap:2}$, the strong Markov property and the independence of the Brownian motion increments it follows that
\begin{alignat}{1}
P^{W}(\sigma_k-&\sigma_{k-1}\ge \alpha_k|(\vec{\Delta}_{k-1},\vec{Y}_{k-1})=(\vec{\alpha}_{k-1}, \vec{\beta}_{k-1})) \nonumber\\
=&P^{W}\Big(\bigcap_{j=k+1}^{\infty}\{W^{(1)}_{t_j+\sigma_{i-1}}-W^{(1)}_{t_{j+1}+\sigma_{i-1}}<\Theta_i(\vec{\Delta}_{i-1},t_j;\vec{Y}_{i-1})\}\Big) \nonumber\\
=&\prod_{j=m}^{\infty}P^{W}\Big(W^{(1)}_{t_j+\sigma_{i-1}}-W^{(1)}_{t_{j+1}+\sigma_{i-1}}<\Theta_i(\vec{\Delta}_{i-1},t_j;\vec{Y}_{i-1})\Big)\nonumber \\
=&\Psi(\vec{\alpha}_k;\vec{\beta}_{k-1}),\label{ap:2}
\end{alignat}
where the last equality follows from $\eqref{ap:3}$ and the fact that
\begin{equation*}
\lim_{l\to\infty}\Psi(\vec{\alpha}_{k-1},t_l;\vec{\beta}_{k-1})=1.
\end{equation*}
Similarly, from $\eqref{ap:4}$ and $\eqref{ap:2}$, we have
\begin{alignat}{1}
P^{W}\Big(Y_k=\beta|\sigma_k&<T,\vec{\Delta}_k=\vec{\sigma}_k, \vec{Y}_{k-1}=\vec{\beta}_{k-1}\Big) \nonumber\\
&=P^{W}\Big(W^{(k+1)}_{\sum_{i=1}^k\alpha_i}-W^{(k+1)}_{\sum_{i=1}^{k-1}\alpha_i}<\Upsilon_k(\vec{\alpha}_k;\vec{\beta}_{k-1};\beta)\Big) \nonumber\\
&=\Phi(\vec{\alpha}_k;\vec{\beta}_{k-1};\beta), \qquad \forall\beta\in \mathcal{S}_k,.\label{ap:5}
\end{alignat}
Using $\eqref{ap:6}$-$\eqref{ap:7}$ and $\eqref{ap:2}$-$\eqref{ap:5}$, we conclude that 
\begin{eqnarray*}
\LL_{P^W}\big((\vec{\sigma}_{m_0}; \vec{Y}_{m_0})\big)=\LL_{\hat{\Q}}\big((\vec{\Th}_{m_0};
\vec{\Delta\SH}_{m_0})\big),
\end{eqnarray*}
where $\Delta\SH_k=\SH_{\Th_k}-\SH_{\Th_{k-1}}$, $k\le m_0$.

\bibliographystyle{agsm}
\bibliography{bib2} 

\end{document}